\documentclass{amsart}
 \usepackage{        amsthm, amssymb, amsfonts,             }
\normalsize
\theoremstyle{definition}
                                                                                
\newtheorem{Thm}{Theorem}

\newtheorem{ccorollary}[Thm]{Corollary}
\newtheorem{ddefinition}[Thm]{Definition}
\newtheorem{llemma}[Thm]{Lemma}
\newtheorem{ttheorem}[Thm]{Theorem}

\newtheorem{eexample}[Thm]{Example}
 \newcommand {\Ab}{\mathbf{A}}  
\newcommand {\Bb}{\mathbf{B}}
\newcommand {\Cb}{\mathbb{C}}
\newcommand {\Cbf}{\mathbf{C}}

\newcommand {\op}{\mathbf{\oplus}} 
\newcommand {\om}{\mathbf{\ominus}} 
\newcommand {\opp}{\op^\prime}
\newcommand {\omp}{\om^\prime}

\newcommand {\asubM}{\|\ab\|_{\lower.1ex \hbox {\scriptsize {M}}}}
\newcommand {\hsubM}{\|\hb\|_{\lower.1ex \hbox {\scriptsize {M}}}}
\newcommand {\hsubMs}{\|\hb\|_{\lower.1ex \hbox {\scriptsize {M}}}^2}
\newcommand {\hasubM}{\|\hb_{\ab}\|_{\lower.1ex \hbox {\scriptsize {M}}}}
\newcommand {\hbsubM}{\|\hb_{\bb}\|_{\lower.1ex \hbox {\scriptsize {M}}}}
\newcommand {\hcsubM}{\|\hb_{\cb}\|_{\lower.1ex \hbox {\scriptsize {M}}}}
\newcommand {\AsubM}{\|\Ab\|_{\lower.1ex \hbox {\scriptsize {M}}}}
\newcommand {\AsubMs}{\|\Ab\|_{\lower.1ex \hbox {\scriptsize {M}}}^2}
\newcommand {\AosubM}{\|A_1\|_{\lower.1ex \hbox {\scriptsize {M}}}}
\newcommand {\AtsubM}{\|A_2\|_{\lower.1ex \hbox {\scriptsize {M}}}}
\newcommand {\BsubM}{\|\Bb\|_{\lower.1ex \hbox {\scriptsize {M}}}}
\newcommand {\CsubM}{\|\Cb\|_{\lower.1ex \hbox {\scriptsize {M}}}}
\newcommand {\CsubMf}{\|\Cbf\|_{\lower.1ex \hbox {\scriptsize {M}}}}
\newcommand {\CsubMs}{\|\Cb\|_{\lower.1ex \hbox {\scriptsize {M}}}^2}
\newcommand {\asupM}{\|\ab\|^{\lower.1ex \hbox {\scriptsize {M}}}}
\newcommand {\aspM}{\ab^{\lower.1ex \hbox {\scriptsize {M}}}}
\newcommand {\asbM}{\ab_{\lower.1ex \hbox {\scriptsize {M}}}}
\newcommand {\AsupM}{\|\Ab\|^{\lower.1ex \hbox {\scriptsize {M}}}}
\newcommand {\BsupM}{\|\Bb\|^{\lower.1ex \hbox {\scriptsize {M}}}}
\newcommand {\CsupM}{\|\Cb\|^{\lower.0ex \hbox {\scriptsize {M}}}}
\newcommand {\CsupMf}{\|\Cbf\|^{\lower.0ex \hbox {\scriptsize {M}}}}
\newcommand {\asubP}{\|\ab\|_{\lower.1ex \hbox {\scriptsize {P}}}}
\newcommand {\asbP}{\ab_{\lower.1ex \hbox {\scriptsize {P}}}}
\newcommand {\AsubP}{\|\Ab\|_{\lower.1ex \hbox {\scriptsize {P}}}}
\newcommand {\BsubP}{\|\Bb\|_{\lower.1ex \hbox {\scriptsize {P}}}}
\newcommand {\CsubP}{\|\Cb\|_{\lower.1ex \hbox {\scriptsize {P}}}}

\newcommand {\phue}{\lower.3ex \hbox {\scriptsize {UE}}}
\newcommand {\pheu}{\lower.3ex \hbox {\scriptsize {EU}}}
\newcommand {\phum}{\lower.3ex \hbox {\scriptsize {UM}}}
\newcommand {\phmu}{\lower.3ex \hbox {\scriptsize {MU}}}
\newcommand {\phme}{\lower.3ex \hbox {\scriptsize {ME}}}
\newcommand {\phem}{\lower.3ex \hbox {\scriptsize {EM}}}

\newcommand {\lowerkluma}{\lower1.5ex \hbox {\phantom{K}}}
\newcommand {\lowerklumb}{\lower3.5ex \hbox {\phantom{K}}}
\newcommand {\lowerklumc}{\lower7.0ex \hbox {\phantom{K}}}
\newcommand {\pp}{\lower.6ex \hbox {\footnotesize {$P_{1} P_{2}$}}}

 \newcommand {\lowAA}{\lower.3ex \hbox {\scriptsize {$\Ab$}}}
 \newcommand {\lowBB}{\lower.3ex \hbox {\scriptsize {$\Bb$}}}
 \newcommand {\lowCC}{\lower.3ex \hbox {\scriptsize {$\Cb$}}}
 \newcommand {\lowA}{\lower.3ex \hbox {\scriptsize {$A$}}}
 \newcommand {\lowB}{\lower.3ex \hbox {\scriptsize {$B$}}}
 \newcommand {\lowC}{\lower.3ex \hbox {\scriptsize {$C$}}}
 \newcommand {\lowE}{\lower.3ex \hbox {\tiny       {$\rm E$}}}
 \newcommand {\lowM}{\lower.3ex \hbox {\tiny       {$\rm M$}}}
 \newcommand {\lowtM}{\lower.01ex \hbox {\tiny       {$\rm M$}}}
 \newcommand {\lowtE}{\lower.01ex \hbox {\tiny       {$\rm E$}}}
 \newcommand {\lowtU}{\lower.01ex \hbox {\tiny       {$\rm U$}}}

 \newcommand {\lowEC}{\lower.3ex \hbox {\scriptsize {$EC$}}}
 \newcommand {\lowER}{\lower.3ex \hbox {\scriptsize {$ER$}}}
 \newcommand {\lowU}{\lower.3ex \hbox {\tiny       {\rm U}}}
 \newcommand {\lowDU}{\lower.3ex \hbox {\scriptsize {\rm DU}}}
 \newcommand {\lowCU}{\lower.3ex \hbox {\scriptsize {\rm CU}}}
 \newcommand {\lowDM}{\lower.3ex \hbox {\scriptsize {\rm DM}}}
 \newcommand {\lowCM}{\lower.3ex \hbox {\scriptsize {\rm CM}}}
 \newcommand {\lowo}{\lower.3ex \hbox {\scriptsize {\rm 0}}}
 \newcommand {\lowf}{\lower.3ex \hbox {\scriptsize {\rm f}}}

\newcommand {\lowmbpa}{\lower.6ex \hbox {\footnotesize {$\ome\bb\ope\ab$}}}
 
\newcommand {\uvc}{\displaystyle\frac{\lower.6ex \hbox {$\ub\ccdot\vb$}}{c^2}}
\newcommand {\uvs}{\displaystyle\frac{\lower.6ex \hbox {$\ub\ccdot\vb$}}{s^2}}
\newcommand {\unpuvc}{ \lower.6ex \hbox {$1 + \uvc$} }
\newcommand {\unpuvs}{ \lower.6ex \hbox {$1 + \uvs$} }
\newcommand {\uvcbar}{\displaystyle\frac{\lower.6ex \hbox
            {$\ubar\ccdot\vb$}}{c^2}}
\newcommand {\unpuvcbar}{ \lower.6ex \hbox {$1 + \uvcbar$} }
\newcommand {\vwc}{\displaystyle\frac{\lower.6ex\hbox{$\vb\ccdot\wb$}}{c^2}}
\newcommand {\vws}{\displaystyle\frac{\lower.6ex\hbox{$\vb\ccdot\wb$}}{s^2}}
\newcommand {\unpvwc}{ \lower.6ex \hbox {$1 + \vwc$} }
\newcommand {\unpvws}{ \lower.6ex \hbox {$1 + \vws$} }
\newcommand {\subE}{\!\lower.1ex \hbox {\tiny E}}
\newcommand {\subG}{\!\lower.1ex \hbox {\tiny G}}
\newcommand {\subH}{\!\lower.1ex \hbox {\tiny H}}
\newcommand {\subEs}{\!\lower.1ex \hbox {\tiny {E,S}}}
\newcommand {\subEt}{\!\lower.1ex \hbox {\tiny {E,2}}}
\newcommand {\subEC}{\!\lower.1ex \hbox {\tiny EC}}
\newcommand {\subU}{\!\lower.1ex \hbox {\tiny U}}
\newcommand {\subM}{\!\lower.1ex \hbox {\tiny M}}
\newcommand {\subC}{\!\lower.1ex \hbox {\tiny C}}
\newcommand {\subbE}{\!\lower.01ex \hbox {\tiny E}}
\newcommand {\subbU}{\!\lower.01ex \hbox {\tiny U}}
\newcommand {\subbC}{\!\lower.01ex \hbox {\tiny C}}
\newcommand {\subbM}{\!\lower.01ex \hbox {\tiny M}}
\newcommand {\subbG}{\!\lower.01ex \hbox {\tiny G}}
\newcommand {\subbH}{\!\lower.01ex \hbox {\tiny H}}
\newcommand {\ope}{\op_{_{\subE}}\!\,}

\newcommand {\ome}{\om_{_{\subE}}\!\,}

\newcommand {\ccdot}{\mathbf{\cdot }} 

\newcommand {\ab}{\mathbf{a}}
\newcommand {\bb}{\mathbf{b}}
\newcommand {\cb}{\mathbf{c}}

\newcommand {\hb}{\mathbf{h}}

\newcommand {\ub}{\mathbf{u}}
\newcommand {\vb}{\mathbf{v}}

\newcommand {\wb}{\mathbf{w}}

\newcommand {\xb}{\mathbf{x}}
\newcommand {\yb}{\mathbf{y}}

\newcommand {\zerb}{\mathbf{0}}

\newcommand {\omegab}{\pmb{\omega}}
\newcommand {\ubar}{\bar{\ub}}

\newcommand {\Nb}{\mathbb{N}}

\newcommand {\Rb}{\mathbb{R}}

\newcommand {\Rmm}{\Rb^{m\times m}}
\newcommand {\Rnn}{\Rb^{n\times n}}
\newcommand {\Rmn}{\Rb^{m\times n}}
\newcommand {\Rnm}{\Rb^{n\times m}}
\newcommand {\Rmcn}{\Rb^{m,n}}

\newcommand {\Rn}{\Rb^n}

\newcommand {\gyr}{{\rm gyr}}
\newcommand {\rgyr}{{\rm rgyr}}
\newcommand {\rgyrab}{\rgyr[P_1,P_2]}
\newcommand {\rgyrba}{\rgyr[P_2,P_1]}
\newcommand {\lgyr}{{\rm lgyr}}
\newcommand {\lgyrab}{\lgyr[P_1,P_2]}

\newcommand {\gyrabp}{\gyr[P_1,P_2]}
\newcommand {\gyrbap}{\gyr[P_2,P_1]}

\newcommand {\dett}{{\textstyle det}}
\newcommand {\Aut}{{\rm Aut}}
\newcommand {\Auto}{{\Aut_0}}

\newcommand {\gyrab}{\gyr[a,b]}

\newcommand {\timess}{\!\times\!}

\newcommand {\half}{\textstyle\frac{1}{2}}

\newcommand {\inn}{\hspace{-0.1cm}\in\hspace{-0.1cm}}

\newcommand {\subEA}{\!\lower.1ex \hbox {\tiny EA}}

 \newcommand {\hA}{\hat{A}}
 \newcommand {\hB}{\hat{B}}
 \newcommand {\hO}{\hat{O}}
 \newcommand {\hP}{\hat{P}}
 \newcommand {\hS}{\hat{S}}
 \newcommand {\lab}{L_{P_1,P_2}}
 \newcommand {\lba}{L_{P_2,P_1}}
 \newcommand {\rab}{R_{P_1,P_2}}
 \newcommand {\rba}{R_{P_2,P_1}}

\begin{document}

\pagenumbering{arabic}
\begin{center}
\huge{
Parametric Realization of the \\
Lorentz Transformation Group in \\
Pseudo-Euclidean Spaces
     }
\end{center}
\begin{center}
Abraham A. Ungar\\
Department of Mathematics\\
North Dakota State University\\
Fargo, ND 58105, USA\\
Email: {\tt Abraham.Ungar@ndsu.edu}
\end{center}

\begin{quotation}
{\bf ABSTRACT}\phantom{OO}
The Lorentz transformation group $SO(m,n)$, $m,n\in\Nb$,
is a group of Lorentz transformations of order $(m,n)$, that is,
a group of special linear transformations
in a pseudo-Euclidean space $\Rb^{m,n}$ of signature $(m,n)$
that leave the pseudo-Euclidean inner product invariant.
A parametrization of $SO(m,n)$ is presented, giving rise to the
composition law of Lorentz transformations of order $(m,n)$ in terms of
parameter composition.
The parameter composition, in turn, gives rise to a novel group-like
structure that underlies $\Rmcn$, called a bi-gyrogroup. Bi-gyrogroups
form a natural generalization of gyrogroups where the latter
form a natural generalization of groups.
Like the abstract gyrogroup, the abstract bi-gyrogroup can play
a universal computational role which extends far beyond the domain of
pseudo-Euclidean spaces.
\end{quotation}

\section{Introduction} \label{secint}

A pseudo-Euclidean space $\Rmcn$ of signature $(m,n)$, $m,n\in\Nb$, is an
$(m+n)$-dimensional space with the pseudo-Euclidean inner product of
signature $(m,n)$. A Lorentz transformation of order $(m,n)$ is
a special linear transformation $\Lambda\in SO(m,n)$ in $\Rmcn$
that leaves the
pseudo-Euclidean inner product invariant. It is special in the sense that
the determinant of the $(m+n)\times (m+n)$ real matrix $\Lambda$ is 1,
and the determinant of its first $m$ rows and columns is positive
\cite[p.~478]{hamermesh62}.
The group of all Lorentz transformations of order $(m,n)$ is also known
as the special pseudo-orthogonal group, denoted by $SO(m,n)$.

A Lorentz transformation without rotations is called a boost. In general,
two successive boosts are not equivalent to a boost. Rather, they are
equivalent to a boost associated with two rotations, called a left rotation
and a right rotation, or collectively, a {\it bi-rotation}.
The two rotations of a bi-rotation are nontrivial if both $m>1$ and $n>1$.
The special case when $m=1$ and $n>1$ was studied in 1988
in \cite{parametrization}.
The study in \cite{parametrization} resulted in the discovery of two
novel algebraic structures that became known as a gyrogroup and
a gyrovector space. Subsequent study of gyrovector spaces reveals in
\cite{mybook01,mybook02,mybook03,mybook04,mybook05,mybook06,mybook07}
that gyrovector spaces form the algebraic setting for hyperbolic geometry,
just as vector spaces form the algebraic setting for Euclidean geometry.
The aim of this paper is to extend the study
of the parametric realization of the Lorentz group
in \cite{parametrization} from $m=1$ to $m\ge1$, and to reveal the
resulting new algebraic structure, called a {\it bi-gyrogroup}.

In order to emphasize that when $m>1$ and $n>1$ a successive application
of two boosts generates a bi-rotation, a Lorentz boost of order $(m,n)$,
$m,n>1$, is called a {\it bi-boost}.
The composition law of two bi-boosts gives rise in this article to a
{\it bi-gyrocommutative bi-gyrogroup} operation, just as the
composition law of two boosts gives rise in \cite{parametrization} to a
gyrocommutative gyrogroup operation, as demonstrated in \cite{mybook01}.
Accordingly, a bi-gyrogroup of order $(m,n)$, $m,n\in\Nb$, is a
group-like structure that specializes to a gyrogroup when
either $m=1$ or $n=1$.

We show in Theorem \ref{dokrn1}
that a Lorentz transformation $\Lambda$ of order $(m,n)$ possesses
the unique parametrization $\Lambda=\Lambda(P,O_n,O_m)$, where
\begin{enumerate}
\item\label{itb01}
$P\in\Rnm$ is any real $n\times m$ matrix; where
\item\label{itb02}
$O_n\in SO(n)$ is any $n\times n$ special orthogonal matrix, taking
$P$ into $O_nP$; and, similarly, where
\item\label{itb03}
$O_m\in SO(m)$ is any $m\times m$ special orthogonal matrix, taking
$P$ into $PO_m$.
\end{enumerate}

In the special case when $m=1$, the Lorentz transformation
of order $(m,n)$
specializes to the Lorentz transformation of
special relativity theory ($n=3$ in physical applications), where
the parameter $P$
is a vector that represents relativistic proper velocities.

The parametrization of the Lorentz transformation $\Lambda$
enables in Theorem \ref{dokrt} the Lorentz transformation
composition (or, product) law to be expressed in terms of
parameter composition.
Under the resulting parameter composition, the parameter $O_n$ of
$\Lambda$, called a {\it left rotation} (of $P\in\Rnm$), forms a group.
The group that the left rotations form
is the special orthogonal group $SO(n)$. Similarly,
under the parameter composition, the parameter $O_m$ of
$\Lambda$, called a {\it right rotation} (of $P\in\Rnm$), forms a group.
The group that the right rotations form
is the special orthogonal group $SO(m)$.
The pair $(O_n,O_m)\in SO(n)\times SO(m)$
is called a bi-rotation, taking $P\in\Rnm$ into $O_n PO_m\in\Rnm$.

Contrasting the left and right rotation parameters, the parameter $P$
does not form a group under parameter composition. Rather, it forms
a novel algebraic structure, called
a {\it bi-gyrocommutative bi-gyrogroup}, defined in
Def.~\ref{defgyrocomy}.
A bi-gyrocommutative bi-gyrogroup is a group-like structure that generalizes
the gyrocommutative gyrogroup structure. The latter, in turn,
is a group-like
structure that forms a natural generalization of the commutative group.

The concept of the gyrogroup emerged from the 1988 study of the
parametrization of the Lorentz group in \cite{parametrization}.
Presently, the gyrogroup concept plays
a universal computational role, which extends far beyond the
domain of special relativity, as noted by
Chatelin in \cite[p.~523]{chatelin12a} and in references therein
and as evidenced, for instance, from
\cite{demirel13,ferreira09,ferreira14,ferreira11,ferreira13,solarin14,
suksumran15b,suksumran15a,suksumran15c}
and
\cite{kasparian04,ungardem05,ungarhyp13,ungar15}.
In a similar way, the concept of the bi-gyrogroup emerges in this paper
from the study of the parametrization of the
Lorentz group $SO(m,n)$, $m,n\in\Nb$.
Hence, like gyrogroups, bi-gyrogroups are capable of playing
a universal computational role that extends far beyond the
domain of Lorentz transformations in pseudo-Euclidean spaces.

\newpage

\section{Lorentz Transformations of Order $(m,n)$}
\label{secc02}

Let $\Rmcn$ be an arbitrary $(m+n)$-dimensional pseudo-Euclidean space
of a signature $(m,n)$, $m,n\in\Nb$, with an
orthonormal basis $e_i$, $i=1,\ldots,m+n$,
\begin{equation} \label{anim01}
e_i\ccdot e_j = \epsilon_i \delta_{ij}
\end{equation}
where
\begin{equation} \label{anim02}
\epsilon_i =
\begin{cases}
+1, & i=1,\ldots,m \\
-1, & i=m+1,\ldots,m+n \,.
\end{cases}
\end{equation}

The inner product $\xb\ccdot\yb$ of two vectors $\xb,\yb\in\Rmcn$,
\begin{equation} \label{anim03}
\begin{split}
\xb &= \sum_{i=1}^{m+n} x_i e_i
\\
\yb &= \sum_{i=1}^{m+n} y_i e_i
\,,
\end{split}
\end{equation}
is
\begin{equation} \label{anim04}
\xb\ccdot\yb =
\sum_{i=1}^{m+n} \epsilon_i x_i y_i
=
\sum_{i=1}^{m} x_i y_i - \sum_{i=m+1}^{m+n} x_i y_i
\,.
\end{equation}

Let $I_m$ be the $m\times m$ identity matrix, and let $\eta$ be the
$(m+n)\times(m+n)$ diagonal matrix
\begin{equation} \label{anim05}
\eta = \begin{pmatrix} I_m & 0_{m,n} \\[4pt] 0_{n,m} & -I_n \end{pmatrix}
\end{equation}
where $0_{m,n}$ is the $m\times n$ zero matrix.
Then, the matrix representation of the inner product \eqref{anim04} is
\begin{equation} \label{anim06}
\xb\ccdot\yb = \xb^t\eta\yb
\end{equation}
where $\xb$ and $\yb$ are the column vectors
\begin{equation} \label{anim07}
\xb=\begin{pmatrix} x_1 \\ x_2 \\ \vdots \\ x_{m+n} \end{pmatrix}
\hspace{0.8cm} {\rm and} \hspace{0.8cm}
\yb=\begin{pmatrix} y_1 \\ y_2 \\ \vdots \\ y_{m+n} \end{pmatrix}
\end{equation}
and exponent $t$ denotes transposition.

Let $\Lambda$ be an $(m+n)\times(m+n)$ matrix that leaves the
inner product \eqref{anim06} invariant. Then, for all $\xb,\yb\in\Rmcn$,
\begin{equation} \label{anim08}
(\Lambda\xb)^t\eta\Lambda\yb = \xb^t\eta\yb
\,,
\end{equation}
implying
$\xb^t\Lambda^t\eta\Lambda\yb=\xb^t\eta\yb$, so that
\begin{equation} \label{anim09}
\Lambda^t\eta\Lambda = \eta
\,.
\end{equation}
The determinant of the matrix equation \eqref{anim09} yields
\begin{equation} \label{anim10}
(\dett\Lambda)^2=1
\,,
\end{equation}
noting that
$\dett(\Lambda^t\eta\Lambda)=(\dett\Lambda^t)(\dett\eta)(\dett\Lambda)$
and $\dett\Lambda^t=\dett\Lambda$.
Hence,
\begin{equation} \label{anim11}
\dett\Lambda = \pm1
\,.
\end{equation}

The special transformations $\Lambda$ that can be reached continuously
from the identity transformation in $\Rmcn$
constitute the special pseudo-orthogonal group $SO(m,n)$, also known as the
{\it (generalized) Lorentz transformation group} of order $(m,n)$.
Each element $\Lambda$ of $SO(m,n)$ is a
Lorentz transformation of order $(m,n)$. It has determinant 1,
\begin{equation} \label{anim12}
\dett\Lambda = 1
\,,
\end{equation}
and the determinant of its first $m$ rows and columns is positive
\cite[p.~478]{hamermesh62}.
The Lorentz transformation of order $(1,3)$ turns out to be
the common
homogeneous, proper, orthochronous Lorentz transformation
of Einstein's special theory of relativity \cite{mybook03}.

Let $\Rmn$ be the set of all $m\timess n$ real matrices.
Following Norbert Dragon \cite{dragon12}, in order to parametrize the
special pseudo-orthogonal group $SO(m,n)$, we partition each
$(m+n)\timess(m+n)$ matrix $\Lambda\in SO(m,n)$
into four blocks consisting of
the submatrices
(i) $A\in \Rmm$,
(ii) $\hA\in \Rnn$,
(iii) $B\in \Rnm$, and
(iv) $\hB\in \Rmn$,
so that
\begin{equation} \label{anim13}
\Lambda = \begin{pmatrix} A&\hat{B} \\ B & \hat{A} \end{pmatrix}
\,.
\end{equation}

By means of \eqref{anim13} and \eqref{anim05}, the matrix equation
\eqref{anim09} takes the form
\begin{equation} \label{anim14}
\begin{pmatrix} A^t & B^t \\ \hat{B}^t & \hat{A}^t \end{pmatrix}
\begin{pmatrix} I_m & 0_{m,n} \\[4pt] 0_{n,m} & -I_n \end{pmatrix}
\begin{pmatrix} A&\hat{B} \\ B & \hat{A} \end{pmatrix}
=
\begin{pmatrix} I_m & 0_{m,n} \\[4pt] 0_{n,m} & -I_n \end{pmatrix}
\end{equation}
or, equivalently,
\begin{equation} \label{anim15}
\begin{pmatrix} A^t A - B^t B & A^t\hB - B^t \hA \\
   \hB^t A - \hA^t B & \hB^t \hB - \hA^t \hA \end{pmatrix}
=
\begin{pmatrix} I_m & 0_{m,n} \\[4pt] 0_{n,m} & -I_n \end{pmatrix}
\,,
\end{equation}
implying
\begin{equation} \label{anim16}
\begin{split}
A^t A &= I_m + B^t B \\
\hA^t \hA &= I_n + \hB^t \hB \\
A^t \hB &= B^t \hA
\,.
\end{split}
\end{equation}

The symmetric matrix $B^tB$ is diagonalizable by an orthogonal matrix
with nonnegative diagonal elements \cite[pp.~171, 396-398, 402]{horn90}.
Hence, the eigenvalues of $I_m+B^tB$ are not smaller than 1, so that
$(\dett A)^2=\dett(I_m+B^tB)\ge1$.
This, in turn, implies that $A$ is invertible.
Similarly,
$(\dett \hA)^2=\dett(I_n+\hB\hB^t)\ge1$, so that $\hA$ is invertible.

An invertible real matrix $A$ can be uniquely decomposed
into the product of an orthogonal matrix
$O\in SO(m)$, $O^t=O^{-1}$, and  a
positive-definite symmetric matrix $S$, $S^t=S$,
with positive eigenvalues \cite[p.~286]{gantmacher59},
\begin{equation} \label{anim17}
A=OS \,.
\end{equation}
Following \eqref{anim17} we have
\begin{equation} \label{anim18}
A^tA = S^t O^t O S = S^2
\end{equation}
with positive eigenvalues $\lambda_i>0$, $i=1,\ldots,m$. Hence,
\begin{equation} \label{anim19}
S=\sqrt{A^tA}
\end{equation}
has the positive eigenvalues $\sqrt{\lambda_i}$
and the same eigenvectors as $S^2$.

The matrix $S$ given by \eqref{anim19} satisfies \eqref{anim17} since
$AS^{-1}$ is orthogonal, as it should be, by \eqref{anim17}.
Indeed,
\begin{equation} \label{anim20}
(AS^{-1})^t AS^{-1} = (S^{-1})^t A^tAS^{-1}=S^{-1}S^2S^{-1} = I_m
\end{equation}
Similarly, $\hA$ is invertible and possesses the decomposition
\begin{equation} \label{anim21}
\hA = \hO\hS
\,,
\end{equation}
where $\hO\in SO(n)$ is an orthogonal matrix and $\hS$ is a
positive-definite symmetric matrix.

By means of \eqref{anim17} and \eqref{anim21}, the the block matrix
\eqref{anim13} possesses the decomposition
\begin{equation} \label{anim22}
\Lambda =
\begin{pmatrix} O & 0_{m,n} \\ 0_{n,m} & \hO \end{pmatrix}
\begin{pmatrix} S & \hP \\ P & \hS \end{pmatrix}
\,,
\end{equation}
where the submatrices $P$ and $\hP$ are to be determined in
\eqref{anim24} below.

Following \eqref{anim22} and \eqref{anim13}, along with
\eqref{anim17} and \eqref{anim21}, we have
\begin{equation} \label{anim23}
\Lambda =
\begin{pmatrix} OS & O\hP \\ \hO P & \hO \hS \end{pmatrix}
=
\begin{pmatrix} A & \hB \\ B & \hA \end{pmatrix}
\,,
\end{equation}
so that $\hO P=B$ and $O\hP=\hB$, that is
\begin{equation} \label{anim24}
\begin{split}
P &= \hO^{-1} B
\\
\hP &= O^{-1} \hB
\,.
\end{split}
\end{equation}

In \eqref{anim22}, $S$ and $\hS$ are invertible symmetric matrices,
and $O$ and $\hO$ are orthogonal matrices with determinant 1.

By means of \eqref{anim09} and \eqref{anim22} we have the
matrix equation
\begin{equation} \label{anim25}
\begin{split}
\begin{pmatrix} S^t & P^t \\ \hP^t & \hS^t \end{pmatrix} &
\begin{pmatrix} O^t & 0_{m,n} \\ 0_{n,m} & \hO^t \end{pmatrix}
\begin{pmatrix} I_m & 0_{m,n} \\ 0_{n,m} & -I_n \end{pmatrix}
\begin{pmatrix} O & 0_{m,n} \\ 0_{n,m} & \hO \end{pmatrix}
\begin{pmatrix} S & \hP \\ P & \hS \end{pmatrix}
\\[8pt]
&=
\begin{pmatrix} I_m & 0_{m,n} \\ 0_{n,m} & -I_n \end{pmatrix}
\,.
\end{split}
\end{equation}

Noting that
\begin{equation} \label{anim26}
\begin{pmatrix} O^t & 0_{m,n} \\ 0_{n,m} & \hO^t \end{pmatrix}
\begin{pmatrix} I_m & 0_{m,n} \\ 0_{n,m} & -I_n \end{pmatrix}
\begin{pmatrix} O & 0_{m,n} \\ 0_{n,m} & \hO \end{pmatrix}
=
\begin{pmatrix} I_m & 0_{m,n} \\ 0_{n,m} & -I_n \end{pmatrix}
\,,
\end{equation}
the matrix equation \eqref{anim25} yields
\begin{equation} \label{anim27}
\begin{pmatrix} S^t & P^t \\ \hP^t & \hS^t \end{pmatrix}
\begin{pmatrix} I_m & 0_{m,n} \\ 0_{n,m} & -I_n \end{pmatrix}
\begin{pmatrix} S & \hP \\ P & \hS \end{pmatrix}
=
\begin{pmatrix} I_m & 0_{m,n} \\ 0_{n,m} & -I_n \end{pmatrix}
\,,
\end{equation}
so that
\begin{equation} \label{anim28}
\begin{pmatrix} S^t & -P^t \\ \hP^t & -\hS^t \end{pmatrix}
\begin{pmatrix} S & \hP \\ P & \hS \end{pmatrix}
=
\begin{pmatrix} I_m & 0_{m,n} \\ 0_{n,m} & -I_n \end{pmatrix}
\end{equation}
and hence
\begin{equation} \label{anim29}
\begin{pmatrix} S^tS-P^tP & S^t\hP-P^t\hS \\
                \hP^tS-\hS^tP & \hP^t\hP-\hS^t\hS \end{pmatrix}
=
\begin{pmatrix} I_m & 0_{m,n} \\ 0_{n,m} & -I_n \end{pmatrix}
\,.
\end{equation}

Noting that $S^tS=S^2$ and $\hS^t\hS=\hS^2$, \eqref{anim29} yields
the equations
\begin{equation} \label{anim30}
\begin{split}
S^2 &= I_m + P^tP \\
\hS^2 &= I_n + \hP^t\hP \\
S^t\hP &= P^t\hS
\end{split}
\end{equation}

Noting that the matrix $S$ is symmetric,
the third and the first equations in \eqref{anim30} imply
\begin{equation} \label{anim31}
\begin{split}
\hP &= S^{-1} P^t \hS \\
S^{-2} &= ( I_m +P^tP)^{-1}
\,.
\end{split}
\end{equation}

Inserting \eqref{anim31} into the second equation in \eqref{anim30},
\begin{equation} \label{anim32}
\begin{split}
\hS^2 &= I_n+\hP^t\hP \\
&= I_n+(\hS^tP(S^t)^{-1}) S^{-1}P^t\hS \\
&= I_n+\hS P S^{-2} P^t \hS \\
&= I_n+\hS P (I_m+P^tP)^{-1} P^t\hS
\,.
\end{split}
\end{equation}

Multiplying both sides of \eqref{anim32} by $\hS^{-2}$, we have
\begin{equation} \label{anim33}
I_n = \hS^{-2} + P(I_m+P^tP)^{-1} P^t
\,.
\end{equation}

Let $\omegab$ be an eigenvector of the matrix $PP^t$, and let
$\lambda$ be its associated eigenvalue,
\begin{equation} \label{anim34}
PP^t\omegab = \lambda\omegab
\,.
\end{equation}
If $P^t\omegab$ is not zero, $P^t\omegab\ne\zerb$, then it is an
eigenvector of $P^tP$ with the same eigenvalue $\lambda$,
\begin{equation} \label{anim35}
(P^t P)P^t\omegab = \lambda P^t\omegab
\,.
\end{equation}
Adding $P^t\omegab$ to both sides of \eqref{anim35} we have
\begin{equation} \label{anim36}
(I_m+P^tP)P^t\omegab = (1+\lambda)P^t\omegab
\,,
\end{equation}
so that
\begin{equation} \label{anim37}
\frac{1}{1+\lambda} P^t\omegab = (I_m+P^tP)^{-1} P^t\omegab
\,.
\end{equation}
Multiplying both sides of \eqref{anim37} by $P$,
\begin{equation} \label{anim38}
\frac{1}{1+\lambda} PP^t\omegab = P(I_m+P^tP)^{-1} P^t\omegab
\,,
\end{equation}
so that, by means of \eqref{anim34},
\begin{equation} \label{anim39}
P(I_m+P^tP)^{-1} P^t\omegab =  \frac{\lambda}{1+\lambda} \omegab
\end{equation}
for any eigenvector $\omegab$ of $PP^t$ for which
$P^t\omegab\ne0$.
Equation \eqref{anim39} remains valid also when $P^t\omegab=0$ since
in this case $\lambda=0$ by \eqref{anim34}.

By means of \eqref{anim33} and \eqref{anim39} we have
\begin{equation} \label{anim40}
\omegab=\hS^{-2}\omegab+P(I_m+P^tP)^{-1}P^t\omegab
=\hS^{-2}\omegab+\frac{\lambda}{1+\lambda} \omegab
\,,
\end{equation}
so that
\begin{equation} \label{anim41}
\hS^{-2}\omegab =\frac{1}{1+\lambda} \omegab
\end{equation}
and, by \eqref{anim34},
\begin{equation} \label{anim42}
\hS^2\omegab = (1+\lambda)\omegab = I_n\omegab + \lambda\omegab
= I_n\omegab + PP^t\omegab = (I_n+PP^t)\omegab
\end{equation}
for any eigenvector $\omegab$ of $PP^t$.

The eigenvectors $\omegab$ constitute a basis of $\Rn$. Hence, it
follows from \eqref{anim42} that
\begin{equation} \label{anim43}
\hS^2 = I_n+PP^t
\end{equation}
and, hence,
\begin{equation} \label{anim44}
\hS = \sqrt{I_n+PP^t}
\,.
\end{equation}

Following \eqref{anim30}\,--\,\eqref{anim31} and \eqref{anim44} we have
\begin{equation} \label{anim45}
\hP = S^{-1} P^t \hS =
\sqrt{I_m+P^tP}^{~-1} P^t \sqrt{I_n+PP^t}
\,.
\end{equation}
Employing \eqref{anim45} and the eigenvectors $\omegab$ of $PP^t$,
we will show in \eqref{anim50} below that $\hP=P^t$.

As in \eqref{anim34}, $\omegab$ is an eigenvector of the matrix
$PP^t$, $P^t\omegab\ne0$, with its associated eigenvalue $\lambda>0$,
implying \eqref{anim37}. Following \eqref{anim37}, the matrix
$(I_m+P^tP)^{-1}$ possesses an eigenvector $P^t\omegab$ with
its associated eigenvalue $1/(1+\lambda)$.
Hence, the matrix
$\sqrt{I_m+P^tP}^{~-1}$ possesses the same eigenvector $P^t\omegab$
with its associated eigenvalue $1/\sqrt{1+\lambda}$,
\begin{equation} \label{anim46}
\sqrt{I_m+P^tP}^{~-1} P^t \omegab
= \frac{1}{\sqrt{1+\lambda}} P^t\omegab
\,.
\end{equation}

Similarly, the matrix $I_n+PP^t$ satisfies the equation
$(I_n+PP^t)\omegab=(1+\lambda)\omegab$,
so that it possesses an eigenvector $\omegab$ with its associated
eigenvalue $1+\lambda$. Hence, the matrix
$\sqrt{I_n+PP^t}$
possesses the same eigenvector $\omegab$ with its associated eigenvalue
$\sqrt{1+\lambda}$,
\begin{equation} \label{anim47}
\sqrt{I_n+PP^t} \, \omegab = \sqrt{1+\lambda} \, \omegab
\,.
\end{equation}

Hence, by \eqref{anim47} and \eqref{anim46},
\begin{equation} \label{anim48} 
\sqrt{I_m+P^tP}^{~-1} P^t \sqrt{I_n+PP^t} \, \omegab
= \sqrt{1+\lambda} \sqrt{I_m+P^tP}^{~-1} P^t \omegab
= P^t \omegab
\end{equation}
for any eigenvector $\omegab$ of $PP^t$ for which $P^t\omegab\ne0$.
Equation \eqref{anim48} remains valid also for $\omegab$
with $P^t\omegab=0$
since in this case $\lambda=0$ by \eqref{anim34}.

The eigenvectors $\omegab$ constitute a basis of $\Rn$. Hence, it
follows from \eqref{anim48} that
\begin{equation} \label{anim49}
\sqrt{I_m+P^tP}^{~-1} P^t \sqrt{I_n+PP^t} = P^t
\,,
\end{equation}
so that, by \eqref{anim45} and \eqref{anim49},
\begin{equation} \label{anim50}
\hP = P^t
\,.
\end{equation}

Following \eqref{anim50} and the first two equations in \eqref{anim30}
we have
\begin{equation} \label{anim51}
\begin{split}
\hP &= P^t \\
S &= \sqrt{I_m+P^tP} \\
\hS &= \sqrt{I_n+PP^t}
\,.
\end{split}
\end{equation}

Inserting \eqref{anim51} into \eqref{anim22} and denoting
$O$ and $\hO$ by $O_m$ and $O_n$, respectively, we obtain the
$(m+n)\timess(m+n)$ matrix $\Lambda$ parametrized by the three
matrix parameters
(i) $P\in \Rnm$,
(ii) $O_m\in SO(m)$, and
(iii) $O_n\in SO(n)$,
\begin{equation} \label{anim52}
\Lambda =
\begin{pmatrix} O_m & 0_{m,n} \\ 0_{n,m} & O_n \end{pmatrix}
\begin{pmatrix} \sqrt{I_m+P^tP} & P^t \\
P & \sqrt{I_n+PP^t} \end{pmatrix}
\,.
\end{equation}

\begin{llemma}\label{mda1}
The following commuting relations hold for all $P\in\Rnm$,
\begin{equation} \label{avir02}
P \sqrt{I_m+P^tP} = \sqrt{I_n+PP^t} \, P
\end{equation}
\begin{equation} \label{avir01}
P^t \sqrt{I_n+PP^t} = \sqrt{I_m+P^tP} \, P^t
\end{equation}
\begin{equation} \label{avir04}
PP^t \sqrt{I_n+PP^t} = \sqrt{I_n+PP^t} \, PP^t
\end{equation}
\begin{equation} \label{avir03}
P^tP \sqrt{I_m+P^tP} = \sqrt{I_m+P^tP} \, P^tP
\,.
\end{equation}
\end{llemma}
\begin{proof}
The commuting relation \eqref{avir01} follows from
\eqref{anim45}, noting that $\hP=P^t$.
The commuting relation \eqref{avir02} is obtained from \eqref{avir01}
by matrix transposition.
The commuting relation \eqref{avir03} is obtained by successive
applications of \eqref{avir01} and \eqref{avir02}.
Finally, the commuting relation \eqref{avir04} is obtained by successive
applications of \eqref{avir02} and \eqref{avir01}.
\end{proof}

\section{Parametric Representation of $SO(m,n)$}
\label{secc03}

The block orthogonal matrix in \eqref{anim52} can be uniquely
resolved as a commuting
product of two orthogonal block matrices,
\begin{equation} \label{anim53}
\begin{pmatrix} O_m & 0_{m,n} \\ 0_{n,m} & O_n \end{pmatrix}
=
\begin{pmatrix} O_m & 0_{m,n} \\ 0_{n,m} & I_n \end{pmatrix}
\begin{pmatrix} I_m & 0_{m,n} \\ 0_{n,m} & O_n \end{pmatrix}
\,.
\end{equation}
The first and the second orthogonal matrices on the right side of
\eqref{anim53} represent, respectively,
(i) a right rotation $O_m\in SO(m)$ of
$\Rnm$, $O_m:P\mapsto PO_m$; and 
(ii)a left rotation $O_n\in SO(n)$ of $\Rnm$, $O_n:P\mapsto O_nP$.
Hence, the orthogonal matrix on the left side
of \eqref{anim53}, which represents the composition of the rotations
$O_m$ and $O_n$, is said to be a {\it bi-rotation} of the
pseudo-Euclidean space $\Rmcn$ about its origin.

By means of \eqref{anim53}, \eqref{anim52} can be written as
\begin{equation} \label{anim54}
\Lambda =
\begin{pmatrix} O_m & 0_{m,n} \\ 0_{n,m} & I_n \end{pmatrix}
\begin{pmatrix} I_m & 0_{m,n} \\ 0_{n,m} & O_n \end{pmatrix}
\begin{pmatrix} \sqrt{I_m+P^tP} & P^t \\ 
P & \sqrt{I_n+PP^t} \end{pmatrix}
\,. 
\end{equation}

\begin{llemma}\label{mgn1}
The commuting relations
\begin{equation} \label{anim55}
\begin{split}
\sqrt{I_m+P^tP} O_m &= O_m \sqrt{I_m+(PO_m)^t(PO_m)}
\\[4pt]
O_n\sqrt{I_n+PP^t}  &= \sqrt{I_n+(O_nP)(O_nP)^t} O_n
\end{split}
\end{equation}
hold for all
$P\in \Rnm$, $O_m\in SO(m)$ and $O_n\in SO(n)$.
\end{llemma}
\begin{proof}
\begin{equation} \label{anim56}
\begin{split}
I_m + P^tP &= O_mI_mO_m^t+O_mO_m^tP^tPO_mO_m^t
\\ &=
O_m(I_m+O_m^tP^tPO_m)O_m^t
\\ &=
O_m \sqrt{I_m+O_m^tP^tPO_m} \sqrt{I_m+O_m^tP^tPO_m} O_m^t
\\ &=
O_m \sqrt{I_m+O_m^tP^tPO_m} O_m^tO_m \sqrt{I_m+O_m^tP^tPO_m} O_m^t
\\ &=
(O_m \sqrt{I_m+O_m^tP^tPO_m} O_m^t)^2
\,.
\end{split}
\end{equation}
Hence,
\begin{equation} \label{anim57}
\sqrt{I_m+P^tP} = O_m \sqrt{I_m+O_m^tP^tPO_m} O_m^t
\,,
\end{equation}
implying the first matrix identity in \eqref{anim55}.

Similarly,
\begin{equation} \label{anim58}
\begin{split}
I_n + PP^t &=
O_n^tI_nO_n+O_n^tO_nPP^tO_n^tO_n
\\ &=
O_n^t(I_n+O_nPP^tO_n^t)O_n
\\ &=
O_n^t \sqrt{I_n+O_nPP^tO_n^t} \sqrt{I_n+O_nPP^tO_n^t} O_n
\\ &=
O_n^t \sqrt{I_n+O_nPP^tO_n^t} O_nO_n^t \sqrt{I_n+O_nPP^tO_n^t} O_n
\\ &=
(O_n^t \sqrt{I_n+O_nPP^tO_n^t} O_n)^2
\,.
\end{split}
\end{equation}
Hence,
\begin{equation} \label{anim59}
\sqrt{I_n+PP^t} = O_n^t \sqrt{I_n+O_nPP^tO_n^t} O_n
\,,
\end{equation}
implying the second matrix identity in \eqref{anim55}.
\end{proof}

\begin{llemma}\label{mgn1p1}
The commuting relations
\begin{equation} \label{adir01}
\begin{split}
\sqrt{I_m+P^tP}^{~-1} O_m &= O_m \sqrt{I_m+(PO_m)^t(PO_m)}^{~-1}
\\[4pt]
O_n\sqrt{I_n+PP^t}^{~-1}  &= \sqrt{I_n+(O_nP)(O_nP)^t}^{~-1} O_n
\end{split}
\end{equation}
hold for all
$P\in \Rnm$, $O_m\in SO(m)$ and $O_n\in SO(n)$.
\end{llemma}
\begin{proof}
Inverting the first matrix identity in \eqref{anim55}, we have
the matrix identity
\begin{equation} \label{adir02}
O_m^{-1} \sqrt{I_m+P^tP}^{~-1}
=
\sqrt{I_m+(PO_m)^t(PO_m)}^{~-1} O_m^{-1}
\,,
\end{equation}
which implies the first equation in the Lemma.

Similarly,
inverting the second matrix identity in \eqref{anim55}, we obtain
a matrix identity that implies the second equation in the Lemma.
\end{proof}

\begin{llemma}\label{mgn2}
The commuting relation
\begin{equation} \label{anim60}
\begin{split}
\begin{pmatrix} O_m & 0_{m,n} \\[8pt] 0_{n,m} & I_n \end{pmatrix}
&\begin{pmatrix} \sqrt{I_m+(PO_m)^t(PO_m)} & (PO_m)^t \\[8pt]
                PO_m & \sqrt{I_n+(PO_m)(PO_m)^t} \end{pmatrix}
\\[12pt] &=
\begin{pmatrix} \sqrt{I_m+P^tP} & P^t \\[8pt] P & \sqrt{I_n+PP^t}
\end{pmatrix}
\begin{pmatrix} O_m & 0_{m,n} \\[8pt] 0_{n,m} & I_n \end{pmatrix}
\end{split}
\end{equation}
holds for all $P\in \Rnm$ and $O_m\in SO(m)$.
\end{llemma}
\begin{proof}
Let $J_1$ and $J_2$ denote the left side and the right side of
\eqref{anim60}, respectively.
Clearly,
\begin{equation} \label{anim61}
J_2 =
\begin{pmatrix} \sqrt{I_m+P^tP} O_m & P^t \\[8pt]
                PO_m & \sqrt{I_n+PP^t} \end{pmatrix}
\end{equation}
and, by means of the first commuting relation in Lemma \ref{mgn1},
\begin{equation} \label{anim62}
\begin{split}
J_1 &= \begin{pmatrix} O_m\sqrt{I_m+(PO_m)^t(PO_m)} & O_m(PO_m)^t \\[8pt]
           PO_m & \sqrt{I_n+(PO_m)(PO_m)^t}  \end{pmatrix}
\\[12pt] &=
\begin{pmatrix} \sqrt{I_m+P^tP} O_m & P^t \\[8pt]
           PO_m & \sqrt{I_n+PP^t} \end{pmatrix}
\,.
\end{split}
\end{equation}
Hence, $J_1=J_2$, and the proof is complete.
\end{proof}

\begin{llemma}\label{mgn2a}
The commuting relation
\begin{equation} \label{anim60a}
\begin{split}
&\begin{pmatrix} \sqrt{I_m+(O_nP)^t(O_nP)} & (O_nP)^t \\[8pt]
                O_nP & \sqrt{I_n+(O_nP)(O_nP)^t} \end{pmatrix}
\begin{pmatrix} I_m & 0_{m,n} \\[8pt] 0_{n,m} & O_n \end{pmatrix}
\\[12pt] &=
\begin{pmatrix} I_m & 0_{m,n} \\[8pt] 0_{n,m} & O_n \end{pmatrix}
\begin{pmatrix} \sqrt{I_m+P^tP} & P^t \\[8pt] P & \sqrt{I_n+PP^t}
\end{pmatrix}
\end{split}
\end{equation}
holds for all
$P\in \Rnm$ and $O_n\in SO(n)$.
\end{llemma}
\begin{proof}
Let $J_3$ and $J_4$ denote the left side and the right side of
\eqref{anim60a}, respectively.
Clearly,
\begin{equation} \label{anim61a}
J_4 =
\begin{pmatrix} \sqrt{I_m+P^tP} & P^t \\[8pt]
                O_nP & O_n \sqrt{I_n+PP^t} \end{pmatrix}
\end{equation}
and, by means of the second commuting relation in Lemma \ref{mgn1},
\begin{equation} \label{anim62a}
\begin{split}
J_3 &= \begin{pmatrix} \sqrt{I_m+(O_nP)^t(O_nP)} & (O_nP)^t O_n \\[8pt]
           O_nP & \sqrt{I_n+(O_nP)(O_nP)^t} O_n \end{pmatrix}
\\[12pt] &=
\begin{pmatrix} \sqrt{I_m+P^tP} & P^t \\[8pt]
           O_nP & O_n \sqrt{I_n+PP^t} \end{pmatrix}
\,.
\end{split}
\end{equation}
Hence, $J_3=J_4$, and the proof is complete.
\end{proof}

By \eqref{anim54} and Lemma \ref{mgn2a},
\begin{equation} \label{anim63}
\begin{split}
\Lambda &=
\begin{pmatrix} O_m & 0_{m,n} \\[8pt] 0_{n,m} & I_n \end{pmatrix}
\begin{pmatrix} I_m & 0_{m,n} \\[8pt] 0_{n,m} & O_n \end{pmatrix}
\begin{pmatrix} \sqrt{I_m+P^tP} & P^t \\[8pt]
           P & \sqrt{I_n+PP^t}  \end{pmatrix}
\\[12pt] &=
\begin{pmatrix} O_m & 0_{m,n} \\[8pt] 0_{n,m} & I_n \end{pmatrix}
\begin{pmatrix} \sqrt{I_m+(O_nP)^t(O_nP)} & (O_nP)^t \\[8pt]
           O_nP & \sqrt{I_n+(O_nP)(O_nP)^t}  \end{pmatrix}
\begin{pmatrix} I_m & 0_{m,n} \\[8pt] 0_{n,m} & O_n \end{pmatrix}
\,,
\end{split}
\end{equation}
where $P$, $O_m$ and $O_n$ are generic elements of $\Rnm$,
$SO(m)$ and $SO(n)$, respectively, forming the three matrix parameters
that determine $\Lambda\in SO(m,n)$.

The matrix parameter $P$ of $\Lambda$ in \eqref{anim54} is a
generic element of $\Rnm$, and the orthogonal matrix
$O_n\in SO(n)$ maps $\Rnm$ onto itself bijectively,
$O_n:~P\rightarrow O_nP$. Hence,
the generic element $O_nP\in \Rnm$ in \eqref{anim63}
can, equivalently, be replaced by the generic element
$P\in \Rnm$, thus obtaining from \eqref{anim63} the 
parametric representation of the generic Lorentz transformation $\Lambda$,
\begin{equation} \label{anim66a}
\Lambda =
\begin{pmatrix} O_m & 0_{m,n} \\[8pt] 0_{n,m} & I_n \end{pmatrix}
\begin{pmatrix} \sqrt{I_m+P^tP} & P^t \\[8pt]
           P & \sqrt{I_n+PP^t}  \end{pmatrix}
\begin{pmatrix} I_m & 0_{m,n} \\[8pt] 0_{n,m} & O_n \end{pmatrix}
\,,
\end{equation}
called the {\it bi-gyration decomposition of} $\Lambda$.

The generic Lorentz transformation matrix $\Lambda$ of order
$(m+n)\timess(m+n)$ is expressed in \eqref{anim66a} as the product of
the following three matrices in \eqref{anim64}\,--\,\eqref{anim66}.

{\bf The bi-boost:}
The $(m+n)\timess(m+n)$ matrix $B(P)$,
\begin{equation} \label{anim64}
B(P) :=
\begin{pmatrix} \sqrt{I_m+P^tP} & P^t \\[8pt]
           P & \sqrt{I_n+PP^t}  \end{pmatrix}
\,,
\end{equation}
is parametrized by $P\in \Rnm$, $m,n\in\Nb$.
In order to emphasize that $B(P)$ is associated in \eqref{anim66a}
with a bi-rotation $(O_m,O_n)$, we call it a {\it bi-boost}.
If $m=1$ and $n=3$, the bi-boost descends to the common boost of a Lorentz
transformation in special relativity theory, studied for instance in
\cite{parametrization,mybook01,mybook03,mybook05}.

{\bf The right rotation:}
The $(m+n)\timess(m+n)$ block orthogonal matrix
\begin{equation} \label{anim65}
\rho(O_m) :=
\begin{pmatrix} O_m & 0_{m,n} \\ 0_{n,m} & I_n \end{pmatrix}
\in \Rb^{(m+n)\times(m+n)}
\,,
\end{equation}
is parametrized by
$O_m\in SO(m)$. For $m>1$, $O_m$ is an $m\timess m$ orthogonal matrix,
destined to be
{\it right-applied} to the $n\timess m$ matrices $P$, $P\rightarrow PO_m$.
Hence,
$\rho(O_m)$ is called a right rotation of the bi-boost $B(P)$.

{\bf The left rotation:}
The $(m+n)\timess(m+n)$ block orthogonal matrix
\begin{equation} \label{anim66}
\lambda(O_n) :=
\begin{pmatrix} I_m & 0_{m,n} \\ 0_{n,m} & O_n \end{pmatrix}
\in \Rb^{(m+n)\times(m+n)}
\,,
\end{equation}
is parametrized by
$O_n\in SO(n)$. For $n>1$, $O_n$ is an $n\timess n$ orthogonal matrix,
destined to be
{\it left-applied} to the $n\timess m$ matrices $P$,
$P\rightarrow O_nP$. Hence,
$\lambda(O_n)$ is called a left rotation of the bi-boost $B(P)$.

A left and a right rotation are called collectively a {\it bi-rotation}.
Suggestively, the term {\it bi-boost} emphasizes that the generic
bi-boost is associated with a generic bi-rotation
$(O_n,O_m)\in SO(n)\times SO(m)$.

With the notation in \eqref{anim64}\,--\,\eqref{anim66}, the
results of Lemma \ref{mgn2a} and Lemma \ref{mgn2}
can be written as commuting relations between bi-boosts
and left and right rotations, as the following lemma asserts.

\begin{llemma}\label{mgn3}
The commuting relations
\begin{equation} \label{anim80}
\begin{split}
\lambda(O_n) B(P) &= B(O_nP)\lambda(O_n)
\\[8pt]
B(P) \rho(O_m) &= \rho(O_m) B(PO_m)
\\[8pt]
\lambda(O_n) B(P) \rho(O_m) &=
\rho(O_m) B(O_nPO_m) \lambda(O_n)
\end{split}
\end{equation}
hold for any
$P\in \Rnm$, $O_m\in SO(m)$ and $O_n\in SO(n)$.
\end{llemma}
\begin{proof}
The first matrix identity in \eqref{anim80} is the result of
Lemma \ref{mgn2a}, expressed in the notation in
\eqref{anim64}\,--\,\eqref{anim66}. Similarly,
the second matrix identity in \eqref{anim80} is the result of
Lemma \ref{mgn2}, expressed in the notation in
\eqref{anim64}\,--\,\eqref{anim66}.
The third matrix identity in \eqref{anim80} follows from the
first and second matrix identities in \eqref{anim80},
noting that $\lambda(O_n)$ and $\rho(O_m)$ commute.
\end{proof}


With the notation in \eqref{anim64}\,--\,\eqref{anim66}, the
Lorentz transformation matrix $\Lambda$ in \eqref{anim63},
parametrized by $P,~O_m$ and $O_n$, is given by the equation
\begin{equation} \label{anim67}
\Lambda(O_m,P,O_n) = \rho(O_m)B(P)\lambda(O_n)
\,.
\end{equation}

It proves useful to use the column notation
\begin{equation} \label{anim68}
\Lambda(O_m,P,O_n) =:
\begin{pmatrix} P \\ O_n \\ O_m \end{pmatrix}
\,,
\end{equation}
so that a product of two Lorentz matrices is written as a product
between two column triples. Thus, for instance, the product (or,
composition) of the two Lorentz transformations
$\Lambda_1=\Lambda(O_{n,1},P_1,O_{m,1})$
and
$\Lambda_2=\Lambda(O_{n,2},P_2,O_{m,2})$
is written as
\begin{equation} \label{anim69}
\Lambda_1\Lambda_2
=
\begin{pmatrix} P_1 \\ O_{n,1} \\ O_{m,1} \end{pmatrix}
\begin{pmatrix} P_2 \\ O_{n,2} \\ O_{m,2} \end{pmatrix}
\,.
\end{equation}
The Lorentz transformation product law, written in column notation,
will be presented in Theorem \ref{dokrt}, p.~\pageref{dokrt},
following the study of associated special left and right automorphisms,
called left and right gyrations or, collectively, bi-gyrations.

Formalizing the main result of Sects.~\ref{secc02} and \ref{secc03},
we have the following definition and theorem.

\begin{ddefinition}\label{dokrd1}
{\bf (Special Pseudo-Orthogonal Group).}
A special pseudo-orthogonal transformation $\Lambda$ in the
pseudo-Euclidean space $\Rmcn$ is a linear transformation in $\Rmcn$,
also known as a (generalized, special) Lorentz transformation of
order $(m,n)$, if relative to the basis \eqref{anim01}\,--\,\eqref{anim02},
it leaves the inner product \eqref{anim04} invariant, has
determinant $\det\Lambda=1$, and the determinant of its first $m$ rows
and columns is positive.
The group of all special pseudo-orthogonal transformations in $\Rmcn$,
that is, the group of all Lorentz transformations of order $(m,n)$,
is denoted by $SO(m,n)$.
\end{ddefinition}

\begin{ttheorem}\label{dokrn1}
{\bf (Lorentz Transformation Bi-gyration Decomposition).}
A matrix $\Lambda\in \Rb^{(m+n)\times(m+n)}$ 
is a Lorentz transformation of order $(m,n)$
(that is, a special pseudo-orthogonal transformation in $\Rmcn$),
$\Lambda\in SO(m,n)$, if and only if it is given uniquely by the
bi-gyration decomposition
\begin{equation} \label{anim66b}
\Lambda =
\begin{pmatrix} O_m & 0_{m,n} \\[8pt] 0_{n,m} & I_n \end{pmatrix}
\begin{pmatrix} \sqrt{I_m+P^tP} & P^t \\[8pt]
           P & \sqrt{I_n+PP^t}  \end{pmatrix}
\begin{pmatrix} I_m & 0_{m,n} \\[8pt] 0_{n,m} & O_n \end{pmatrix}
\end{equation}
or, prarmetrically in short,
\begin{equation} \label{anim67t}
\Lambda = \Lambda(O_m,P,O_n) = \rho(O_m)B(P)\lambda(O_n)
=
\begin{pmatrix} P \\ O_n \\ O_m \end{pmatrix}
\,.
\end{equation}
\end{ttheorem}
\begin{proof}
Result \eqref{anim66b} is identical with \eqref{anim66a}.
\end{proof}

\begin{ttheorem}\label{dokrv}
{\bf (Lorentz Transformation Polar Decomposition).}
Any Lorentz Transformation matrix $\Lambda\in SO(m,n)$ possesses
the polar decomposition
\begin{equation} \label{anim66c}
\Lambda =
\begin{pmatrix} \sqrt{I_m+P^tP} & P^t \\[8pt]
           P & \sqrt{I_n+PP^t}  \end{pmatrix}
\begin{pmatrix} O_m & 0_{m,n} \\[8pt] 0_{n,m} & I_n \end{pmatrix}
\begin{pmatrix} I_m & 0_{m,n} \\[8pt] 0_{n,m} & O_n \end{pmatrix}
\,,
\end{equation}
\end{ttheorem}
\begin{proof}
By Lemma \ref{mgn2} and \eqref{anim63}, we have
\begin{equation} \label{anim66d}
\begin{split}
\Lambda &=
\begin{pmatrix} I_m & 0_{m,n} \\[8pt] 0_{n,m} & O_n \end{pmatrix}
\begin{pmatrix} \sqrt{I_m+(PO_m)^t(PO_m)} & (PO_m)^t \\[8pt]
           PO_m & \sqrt{I_n+(PO_m)(PO_m)^t}  \end{pmatrix}
\begin{pmatrix} O_m & 0_{m,n} \\[8pt] 0_{n,m} & I_n \end{pmatrix}
\\[12pt] &=
\begin{pmatrix} \sqrt{I_m+P^tP} & P^t \\[8pt]
           P & \sqrt{I_n+PP^t}  \end{pmatrix}
\begin{pmatrix} I_m & 0_{m,n} \\[8pt] 0_{n,m} & O_n \end{pmatrix}
\begin{pmatrix} O_m & 0_{m,n} \\[8pt] 0_{n,m} & I_n \end{pmatrix}
\,,
\end{split}
\end{equation}
noting that $P\in\Rnm$ is a generic main parameter of $\Lambda\in SO(m,n)$
if and only if $PO_m\in\Rnm$ is a generic main parameter of $\Lambda$
for any $O_m\in SO(m)$.
\end{proof}

\section{Inverse Lorentz Transformation}
\label{secc04}

\begin{ttheorem}\label{dokrn2}
{\bf (The Inverse Bi-boost).}
The inverse of the bi-boost $B(P)$, $P\in \Rnm$, is $B(-P)$,
\begin{equation} \label{duck81}
B(P)^{-1} = B(-P)
\,.
\end{equation}
\begin{proof}
By Lemma \ref{mda1} we have the commuting relations
\begin{equation} \label{anim70}
\begin{split}
P^t \sqrt{I_n+PP^t} &= \sqrt{I_m+P^tP} P^t
\\
\sqrt{I_n+PP^t} P &= P \sqrt{I_m+P^tP}
\end{split}
\end{equation}
which, in the notation in \eqref{anim51}, are
\begin{equation} \label{anim71}
\begin{split}
\hP\hS &= S\hP
\\
\hS P &= PS
\,,
\end{split}
\end{equation}
and, clearly,
\begin{equation} \label{anim72}
\begin{split}
S^2-\hP P &= I_m
\\
\hS^2 - P\hP &= I_n
\,.
\end{split}
\end{equation}
Hence, by \eqref{anim64} and \eqref{anim51},
\begin{equation} \label{anim73}
\begin{split}
B(P)B(-P) &=
\begin{pmatrix} S & \hP \\ P & \hS \end{pmatrix}
\begin{pmatrix} S &-\hP \\ -P & \hS \end{pmatrix}
\\[12pt] &=
\begin{pmatrix} S^2-\hP P & -S\hP+\hP\hS \\
  PS-\hS P & -P\hP+\hS^2 \end{pmatrix}
\\[12pt] &=
\begin{pmatrix} I_m & 0_{m,n} \\ 0_{n,m} & I_n \end{pmatrix}
= I_{m+n}
\,,
\end{split}
\end{equation}
as desired.
\end{proof}
\end{ttheorem}

A Lorentz transformation matrix $\Lambda$ of
order $(m+n)\timess (m+n)$, $m,n\ge2$,
involves the bi-rotation $\lambda(O_n)\rho(O_m)$, as shown in
\eqref{anim67t}. Bi-boosts are Lorentz transformations without
bi-rotations, that is by \eqref{anim67}\,--\,\eqref{anim68},
bi-boosts $B(P)$ are
\begin{equation} \label{anim77}
\Lambda(I_m,P,I_n)=\rho(I_m)B(P)\lambda(I_n) = B(P) =
\begin{pmatrix} P \\ I_n \\ I_m \end{pmatrix}
\,,
\end{equation}
for any $P\in \Rnm$.

Rewriting \eqref{duck81} in the column notation, we have
\begin{equation} \label{anim78}
\begin{pmatrix} P \\ I_n \\ I_m \end{pmatrix}^{-1}
=
\begin{pmatrix} -P \\ I_n \\ I_m \end{pmatrix}
\,,
\end{equation}
so that, accordingly,
\begin{equation} \label{anim79}
\begin{pmatrix}  P \\ I_n \\ I_m \end{pmatrix}
\begin{pmatrix} -P \\ I_n \\ I_m \end{pmatrix}
=
\begin{pmatrix}  0_{n,m} \\ I_n \\ I_m \end{pmatrix}
\,,
\end{equation}
$(0_{n,m},I_n,I_m)^t$ being the identity Lorentz transformation
of order $(m,n)$.

\begin{ttheorem}\label{dokrn3}
{\bf (The Inverse Lorentz Transformation).}
The inverse of a Lorentz transformation $\Lambda=(P,O_n,O_m)^t$
is given by the equation
\begin{equation} \label{anim83}
\begin{pmatrix} P \\ O_n \\ O_m  \end{pmatrix}^{-1}
=
\begin{pmatrix} -O_n^tPO_m^t \\ O_n^t \\ O_m^t \end{pmatrix}
\end{equation}
\end{ttheorem}
\begin{proof}
The proof is given by the following chain of equations, which are numbered
for subsequent explanation.
\begin{equation} \label{anim84}
\begin{split}
\Lambda(O_m,P,O_n)^{-1}
&
\overbrace{=\!\!=\!\!=}^{(1)} \hspace{0.2cm}
\{\rho(O_m)B(P)\lambda(O_n)\}^{-1}
\\&
\overbrace{=\!\!=\!\!=}^{(2)} \hspace{0.2cm}
\lambda(O_n^t) B(-P) \rho(O_m^t)
\\&
\overbrace{=\!\!=\!\!=}^{(3)} \hspace{0.2cm}
B(-O_n^tP) \lambda(O_n^t) \rho(O_m^t)
\\&
\overbrace{=\!\!=\!\!=}^{(4)} \hspace{0.2cm}
B(-O_n^tP) \rho(O_m^t) \lambda(O_n^t)
\\&
\overbrace{=\!\!=\!\!=}^{(5)} \hspace{0.2cm}
\rho(O_m^t) B(-O_n^tPO_m^t) \lambda(O_n^t)
\,.
 \end{split}
 \end{equation}
Derivation of the numbered equalities in \eqref{anim84} follows:
\begin{enumerate}
\item \label{madong1}
By \eqref{anim67t}.
\item \label{madong2}
Obvious, noting \eqref{duck81}.
\item \label{madong3}
Follows from \eqref{madong2} by the first matrix identity in \eqref{anim80}.
\item \label{madong4}
Follows from \eqref{madong3} by commuting $\lambda(O_n^t)$ and $\rho(O_m^t)$.
\item \label{madong5}
Follows from \eqref{madong4} by the second matrix identity in \eqref{anim80}.
\end{enumerate}
\end{proof}

\section{Bi-boost Parameter Recognition}
\label{secc05}

Composing the bi-gyration decomposition \eqref{anim66b}
of the Lorentz transformation $\Lambda\in SO(m,n)$
in Theorem \ref{dokrn1}, we have the Lorentz transformation
\begin{equation} \label{anim86}
\Lambda =
\begin{pmatrix} O_m \sqrt{I_m+P^tP} & O_mP^tO_n \\[8pt]
                P & \sqrt{I_n+PP^t} O_n \end{pmatrix}
=:
\begin{pmatrix} E_{11} & E_{12} \\[8pt] E_{21} & E_{22} \end{pmatrix}
\,,
\end{equation}
parametrized by the three parameters
\begin{enumerate}
\item \label{marong1}
$P\in\Rnm$, an $n\times m$ real matrix, called the {\it main parameter}
of the Lorentz transformation $\Lambda\in SO(m,n)$;
\item \label{marong2}
$O_n\in SO(n)$, a left rotation of $P$ (or, equivalently,
a right rotation of $P^t$); and
\item \label{marong3}
$O_m\in SO(m)$, a right rotation of $P$ (or, equivalently,
a left rotation of $P^t$).
\end{enumerate}

We naturally face the task of determining the matrix parameters
$P$, $O_n$ and $O_m$ of the $SO(m,n)$ matrix $\Lambda$ in
\eqref{anim86} from its block entries $E_{ij}$, $i,j=1,2$.

The matrix parameters $O_m$ and $O_n$ of $\Lambda$ in \eqref{anim66b}
cannot be recognized from \eqref{anim86} straightforwardly by inspection.
Fortunately, however, the matrix parameter $P$ is recovered from
\eqref{anim86} by straightforward inspection,
$P=E_{21}$, thus obtaining the first equation in \eqref{anim87} below.
Then, following \eqref{anim86} we have
$I_m+P^tP=I_m+E_{21}^tE_{21}$
and
$I_n+PP^t=I_n+E_{21}E_{21}^t$,
so that \eqref{anim86} yields the following
{\it parameter recognition formulas}:
\begin{equation} \label{anim87}
\begin{split}
P&=E_{21}
\\
O_m &= E_{11} \sqrt{I_m+E_{21}^tE_{21}}^{~-1}
\\
O_n &= \sqrt{I_n+E_{21}E_{21}^t}^{~-1} E_{22}
\\
O_m P^t O_n &= E_{12}
\,.
\end{split}
\end{equation}

In the parameter recognition formulas \eqref{anim87}
the parameters $P$, $O_n$ and $O_m$
of the composite Lorentz transformation $\Lambda$ in \eqref{anim86}
and in the decomposed Lorentz transformation $\Lambda$ in \eqref{anim66b}
are recognized from the block entries $E_{ij}$, $i,j=1,2$, of the
composite Lorentz transformation \eqref{anim86}.
Our ability to recover the main parameter of a Lorentz transformation
suggests the following definition of main parameter composition,
called {\it bi-gyroaddition}.
\begin{ddefinition}\label{dfkhvb}
{\bf (Bi-gyroaddition, Bi-gyrogroupoid).}
Let $\Lambda=B(P_1)B(P_2)$ be a Lorentz transformation given by
the product of two bi-boosts parametrized by $P_1,P_2\in\Rnm$.
Then, the main parameter, $P_{12}$, of $\Lambda$ is said to be
the composition of $P_1$ and $P_2$,
\begin{equation} \label{ahfc}
P_{12}=P_1\op P_2
\,,
\end{equation}
giving rise to a binary operation, $\op$, called bi-gyroaddition,
in the space $\Rnm$ of all $n\timess m$ real matrices.
Being a groupoid of the parameter $P\inn\Rnm$,
the resulting groupoid $(\Rnm,\op)$ is called the
parameter bi-gyrogroupoid.
\end{ddefinition}

Definition \ref{dfkhvb} encourages us to the study of the
bi-boost composition law in Sect.~\ref{secc06}.

\section{Bi-boost Composition Parameters}
\label{secc06}

In general, the product of two bi-boosts is not a bi-boost.
However, the product of two bi-boosts is an element of the
Lorentz group $SO(m,n)$ and, hence, by
Theorem \ref{dokrn1}, can be parametrized, as shown in
Sect.~\ref{secc05}.
Following \eqref{anim64}, let
\begin{equation} \label{anim88}
B(P_k) =
\begin{pmatrix} \sqrt{I_m+P_k^tP_k} & P_k^t \\[8pt]
           P_k & \sqrt{I_n+P_kP_k^t}  \end{pmatrix}
\,,
\end{equation}
$k=1,2$, be two bi-boosts, so that their product is
\begin{equation} \label{anim89}
\begin{split}
&\Lambda=B(P_1)B(P_2)
\\[8pt] &=
\begin{pmatrix} \sqrt{I_m+P_1^tP_1} \sqrt{I_m+P_2^tP_2} + P_1^tP_2 &
        \sqrt{I_m+P_1^tP_1} P_2^t + P_1^t \sqrt{I_n+P_2P_2^t}
\\[8pt]
P_1 \sqrt{I_m+P_2^tP_2} + \sqrt{I_n+P_1P_1^t} P_2 &
P_1P_2^t + \sqrt{I_n+P_1P_1^t} \sqrt{I_n+P_2P_2^t}
\end{pmatrix}
\\[8pt] &=:
\begin{pmatrix} E_{11} & E_{12} \\ E_{21} & E_{22} \end{pmatrix}
\,.
\end{split}
\end{equation}

By the parameter recognition formulas \eqref{anim87}, the main parameter
\begin{equation} \label{camon}
P_{12}= P_1\op P_2
\end{equation}
and the left and right rotation parameters
$O_{n,12}$ and $O_{m,12}$
of the bi-boost product $\Lambda=B(P_1)B(P_2)$ in \eqref{anim89}
are given by
\begin{equation} \label{anim90}
\begin{split}
P_{12} = P_1 \op P_2 &=E_{21}
\\
O_{n,12} &= \sqrt{I_n+E_{21}E_{21}^t}^{~-1} E_{22}
\\
O_{m,12} &= E_{11} \sqrt{I_m+E_{21}^tE_{21}}^{~-1}
\\
O_{m,12} P_{12}^t O_{n,12} &= E_{12}
\,,
\end{split}
\end{equation}
where $E_{ij}$, $i,j=1,2$, are defined by the last equation in \eqref{anim89}.

Hence, by \eqref{anim67},
\begin{equation} \label{anim91}
\Lambda = B(P_1)B(P_2) = \rho(O_{m,12}) B(P_1 \op P_2) \lambda(O_{n,12})
\,.
\end{equation}
Following Def.~\ref{dfkhvb},
we view $\op$ as a binary operation between elements $P\in \Rnm$,
thus obtaining the {\it bi-gyrogroupoid} $(\Rnm, \op)$ that will give rise to
a group-like structure called a {\it bi-gyrogroup}.
Accordingly, the binary operation $\op$ is the {bi-gyrooperation} of $\Rnm$,
called {\it bi-gyroaddition}, and
$P_1\op P_2$ is the {\it bi-gyrosum}
of $P_1$ and $P_2$ in $\Rnm$.

It is now convenient to rename the
right rotation $O_{m,12}$ and the left rotation $O_{n,12}$
in \eqref{anim90}\,--\,\eqref{anim91} as
a right and a left {\it gyrations}. In symbols,
\begin{equation} \label{anim92}
\begin{split}
O_{m,12} &=: \rgyr[P_1,P_2]
\in SO(m)
\\
O_{n,12} &=: \lgyr[P_1,P_2]
\in SO(n)
\,.
\end{split}
\end{equation}
We call $\rgyr[P_1,P_2]$
the {\it right gyration} generated by $P_1$ and $P_2$,
and call $\lgyr[P_1,P_2]$
the {\it left gyration} generated by $P_1$ and $P_2$.
The pair of a left and a right gyration, each generated by $P_1$ and $P_2$,
is viewed collectively as the {\it bi-gyration} generated by $P_1$ and $P_2$.

The bi-boost product \eqref{anim91} is now written as
\begin{equation} \label{anim93}
B(P_1)B(P_2) = \rho(\rgyr[P_1,P_2]) B(P_1 \op P_2)
\lambda(\lgyr[P_1,P_2])
\,,
\end{equation}
demonstrating that the product of two bi-boosts generated by
$P_1$ and $P_2$ is a bi-boost generated by $P_1\op P_2$ along with
a bi-gyration generated by $P_1$ and $P_2$.

The bi-gyrosum $P_1\op P_2$ of $P_1$ and $P_2$, and the bi-gyrations
generated by $P_1$ and $P_2$ that appear in \eqref{anim93} are
determined from \eqref{anim89}\,--\,\eqref{anim92},
\begin{equation} \label{anim94}
\begin{split}
P_1\op P_2 &= P_1 \sqrt{I_m+P_2^tP_2} + \sqrt{I_n+P_1P_1^t} P_2
\\[6pt]
\rgyr[P_1,P_2] &=
\left\{ P_1^tP_2 + \sqrt{I_m+P_1^tP_1}\sqrt{I_m+P_2^tP_2} \right\}
\sqrt{I_m+(P_1\op P_2)^t(P_1\op P_2)}^{~-1}
\\[6pt]
\lgyr[P_1,P_2] &=
\sqrt{I_n+(P_1\op P_2)(P_1\op P_2)^t}^{~-1}
\left\{ P_1P_2^t + \sqrt{I_n+P_1P_1^t}\sqrt{I_n+P_2P_2^t} \right\}
\\[6pt]
\rgyr[P_1,P_2] & (P_1 \op P_2)^t \lgyr[P_1,P_2]
 = \sqrt{I_m+P_1^tP_1} P_2^t + P_1^t \sqrt{I_n+P_2P_2^t}
\\[6pt] &
\overbrace{=\!\!=\!\!=}^{(1)} \hspace{0.2cm}
P_1^t \op P_2^t
~
\overbrace{=\!\!=\!\!=}^{(2)} \hspace{0.2cm}
(P_2 \op P_1)^t
\,.
\end{split}
\end{equation}

The equation marked by $(1)$ in \eqref{anim94} follows immediately from
the first equation in \eqref{anim94}, replacing
$P_1,P_2\in\Rnm$ by $P_1^t,P_2^t\in\Rmn$.

The equation marked by $(2)$ in \eqref{anim94} is derived from the
first equation in \eqref{anim94} in the following straightforward
chain of equations.
\begin{equation} \label{drain01}
\begin{split}
(P_2\op P_1)^t &= \left\{
P_2 \sqrt{I_m+P_1^tP_1} + \sqrt{I_n+P_2P_2^t} P_1
\right\}^t
\\[8pt] &=
\sqrt{I_m+P_1^tP_1} P_2^t + P_1^t \sqrt{I_n+P_2P_2^t}
\\[8pt] &=
(P_1^t) \sqrt{I_n+(P_2^t)^tP_2^t} + \sqrt{I_m+(P_1^t)(P_1^t)^t}(P_2^t)
\\[8pt] &=
P_1^t \op P_2^t
\,.
\end{split}
\end{equation}

Formalizing results in \eqref{anim94}, we obtain the following
theorem.
\begin{ttheorem}\label{smld} 
{\bf (Bi-gyroaddition and Bi-gyration).}
The bi-gyroaddition and bi-gyration in the parameter
bi-gyrogroupoid $(\Rnm,\op)$ are given by the equations
\begin{equation} \label{anim94s}
\begin{split}
P_1\op P_2 &= P_1 \sqrt{I_m+P_2^tP_2} + \sqrt{I_n+P_1P_1^t} P_2
\\[6pt]
\lgyr[P_1,P_2] &=
\sqrt{I_n+(P_1\op P_2)(P_1\op P_2)^t}^{~-1}
\left\{ P_1P_2^t + \sqrt{I_n+P_1P_1^t}\sqrt{I_n+P_2P_2^t} \right\}
\\[6pt]
\rgyr[P_1,P_2] &=
\left\{ P_1^tP_2 + \sqrt{I_m+P_1^tP_1}\sqrt{I_m+P_2^tP_2} \right\}
\sqrt{I_m+(P_1\op P_2)^t(P_1\op P_2)}^{~-1}
\end{split}
\end{equation}
for all $P_1,P_2\in\Rnm$.
\end{ttheorem}

The following corollary results immediately from Theorem \ref{smld}.
\begin{ccorollary}\label{vftr}
{\bf (Trivial Bi-gyrations).}
\begin{equation} \label{anglv1}
\begin{split}
\lgyr[0_{n,m},P] = \lgyr[P,0_{n,m}] &= I_n
\\[4pt]
\lgyr[\om P,P] = \lgyr[P,\om P] &= I_n
\\[4pt]
\rgyr[0_{n,m},P] = \rgyr[P,0_{n,m}] &= I_m
\\[4pt]
\rgyr[\om P,P] = \lgyr[P,\om P] &= I_m
\end{split}
\end{equation}
for all $P\in\Rnm$.
\end{ccorollary}

The trivial bi-gyration
\begin{equation} \label{anglv2}
\begin{split}
\lgyr[P,P] &= I_n
\\[4pt]
\rgyr[P,P] &= I_m
\end{split}
\end{equation}
for all $P\in\Rnm$ cannot be derived immediately from Theorem \ref{smld}.
It will, therefore, be derived in \eqref{avnet08} and \eqref{avnet06},
and formalized in Theorem \ref{kdrnb}, p.~\pageref{kdrnb}.

The bi-boost product \eqref{anim93}, written in the column notation,
takes the elegant form
\begin{equation} \label{anim95}
B(P_1)B(P_2) =
\begin{pmatrix} P_1 \\[4pt] I_n \\[4pt] I_m \end{pmatrix}
\begin{pmatrix} P_2 \\[4pt] I_n \\[4pt] I_m \end{pmatrix}
=
\begin{pmatrix} P_1\op P_2 \\[4pt] \lgyr[P_1,P_2]
               \\[4pt] \rgyr[P_1,P_2] \end{pmatrix}
\,,
\end{equation}
for all $P_1,P_2\in\Rnm$.

When $P_1=P$ and $P_2=-P$, \eqref{anim95} specializes to
\begin{equation} \label{anim95p1}
B(P)B(-P)=
\begin{pmatrix} P \\[4pt] I_n \\[4pt] I_m \end{pmatrix}
\begin{pmatrix}-P \\[4pt] I_n \\[4pt] I_m \end{pmatrix}
=
\begin{pmatrix} P\op(-P) \\[4pt] \lgyr[P,-P]
                         \\[4pt] \rgyr[P,-P] \end{pmatrix}
\,,
\end{equation}
for all $P\in\Rnm$.
But, the left side of \eqref{anim95p1} is also determined in
\eqref{anim79}, implying the identities
\begin{equation} \label{anim95p2}
\begin{split}
P\op(-P) &= 0_{n,m}
\\
\lgyr[P,-P] &= I_n
\\
\rgyr[P,-P] &= I_m
\,.
\end{split}
\end{equation}

The first equation in \eqref{anim95p2} implies that
\begin{equation} \label{anim95pa}
-P=:\om P
\end{equation}
is the inverse of $P$ with respect to the binary operation
$\op$ in $\Rnm$. Hence, we use the notations $-P$ and $\om P$
interchangeably. Furthermore, we naturally use the notation
$P_1\op(-P_2) = P_1\op(\om P_2) = P_1\om P_2$, and
rewrite \eqref{anim95p2} as
\begin{equation} \label{anim95p3}
\begin{split}
P\om P &= 0
\\
\lgyr[P,\om P] &= I_n
\\
\rgyr[P,\om P] &= I_m
\end{split}
\end{equation}
in agreement with \eqref{anglv1}.

Similarly, we rewrite \eqref{anim78} as
\begin{equation} \label{anim95d3}
\begin{pmatrix} P \\ I_n \\ I_m \end{pmatrix}^{-1}
=
\begin{pmatrix} \om P \\ I_n \\ I_m \end{pmatrix}
\,,
\end{equation}
that is,
\begin{equation} \label{dfhn}
B(P)^{-1}=B(\om P)
\end{equation}
for all $P\in\Rnm$.

The first equation in \eqref{anim94} implies that
\begin{equation} \label{anim95d4}
(-P_1)\op(-P_2)=-(P_1\op P_2)
\,.
\end{equation}
Hence, following the definition of $\om P$ in \eqref{anim95pa},
and by \eqref{anim95d4},
the bi-gyroaddition $\op$ obeys the {\it gyroautomorphic inverse property}
\begin{equation} \label{anim94p1}
\om(P_1\op P_2) = \om P_1 \om P_2
\,,
\end{equation}
for all $P_1,P_2\in\Rnm$.

It follows from the gyroautomorphic inverse property \eqref{anim94p1}
and from \eqref{anim94} that bi-gyrations are {\it even},
that is,
\begin{equation} \label{anim94p2}
\begin{split}
\lgyr[-P_1,-P_2] &= \lgyr[P_1,P_2]
\\
\rgyr[-P_1,-P_2] &= \lgyr[P_1,P_2]
\end{split}
\end{equation}
or, equivalently,
\begin{equation} \label{anim94e2}
\begin{split}
\lgyr[\om P_1,\om P_2] &= \lgyr[P_1,P_2]
\\
\rgyr[\om P_1,\om P_2] &= \lgyr[P_1,P_2]
\,.
\end{split}
\end{equation}

 \section{Automorphisms of the Parameter Bi-gyrogroupoid}
\label{secc07}

Left and right rotations turn out to be left and right automorphisms
of the parameter bi-gyrogroupoid $(\Rnm,\op)$. We recall that
a groupoid, $(S,+)$, is a nonempty set, $S$, with a binary operation, $+$.
A {\it left automorphism} of a groupoid $(S,+)$ is a bijection $f$ of $S$,
$f:S\rightarrow S$, $s\mapsto fs$, that
respects the binary operation, that is, $f(s_1+s_2)=fs_1+fs_2$.
Similarly,
a {\it right automorphism} of a groupoid $(S,+)$ is a bijection $f$ of $S$,
$f:S\rightarrow S$, $s\mapsto sf$, that
respects the binary operation, that is, $(s_1+s_2)f=s_1f+s_2f$.
The need to distinguish between left and right automorphisms of the
bi-gyrogroupoid $(\Rnm,\op)$ is clear from Theorem \ref{dokrd} below.

\begin{ttheorem}\label{dokrd}
{\bf (Left and Right Automorphisms of $(\Rnm,\op)$).}
Any rotation $O_n\in SO(n)$ is a left automorphism of the
parameter bi-gyrogroupoid $(\Rnm,\op)$, and
any rotation $O_m\in SO(m)$ is a right automorphism of the
parameter bi-gyrogroupoid $(\Rnm,\op)$, that is,
\begin{equation} \label{adin105}
\begin{split}
O_n(P_1 \op P_2) &= O_nP_1 \op O_nP_2
\\
(P_1 \op P_2)O_m &= P_1O_m \op P_2O_m
\\
O_n(P_1 \op P_2)O_m &= O_nP_1O_m \op O_nP_2O_m
\end{split}
\end{equation}
for all $P_1,P_2\in\Rnm$, $O_n\in SO(n)$ and $O_m\in SO(m)$.
\end{ttheorem}
\begin{proof}
By the first equation in \eqref{anim94s} and the second equation in
\eqref{anim55},
\begin{equation} \label{adin106}
\begin{split}
O_n(P_1\op P_2) &= O_n(
P_1\sqrt{I_m+P_2^tP_2} + \sqrt{I_n+P_1P_1^t} P_2
)
\\[4pt] &=
O_n P_1\sqrt{I_m+(O_nP_2)^t(O_nP_2)} + O_n\sqrt{I_n+P_1P_1^t} P_2
\\[4pt] & \hspace{-1.0cm} =
O_n P_1\sqrt{I_m+(O_nP_2)^t(O_nP_2)} + \sqrt{I_n+(O_nP_1)(O_nP_1)^t} O_nP_2
\\[4pt] &=
O_nP_1 \op O_nP_2
\,,
\end{split}
\end{equation}
thus proving the first identity in \eqref{adin105}.

Similarly, by the first equation in \eqref{anim94s} and the first equation in
\eqref{anim55},
\begin{equation} \label{adin107}
\begin{split}
(P_1\op P_2) O_m &= (
P_1\sqrt{I_m+P_2^tP_2} + \sqrt{I_n+P_1P_1^t} P_2
)O_m
\\[4pt] &=
P_1\sqrt{I_m+P_2^tP_2} O_m + \sqrt{I_n+P_1P_1^t} P_2 O_m
\\[4pt] & \hspace{-1.0cm} =
P_1 O_m\sqrt{I_m+(P_2O_m)^t(P_2 O_m)} + \sqrt{I_n+(P_1O_m)(P_1O_m)^t} P_2O_m
\\[4pt] &=
P_1O_m \op P_2O_m
\,,
\end{split}
\end{equation}
thus proving the second identity in \eqref{adin105}.
The third identity in \eqref{adin105} follows immediately from the
first two identities in \eqref{adin105}.
\end{proof}

By \eqref{anim92} and Theorem \ref{dokrd}, left gyrations,
$\lgyr[P_1,P_2]$, and right gyrations, $\rgyr[P_1,P_2]$, $P_1,P_2\in\Rnm$,
are left and right automorphisms
of $(\Rnm,\op)$. Hence, left and right gyrations
are also called left and right gyroautomorphisms of $(\Rnm,\op)$ or,
collectively, {\it bi-gyroautomorphisms} of
the parameter bi-gyrogroupoid $(\Rnm,\op)$.

Since $-P=\om P$, we clearly have the identities
\begin{equation} \label{arnd}
\begin{split}
O_n(\om P) &= \om O_nP
\\[4pt]
(\om P) O_m &= \om PO_m
\\[4pt]
O_n(\om P) O_m &= \om O_nPO_m
\end{split}
\end{equation}
for all $P\in\Rnm$, $O_n\in SO(n)$ and $O_m\in SO(m)$.

 \section{The Bi-boost Square}
\label{secc08a}

We are now in the position to determine the parameters of the squared
bi-boost.
If we use the convenient notation
\begin{equation} \label{avnet01}
\begin{split}
b_m &:= \sqrt{I_m+P^tP}
\\[6pt]
b_n &:= \sqrt{I_n+PP^t}
\,,
\end{split}
\end{equation}
$P\in\Rnm$,
then, by \eqref{anim64},
\begin{equation} \label{bibsteq}
B(P) = \begin{pmatrix} b_m & P^t \\[4pt] P & b_n \end{pmatrix}
\,,
\end{equation}
and the squared bi-boost $B(P)$ leads to the following
chain of equations, which are numbered
for subsequent explanation.
\begin{equation} \label{avnet02}
\begin{split}
B(P)^2 ~
&\overbrace{=\!\!=\!\!=}^{(1)} \hspace{0.2cm}
\begin{pmatrix} b_m & P^t \\[4pt] P & b_n \end{pmatrix}
\begin{pmatrix} b_m & P^t \\[4pt] P & b_n \end{pmatrix}
\\
&\overbrace{=\!\!=\!\!=}^{(2)} \hspace{0.2cm}
\begin{pmatrix} b_m^2+P^tP & b_mP^t+P^tb_n \\[4pt]
Pb_m+b_nP & b_n^2+PP^t \end{pmatrix}
\\
&\overbrace{=\!\!=\!\!=}^{(3)} \hspace{0.2cm}
\begin{pmatrix} I_m+2P^tP & 2P^tb_n \\[4pt]
2b_nP & I_n+2PP^t \end{pmatrix}
\\
&\overbrace{=\!\!=\!\!=}^{(4)}:\hspace{0.2cm}
\begin{pmatrix} E_{11} & E_{12} \\[4pt] E_{21} & E_{22} \end{pmatrix}
\,.
\end{split}
\end{equation}
Derivation of the numbered equalities in \eqref{avnet02} follows:
\begin{enumerate}
\item \label{nidrt01}
This equation follows from \eqref{bibsteq}.
\item \label{nidrt02}
Follows from Item \eqref{nidrt01} by block matrix multiplication.
\item \label{nidrt03}
Results from \eqref{avnet01} and the commuting relations
\eqref{avir01} and \eqref{avir02}.
\item \label{nidrt04}
This equation defines $E_{ij}$, $i,j=1,2$.
\end{enumerate}

Hence, by the parameter recognition formulas
\eqref{anim90}, along with \eqref{anim92}, we have
\begin{equation} \label{avnet03}
\begin{split}
P\op P &= E_{21}
\\[6pt]
\rgyr[P,P] &= E_{11} \sqrt{I_m+E_{21}^tE_{21}}^{~-1}
\\[6pt]
\lgyr[P,P] &= \sqrt{I_n+E_{21}E_{21}^t}^{~-1} E_{22}
\\[6pt]
\rgyr[P,P] (P\op P)^t \lgyr[P,P] &= E_{12}
\,,
\end{split}
\end{equation}
where $E_{ij}$ are given by Item \eqref{nidrt04} of \eqref{avnet02}.

Following the first equation in \eqref{avnet03} and the definition of
$E_{21}$ in Item \eqref{nidrt04} of \eqref{avnet02}, and
\eqref{avir02}, we have the equations
\begin{equation} \label{avnet04}
E_{21}=P\op P=2b_nP=2Pb_m
\,.
\end{equation}

Let us consider
the following chain of equations, some of which are numbered
for subsequent explanation.
\begin{equation} \label{avnet05}
\begin{split}
E_{11}
&\overbrace{=\!\!=\!\!=}^{(1)} \hspace{0.2cm}
I_m + 2P^tP
\\
&\hspace{0.1cm} {=\!\!=\!\!=} \hspace{0.2cm}      
\left\{ I_m + 4P^tP + 4(P^tP)^2 \right\}^{\half}
\\
&\hspace{0.1cm} {=\!\!=\!\!=} \hspace{0.2cm}      
\left\{ I_m + 4(I_m+P^tP)P^tP \right\}^{\half}
\\
&\hspace{0.1cm} {=\!\!=\!\!=} \hspace{0.2cm}      
\left\{ I_m + 4b_m^2 P^tP \right\}^{\half}
\\
&\overbrace{=\!\!=\!\!=}^{(2)} \hspace{0.2cm}
\left\{ I_m + 4P^t b_n^2 P \right\}^{\half}
\\
&\hspace{0.1cm} {=\!\!=\!\!=} \hspace{0.2cm}      
\left\{ I_m + 2(b_nP)^t 2b_nP \right\}^{\half}
\\
&\overbrace{=\!\!=\!\!=}^{(3)} \hspace{0.2cm}
\sqrt{I_m + E_{21}^t E_{21}}
\,.
\end{split}
\end{equation}
Derivation of the numbered equalities in \eqref{avnet05} follows:
\begin{enumerate}
\item \label{tigri01}
This equation follows from the definition of
$E_{11}$ in Item \eqref{nidrt04} of \eqref{avnet02}.
\item \label{tigri02}
This equation is obtained from its predecessor by two successive
applications of the commuting relation \eqref{avir01}.
\item \label{tigri03}
Follows from \eqref{avnet04}.
\end{enumerate}

We see from \eqref{avnet05} and the second equation in \eqref{avnet03}
that the right gyration generated by $P$ and $P$ is trivial,
\begin{equation} \label{avnet06}
\rgyr[P,P] = I_m
\end{equation}
for all $P\in\Rnm$.

Similarly to \eqref{avnet05}, let us consider
the following chain of equations, some of which are numbered
for subsequent explanation.
\begin{equation} \label{avnet07}
\begin{split}
E_{22}
&\overbrace{=\!\!=\!\!=}^{(1)} \hspace{0.2cm}
I_n + 2PP^t
\\
&\hspace{0.1cm} {=\!\!=\!\!=} \hspace{0.2cm}      
\left\{ I_n + 4PP^t + 4(PP^t)^2 \right\}^{\half}
\\
&\hspace{0.1cm} {=\!\!=\!\!=} \hspace{0.2cm}      
\left\{ I_n + 4(I_n+PP^t)PP^t \right\}^{\half}
\\
&\hspace{0.1cm} {=\!\!=\!\!=} \hspace{0.2cm}      
\left\{ I_n + 4b_n^2 PP^t \right\}^{\half}
\\
&\overbrace{=\!\!=\!\!=}^{(2)} \hspace{0.2cm}
\left\{ I_n + 4P b_m^2 P^t \right\}^{\half}
\\
&\hspace{0.1cm} {=\!\!=\!\!=} \hspace{0.2cm}      
\left\{ I_n + 2Pb_m 2(Pb_m)^t \right\}^{\half}
\\
&\overbrace{=\!\!=\!\!=}^{(3)} \hspace{0.2cm}
\sqrt{I_n + E_{21} E_{21}^t}
\,.
\end{split}
\end{equation}
Derivation of the numbered equalities in \eqref{avnet07} follows:
\begin{enumerate}
\item \label{tigrs01}
This equation follows from the definition of
$E_{22}$ in Item \eqref{nidrt04} of \eqref{avnet02}.
\item \label{tigrs02}
This equation is obtained from its predecessor by two successive
applications of the commuting relation \eqref{avir02}.
\item \label{tigrs03}
Follows from \eqref{avnet04}.
\end{enumerate}

We see from \eqref{avnet07} and the third equation in \eqref{avnet03}
that the left gyration generated by $P$ and $P$ is trivial,
\begin{equation} \label{avnet08}
\lgyr[P,P] = I_n
\end{equation}
for all $P\in\Rnm$.

It follows from \eqref{avnet03}\,--\,\eqref{avnet08} that
\begin{equation} \label{kukor01}
\begin{split}
E_{21} &= P\op P
\\[6pt]
E_{12} &= (P\op P)^t
\\[6pt]
E_{11} &= \sqrt{I_m + (P\op P)^t (P\op P)}
\\[6pt]
E_{22} &= \sqrt{I_n + (P\op P) (P\op P)^t}
\,.
\end{split}
\end{equation}

Hence, by the extreme sides of \eqref{avnet02}
\begin{equation} \label{kukor02}
B(P)^2 = B(P\op P)
\,,
\end{equation}
so that a the square of a bi-boost is, again, a bi-boost.

As a byproduct of squaring the bi-boost, we have obtained the results
in \eqref{avnet06} and \eqref{avnet08}, which we formalize in the
following theorem.
\begin{ttheorem}\label{kdrnb} 
{\bf (A Trivial Bi-gyration).}
\begin{equation} \label{kdrnb1}
\begin{split}
\lgyr[P,P] &= I_n
\\[4pt]
\rgyr[P,P] &= I_m
\end{split}
\end{equation}
for all $P\in\Rnm$.
\end{ttheorem}

\section{Commuting Relations Between Bi-gyrations and Bi-rotations}
\label{secc08}

Bi-gyrations $(\lgyr[P_1,P_2], \rgyr[P_1,P_2])\in SO(n)\times SO(m)$
and bi-rotations $( O_n,O_m)\in SO(n)\times SO(m)$
commute in a special, interesting way stated in the following theorem.

\begin{ttheorem}\label{spcomw} 
{\bf (Bi-gyration -- bi-rotation Commuting Relation).}
\begin{equation} \label{spc01}
O_n\lgyr[P_1,P_2] = \lgyr[O_nP_1,O_nP_2]O_n
\end{equation}
and
\begin{equation} \label{spc02}
\rgyr[P_1,P_2]O_m = O_m\rgyr[P_1O_m,P_2O_m]
\end{equation}
for all $P_1,P_2\in\Rnm$, $O_n\in SO(n)$ and $O_m\in SO(m)$.
\end{ttheorem}
\begin{proof}
The matrix identity \eqref{spc01} is proved in the following
chain of equations, which are numbered
for subsequent explanation.
\begin{equation} \label{spc03}
\begin{split}
O_n\lgyr[P_1,P_2] ~
&
\overbrace{=\!\!=\!\!=}^{(1)} \hspace{0.2cm}
O_n\sqrt{I_n+(P_1\op P_2)(P_1\op P_2)^t}^{~-1}
(P_1P_2^t + \sqrt{I_n+P_1P_1^t}\sqrt{I_n+P_2P_2^t})
\\&
\hspace{-2.0cm}
\overbrace{=\!\!=\!\!=}^{(2)} \hspace{0.2cm}
\sqrt{I_n+(O_nP_1\op O_nP_2)(O_nP_1\op O_nP_2)^t}^{~-1}
O_n(P_1P_2^t + \sqrt{I_n+P_1P_1^t}\sqrt{I_n+P_2P_2^t})
\\&
\hspace{-2.0cm}
\overbrace{=\!\!=\!\!=}^{(3)} \hspace{0.2cm}
\sqrt{I_n+(O_nP_1\op O_nP_2)(O_nP_1\op O_nP_2)^t}^{~-1}
(O_nP_1P_2^t + O_n\sqrt{I_n+P_1P_1^t}\sqrt{I_n+P_2P_2^t})
\\&
\hspace{-2.0cm}
\overbrace{=\!\!=\!\!=}^{(4)} \hspace{0.2cm}
\sqrt{I_n+(O_nP_1\op O_nP_2)(O_nP_1\op O_nP_2)^t}^{~-1}
\\ & \hspace{-1.0cm} \times
\left\{ (O_nP_1)(O_nP_2)^t +
\sqrt{I_n+(O_nP_1)(O_nP_1)^t}\sqrt{I_n+(O_nP_2)(O_nP_2)^t} \right\} O_n
\\&
\hspace{-2.0cm}
\overbrace{=\!\!=\!\!=}^{(5)} \hspace{0.2cm}
\lgyr[O_nP_1,O_nP_2]O_n
\,.
 \end{split}
\end{equation}
Derivation of the numbered equalities in \eqref{spc03} follows:
\begin{enumerate}
\item \label{dphs1}
This equation follows from the third equation in \eqref{anim94}.
\item \label{dphs2}
Follows from Item \eqref{dphs1} by Lemma \ref{mgn1p1}, p.~\pageref{mgn1p1},
and Theorem \ref{dokrd}, p.~\pageref{dokrd}.
\item \label{dphs3}
Follows from Item \eqref{dphs2} by the linearity of $O_n$.
\item \label{dphs4}
Follows from Item \eqref{dphs3} by the obvious matrix identity
$O_nP_1P_2^t=(O_nP_1)(O_nP_2)^tO_n$, and from
Lemma \ref{mgn1}, p.~\pageref{mgn1}.
\item \label{dphs5}
Follows from Item \eqref{dphs4} by the linearity of $O_n$ and by the
third equation in \eqref{anim94}.
\end{enumerate}

The proof of the matrix identity \eqref{spc02} in \eqref{spc03s} below
is similar to the proof of the matrix identity \eqref{spc01}
in  \eqref{spc03}:
\begin{equation} \label{spc03s}
\begin{split}
\lgyr[P_1,P_2] O_m
~
&
\overbrace{=\!\!=\!\!=}^{(1)} \hspace{0.2cm}
(P_1^tP_2 + \sqrt{I_m+P_1^tP_1}\sqrt{I_m+P_2^tP_2})
\sqrt{I_m+(P_1\op P_2)^t(P_1\op P_2)}^{~-1} O_m
\\&
\hspace{-2.0cm}
\overbrace{=\!\!=\!\!=}^{(2)} \hspace{0.2cm}
(P_1^tP_2 + \sqrt{I_m+P_1^tP_1}\sqrt{I_m+P_2^tP_2}) O_m
\sqrt{I_m+(P_1O_m\op P_2O_m)^t(P_1O_m\op P_2O_m)}^{~-1}
\\&
\hspace{-2.0cm}
\overbrace{=\!\!=\!\!=}^{(3)} \hspace{0.2cm}
(P_1^tP_2O_m + \sqrt{I_m+P_1^tP_1}\sqrt{I_m+P_2^tP_2}O_m)
\\ &
\times
\sqrt{I_m+(P_1O_m\op P_2O_m)^t(P_1O_m\op P_2O_m)}^{~-1}
\\&
\hspace{-2.0cm}
\overbrace{=\!\!=\!\!=}^{(4)} \hspace{0.2cm}
\left\{ O_m(P_1O_m)^t(P_2O_m) + O_m
\sqrt{I_m+(P_1O_m)^t(P_1O_m)}\sqrt{I_m+(P_2O_m)^t(P_2O_m)} \right\}
\\ &
\times
\sqrt{I_m+(P_1O_m\op P_2O_m)^t(P_1O_m\op P_2O_m)}^{~-1}
\\&
\hspace{-2.0cm}
\overbrace{=\!\!=\!\!=}^{(5)} \hspace{0.2cm}
O_m\rgyr[P_1O_m,P_2O_m]
\,.
 \end{split}
\end{equation}
Derivation of the numbered equalities in \eqref{spc03s} follows:
\begin{enumerate}
\item \label{dphs1s}
This equation follows from the second equation in \eqref{anim94}.
\item \label{dphs2s}
Follows from Item \eqref{dphs1} by Lemma \ref{mgn1p1}, p.~\pageref{mgn1p1},
and Theorem \ref{dokrd}, p.\pageref{dokrd}.
\item \label{dphs3s}
Follows from Item \eqref{dphs2} by the linearity of $O_m$.
\item \label{dphs4s}
Follows from Item \eqref{dphs3} by the obvious matrix identity
$P_1^tP_2O_m=O_m(P_1O_m)^t(P_2O_m)$, and from
Lemma \ref{mgn1}, p.~\pageref{mgn1}, and from the linearity of $O_m$.
\item \label{dphs5s}
Follows from Item \eqref{dphs4s} by the linearity of $O_m$ and by the
second equation in \eqref{anim94}.
\end{enumerate}
\end{proof} 

The following corollary results immediately from Theorem \ref{spcomw}.
\begin{ccorollary}\label{vfsv}
Let $P_1,P_2\in\Rnm$, $O_n\in SO(n)$ and $O_m\in SO(m)$.
Then,
\begin{equation} \label{urfka}
\lgyr[O_nP_1,O_nP_2] = \lgyr[P_1,P_2]
\end{equation}
if and only if $O_n$ and $\lgyr[P_1,P_2]$ commute, that is,
$O_n\lgyr[P_1,P_2]=\lgyr[P_1,P_2]O_n$.

Similarly,
\begin{equation} \label{urfkb}
\rgyr[P_1O_m,P_2O_m] = \rgyr[P_1,P_2]
\end{equation}
if and only if $O_m$ and $\rgyr[P_1,P_2]$ commute, that is,
$O_m\rgyr[P_1,P_2]=\rgyr[P_1,P_2]O_m$.
\end{ccorollary}

\begin{eexample}\label{enyuch}
The left (right) gyration $\lgyr[P_1,P_2]$ ($\rgyr[P_1,P_2]$)
commutes with itself. Hence, by Corollary \ref{vfsv},
\begin{equation} \label{kvrcd}
\begin{split}
\lgyr[\lgyr[P_1,P_2]P_1,\lgyr[P_1,P_2]P_2] &= \lgyr[P_1,P_2]
\\
\rgyr[P_1\rgyr[P_1,P_2],P_2\rgyr[P_1,P_2]] &= \rgyr[P_1,P_2]
\,.
\end{split}
\end{equation}
\end{eexample}

Left gyrations are invariant under
parameter right rotations $O_m\in SO(m)$, and
right gyrations are invariant under
parameter left rotations $O_n\in SO(n)$, as the following theorem
asserts.

\begin{ttheorem}\label{spcowd}
{\bf (Bi-gyration Invariance Relation).}
\begin{equation} \label{bilva}
\lgyr[P_1O_m,P_2O_m] = \lgyr[P_1,P_2]
\end{equation}
\begin{equation} \label{bilvb}
\rgyr[O_nP_1,O_nP_2] = \rgyr[P_1,P_2]
\end{equation}
for all $P_1,P_2\in\Rnm$, $O_n\in SO(n)$ and $O_m\in SO(m)$.
\end{ttheorem}
\begin{proof}
The proof follows straightforwardly from the second and the third
equations in \eqref{anim94}, p.~\pageref{anim94}, and from
Theorem \ref{dokrd}, p.~\pageref{dokrd}, noting that
$(P_1O_m)(P_2O_m)^t=P_1P_2^t$ and
$(O_nP_1)^t(O_nP_2)=P_1^tP_2$
for all $P_1,P_2\in\Rnm$, $O_n\in SO(n)$ and $O_m\in SO(m)$.
\end{proof}

\section{Product of Lorentz Transformations}
\label{secc09}

Let $\Lambda_1$ and $\Lambda_2$ be two Lorentz transformations
of order $(m,n)$, $m,n\in\Nb$, so that, according to \eqref{anim67t},
\begin{equation} \label{anim96}
\begin{split}
\Lambda_1&=\Lambda(O_{n,1},P_1,O_{m,1})
=\rho(O_{m,1}) B(P_1) \lambda(O_{n,1})
\\[4pt]
\Lambda_2&=\Lambda(O_{n,2},P_2,O_{m,2})
=\rho(O_{m,2}) B(P_2) \lambda(O_{n,2})
\,.
\end{split}
\end{equation}

The product $\Lambda_1\Lambda_2$ of $\Lambda_1$ and $\Lambda_2$
is obtained in the following chain of equations, which are numbered
for subsequent explanation.
\begin{equation} \label{anim97}
\begin{split}
\Lambda_1\Lambda_2
&
\overbrace{=\!\!=\!\!=}^{(1)} \hspace{0.2cm}
\rho(O_{m,1}) B(P_1) \lambda(O_{n,1})
\rho(O_{m,2}) B(P_2) \lambda(O_{n,2})
\\&
\overbrace{=\!\!=\!\!=}^{(2)} \hspace{0.2cm}
\rho(O_{m,1}) B(P_1) \rho(O_{m,2}) \lambda(O_{n,1}) B(P_2) \lambda(O_{n,2})
\\&
\overbrace{=\!\!=\!\!=}^{(3)} \hspace{0.2cm}
\rho(O_{m,1}) \rho(O_{m,2}) B(P_1O_{m,2})B(O_{n,1}P_2)\lambda(O_{n,1})
\lambda(O_{n,2})
\\&
\overbrace{=\!\!=\!\!=}^{(4)} \hspace{0.2cm}
\rho(O_{m,1}O_{m,2}) B(P_1O_{m,2})B(O_{n,1}P_2)\lambda(O_{n,1}O_{n,2})
\\&
\overbrace{=\!\!=\!\!=}^{(5)} \hspace{0.2cm}
\rho(O_{m,1}O_{m,2})
\\ & \hspace{0.8cm} \times
\rho(\rgyr[P_1O_{m,2}, O_{n,1} P_2])
B(P_1O_{m,2} \op O_{n,1} P_2)
\lambda(\lgyr[P_1O_{m,2}, O_{n,1} P_2])
\\ & \hspace{0.8cm} \times
\lambda(O_{n,1}O_{n,2})
\\&
\overbrace{=\!\!=\!\!=}^{(6)} \hspace{0.2cm}
\rho(O_{m,1}O_{m,2}\rgyr[P_1O_{m,2}, O_{n,1} P_2])
\\ & \hspace{0.8cm} \times
B(P_1O_{m,2} \op O_{n,1} P_2)
\\ & \hspace{0.8cm} \times
\lambda(\lgyr[P_1O_{m,2}, O_{n,1} P_2]O_{n,1}O_{n,2})
\,.
 \end{split}
\end{equation}
Derivation of the numbered equalities in \eqref{anim97} follows:
\begin{enumerate}
\item \label{madods1}
This equation follows from \eqref{anim96}.
\item \label{madods2}
Follows from \eqref{madods1} since $\lambda(O_{n,1})$ and
$\rho(O_{m,2})$ commute.
\item \label{madods3}
Follows from \eqref{madods2} by Lemma \ref{mgn3}, p.~\pageref{mgn3}.
\item \label{madods4}
Follows from \eqref{madods3} by the obvious matrix identities
$\rho(O_{m,1})\rho(O_{m,2})=\rho(O_{m,1}O_{m,2})$
and
$\lambda(O_{n,1})\lambda(O_{n,2})=\lambda(O_{n,1}O_{n,2})$.
\item \label{madods5}
Follows from \eqref{madods4} by the bi-boost composition law
\eqref{anim93}, p.~\pageref{anim93}.
\item \label{madods6}
Obvious (Similar to the argument in Item \eqref{madods4}).
\end{enumerate}

In the column notation \eqref{anim68},
the result of \eqref{anim97} gives the product law
of Lorentz transformations in the following theorem.

\begin{ttheorem}\label{dokrt}
{\bf (Lorentz Transformation Product Law).}
The product of two Lorentz transformations
$\Lambda_1=(P_1,O_{n,1},O_{m,1})^t$
and $\Lambda_2=(P_2,O_{n,2},O_{m,2})^t$
of order $(m,n)$, $m,n\in\Nb$, is given by
\begin{equation} \label{anim98}
\Lambda_1 \Lambda_2 =
\begin{pmatrix}  P_1 \\[4pt] O_{n,1} \\[4pt] O_{m,1} \end{pmatrix}
\begin{pmatrix}  P_2 \\[4pt] O_{n,2} \\[4pt] O_{m,2} \end{pmatrix}
= \begin{pmatrix}
P_1 O_{m,2} \op O_{n,1} P_2 \\[4pt]
\lgyr[P_1 O_{m,2},O_{n,1} P_2] O_{n,1} O_{n,2} \\[4pt]
O_{m,1} O_{m,2} \rgyr[P_1 O_{m,2},O_{n,1} P_2]
\end{pmatrix}
\,.
\end{equation}
\end{ttheorem}

\begin{eexample}\label{exmp1}
{\it
In the special case when $P_1=P_2=0_{n,m}$ and $O_{m,1}=O_{m,2}=I_m$,
the parameter composition law \eqref{anim98} yields the equation
\begin{equation} \label{aqua01}
\begin{pmatrix}  0_{n,m} \\[4pt] O_{n,1} \\[4pt] I_m \end{pmatrix}
\begin{pmatrix}  0_{n,m} \\[4pt] O_{n,2} \\[4pt] I_m \end{pmatrix}
=
\begin{pmatrix}  0_{n,m} \\[4pt] O_{n,1}O_{n,2} \\[4pt] I_m \end{pmatrix}
\end{equation}
demonstrating that under the parameter composition law \eqref{anim98}
the parameter $O_n$ forms the spacial orthogonal group $SO(n)$.
}
\end{eexample}
\begin{eexample}\label{exmp2}
{\it
In the special case when $P_1=P_2=0_{n,m}$ and $O_{n,1}=O_{n,2}=I_n$,
the parameter composition law \eqref{anim98} yields the equation
\begin{equation} \label{aqua02}
\begin{pmatrix}  0_{n,m} \\[4pt] I_n \\[4pt] O_{m,1} \end{pmatrix}
\begin{pmatrix}  0_{n,m} \\[4pt] I_n \\[4pt] O_{m,2} \end{pmatrix}
=
\begin{pmatrix}  0_{n,m} \\[4pt] I_n \\[4pt] O_{m,1}O_{m,2} \end{pmatrix}
\end{equation}
demonstrating that under the parameter composition law \eqref{anim98}
the parameter $O_m$ forms the spacial orthogonal group $SO(m)$.
}
\end{eexample}
\begin{eexample}\label{exmp3}
{\it
In the special case when $O_{n,1}=O_{n,2}=I_n$ and $O_{m,1}=O_{m,2}=I_m$
the parameter composition law \eqref{anim98} yields the equation
\begin{equation} \label{aqua03}
\begin{pmatrix} P_1 \\[4pt] I_n \\[4pt] I_m \end{pmatrix}
\begin{pmatrix} P_2 \\[4pt] I_n \\[4pt] I_m \end{pmatrix}
=
\begin{pmatrix} P_1\op P_2\\[4pt]\lgyr[P_1,P_2]\\[4pt]\rgyr[P_1,P_2]
\end{pmatrix}
\,.
\end{equation}
Clearly, under the parameter composition law \eqref{anim98} the
parameter $P\in\Rnm$ does not form a group.
Indeed, following the parametrization of the generalized
Lorentz group $SO(m,n)$ in \eqref{anim96}, we face the task of determining
the composition law of the parameter $P\in\Rnm$ along with the resulting
group-like structure of the parameter set $\Rnm$.
We will find in the sequel that the group-like structure of $\Rnm$
that results from the composition law of the parameter $P$ is
a natural generalization of the gyrocommutative gyrogroup structure,
called a bi-gyrocommutative bi-gyrogroup.
}
\end{eexample}

The Lorentz transformation product \eqref{anim98} represents
matrix multiplication. As such, it is associative and, clearly,
its inverse obeys the identity
\begin{equation} \label{anim99}
(\Lambda_1 \Lambda_2)^{-1} = \Lambda_2^{-1} \Lambda_1^{-1}
\,.
\end{equation}

\section{The Bi-gyrocommutative Law}
\label{secc10}

Bi-boosts are Lorentz transformations without bi-rotations. Let
\begin{equation} \label{anim104}
B(P_k) = (P_k,I_n,I_m)^t
\,,
\end{equation}
$P_k\in\Rnm$,
$k=1,2$, be two bi-boosts in $\Rb^{(m+n)\times(m+n)}$.
Then, by \eqref{anim98} with $O_{n,1}=O_{n,2}=I_n$ and
$O_{m,1}=O_{m,2}=I_m$ (or by \eqref{anim95}), and by \eqref{anim83}
with $O_n=\lgyr[P_1,P_2]$ and $O_m=\rgyr[P_1,P_2]$,
\begin{equation} \label{anim105}
\begin{split}
(B(P_1)B(P_2))^{-1} &=
\left\{
\begin{pmatrix} P_1 \\[4pt] I_n \\[4pt] I_m \end{pmatrix}
\begin{pmatrix} P_2 \\[4pt] I_n \\[4pt] I_m \end{pmatrix}
\right\}^{-1}
=
\begin{pmatrix} P_1 \op P_2 \\[4pt] \lgyr[P_1,P_2] \\[4pt] \rgyr[P_1,P_2]
\end{pmatrix}^{-1}
\\[12pt] &=
\begin{pmatrix}
-\lgyr^{-1}[P_1,P_2] (P_1\op P_2) \rgyr^{-1}[P_1,P_2] \\[4pt]
\lgyr^{-1}[P_1,P_2] \\[4pt] \rgyr^{-1}[P_1,P_2]
\end{pmatrix}
\,.
\end{split}
\end{equation}
Here $\lgyr^{-1}[P_1,P_2]=(\lgyr[P_1,P_2])^{-1}$ and, similarly,
$\rgyr^{-1}[P_1,P_2]=(\rgyr[P_1,P_2])^{-1}$.

Calculating $(B(P_1)B(P_2))^{-1}$ in a different way, as indicated
in \eqref{anim99}, yields
\begin{equation} \label{anim106}
\begin{split}
(B(P_1)B(P_2))^{-1} = B(P_2)^{-1} B(P_1)^{-1}
&=
\begin{pmatrix} P_2 \\[4pt] I_n \\[4pt] I_m \end{pmatrix}^{-1}
\begin{pmatrix} P_1 \\[4pt] I_n \\[4pt] I_m \end{pmatrix}^{-1}
=
\begin{pmatrix}-P_2 \\[4pt] I_n \\[4pt] I_m \end{pmatrix}
\begin{pmatrix}-P_1 \\[4pt] I_n \\[4pt] I_m \end{pmatrix}
\\[12pt] &=
\begin{pmatrix}
(-P_2) \op (-P_1) \\[4pt] \lgyr[-P_2,-P_1] \\[4pt] \rgyr[-P_2,-P_1]
\end{pmatrix}
\,.
\end{split}
\end{equation}

Hence, the extreme right sides of \eqref{anim105}\,--\,\eqref{anim106}
are equal, implying the equality of their respective entries, giving
rise to the three equations in \eqref{anim107}\,--\,\eqref{anim108} below.

The second and third entries of the extreme right sides of
\eqref{anim105}\,--\,\eqref{anim106}, along with
the even property \eqref{anim94p2} of bi-gyrations,
imply the {\it bi-gyration inversion law},
\begin{equation} \label{anim107}
\begin{split}
\lgyr^{-1}[P_1,P_2] &= \lgyr [-P_2,-P_1] = \lgyr [P_2,P_1]
\\
\rgyr^{-1}[P_1,P_2] &= \rgyr [-P_2,-P_1] = \rgyr [P_2,P_1]
\,.
\end{split}
\end{equation}
for all $P_1,P_2\in\Rnm$.

The first entry of the extreme right sides of
\eqref{anim105}\,--\,\eqref{anim106},
along with \eqref{anim107} and
the gyroautomorphic inverse property \eqref{anim94p1}, yields
\begin{equation} \label{anim108}
\begin{split}
(-P_2) \op (-P_1) &=-\lgyr^{-1}[P_1,P_2] (P_1\op P_2) \rgyr^{-1}[P_1,P_2]
\\ &=
-\lgyr[-P_2,-P_1] (P_1\op P_2) \rgyr[-P_2,-P_1]
\\ &=
\lgyr[-P_2,-P_1] \{-(P_1\op P_2)\} \rgyr[-P_2,-P_1]
\\ &=
\lgyr[-P_2,-P_1] \{(-P_1)\op(-P_2)\} \rgyr[-P_2,-P_1]
\,,
\end{split}
\end{equation}
for all $P_1,P_2\in\Rnm$.

Renaming $-P_1$ and $-P_2$ as $P_2$ and $P_1$, the extreme sides of
\eqref{anim108} give the
{\it bi-gyrocommutative law} of the bi-gyroaddition $\op$,
\begin{equation} \label{anim109}
P_1 \op P_2 = \lgyr[P_1,P_2] (P_2\op P_1) \rgyr[P_1,P_2]
\,,
\end{equation}
for all $P_1,P_2\in\Rnm$.

Instructively, a short derivation of the bi-gyrocommutative law
of $\op$ is presented below.
Transposing the extreme sides of the fourth matrix equation in
\eqref{anim94}, noting that by \eqref{anim107}
\begin{equation} \label{maday}
\begin{split}
\lgyr[P_1,P_2]^t &= \lgyr^{-1}[P_1,P_2] = \lgyr[P_2,P_1]
\\
\rgyr[P_1,P_2]^t &= \rgyr^{-1}[P_1,P_2] = \rgyr[P_2,P_1]
\,,
\end{split}
\end{equation}
and renaming the pair $(P_1,P_2)$ as $(P_2,P_1)$,
we obtain the matrix identity
\begin{equation} \label{dvon}
P_1\op P_2 = \lgyr[P_1,P_2](P_2 \op P_1) \rgyr[P_1,P_2]
\,.
\end{equation}
for all $P_1,P_2\in\Rnm$.

The matrix identity \eqref{dvon} gives the {\it bi-gyrocommutative law}
of the binary operation $\op$ in $\Rnm$, according to which
$P_1\op P_2$ equals $P_2\op P_1$ bi-gyrated by the
bi-gyration $(\lgyr[P_1,P_2],\rgyr[P_1,P_2])$ generated by $P_1$ and $P_2$,
for all $P_1,P_2\in\Rnm$.

Formalizing the result in \eqref{anim109} and
in \eqref{dvon} we obtain the following theorem.
\begin{ttheorem}\label{tmkdrh} 
{\bf (Bi-gyrocommutative Law).}
The binary operation $\op$ in $\Rnm$ possesses the
bi-gyrocommutative law
\begin{equation} \label{anim109s}
P_1 \op P_2 = \lgyr[P_1,P_2] (P_2\op P_1) \rgyr[P_1,P_2]
\end{equation}
for all $P_1,P_2\in\Rnm$.
\end{ttheorem}

When $m=1$ right gyrations are trivial, $\rgyr[P_1,P_2]=I_m$. Hence,
in the special case when $m=1$, the bi-gyrocommutative law
\eqref{anim109s} of bi-gyrogroup theory descends to the gyrocommutative law
of gyrogroup theory found, for instance, in
\cite{mybook01,mybook02,mybook03,mybook04,mybook05,mybook06,mybook07}.

Formalizing the results in \eqref{maday} we obtain the following theorem.
\begin{ttheorem}\label{tmdkert} 
{\bf (Bi-gyration Inversion Law).}
The bi-gyrogroupoid $(\Rnm,\op)$ possesses the left gyration inversion law,
\begin{subequations} \label{madays}
\begin{equation} \label{madaysa}
\lgyr^{-1}[P_1,P_2] = \lgyr[P_2,P_1]
\,,
\end{equation}
and the right gyration inversion law,
\begin{equation} \label{madaysb}
\rgyr^{-1}[P_1,P_2] = \rgyr[P_2,P_1]
\,,
\end{equation}
\end{subequations}
for all $P_1,P_2\in\Rnm$.
\end{ttheorem}

Identities \eqref{madaysa}\,--\,\eqref{madaysb} express the
{\it inversive symmetric property} of bi-gyrations.

\section{The Bi-gyroassociative Law}
\label{seccbd}

Matrix multiplication is associative. Hence
\begin{equation} \label{anim112}
(\Lambda_1 \Lambda_2) \Lambda_3
=
\Lambda_1 (\Lambda_2 \Lambda_3)
\,.
\end{equation}

On the one hand,
by \eqref{anim95} and \eqref{anim98},
\begin{equation} \label{anim113}
\begin{split}
(B(P_1)B(P_2))B(P_3)
&=
\left\{
\begin{pmatrix} P_1 \\[4pt] I_n \\[4pt] I_m \end{pmatrix}
\begin{pmatrix} P_2 \\[4pt] I_n \\[4pt] I_m \end{pmatrix}
\right\}
\begin{pmatrix} P_3 \\[4pt] I_n \\[4pt] I_m \end{pmatrix}
=
\begin{pmatrix} P_1\op P_2 \\[4pt] \lgyr[P_1,P_2]
               \\[4pt] \rgyr[P_1,P_2] \end{pmatrix}
\begin{pmatrix} P_3 \\[4pt] I_n \\[4pt] I_m \end{pmatrix}
\\[12pt] &=
\begin{pmatrix} (P_1\op P_2) \op \lgyr[P_1,P_2] P_3 \\[4pt]
\lgyr[P_1\op P_2,\lgyr[P_1,P_2]P_3] \lgyr[P_1,P_2] \\[4pt]
\rgyr[P_1,P_2]\rgyr[P_1\op P_2,\lgyr[P_1,P_2]P_3]
\end{pmatrix}
\,.
\end{split}
\end{equation}

On the other hand, similarly, by \eqref{anim95} and \eqref{anim98},
\begin{equation} \label{anim114}
\begin{split}
B(P_1)(B(P_2)B(P_3))
&=
\begin{pmatrix} P_1 \\[4pt] I_n \\[4pt] I_m \end{pmatrix}
\left\{
\begin{pmatrix} P_2 \\[4pt] I_n \\[4pt] I_m \end{pmatrix}
\begin{pmatrix} P_3 \\[4pt] I_n \\[4pt] I_m \end{pmatrix}
\right\}
=
\begin{pmatrix} P_1 \\[4pt] I_n \\[4pt] I_m \end{pmatrix}
\begin{pmatrix} P_2\op P_3 \\[4pt] \lgyr[P_2,P_3]
               \\[4pt] \rgyr[P_2,P_3] \end{pmatrix}
\\[12pt] &=
\begin{pmatrix} P_1 \rgyr[P_2,P_3] \op (P_2 \op P_3) \\[4pt]
\lgyr[P_1\rgyr[P_2,P_3],P_2\op P_3] \lgyr[P_2,P_3] \\[4pt]
\rgyr[P_2,P_3] \rgyr[P_1\rgyr[P_2,P_3],P_2\op P_3]
\end{pmatrix}
\,.
\end{split}
\end{equation}

Hence, by \eqref{anim112}\,--\,\eqref{anim114},
corresponding entries of the extreme right sides of
\eqref{anim113} and \eqref{anim114} are equal, giving rise to the
{\it bi-gyroassociative law}
\begin{equation} \label{biassoc}
(P_1\op P_2) \op \lgyr[P_1,P_2] P_3
=
P_1 \rgyr[P_2,P_3] \op (P_2 \op P_3)
\end{equation}
and to the bi-gyration identities
\begin{equation} \label{anim115}
\begin{split}
\lgyr[P_1\op P_2,\lgyr[P_1,P_2]P_3] \lgyr[P_1,P_2]
&=
\lgyr[P_1\rgyr[P_2,P_3],P_2\op P_3] \lgyr[P_2,P_3]
\\[8pt]
\rgyr[P_1,P_2]\rgyr[P_1\op P_2,\lgyr[P_1,P_2]P_3]
&=
\rgyr[P_2,P_3] \rgyr[P_1\rgyr[P_2,P_3],P_2\op P_3]
\end{split}
\end{equation}
for all $P_1,P_2,P_3\in\Rnm$.

Formalizing the result in \eqref{biassoc} we obtain the following theorem.
\begin{ttheorem}\label{tmkdrk}
{\bf (Bi-gyroassociative Law in $(\Rnm,\op)$).}
The bi-gyroaddition $\op$ in $\Rnm$ possesses the
bi-gyroassociative law
\begin{equation} \label{biassocs}
(P_1\op P_2) \op \lgyr[P_1,P_2] P_3
=
P_1 \rgyr[P_2,P_3] \op (P_2 \op P_3)
\end{equation}
for all $P_1,P_2\in\Rnm$.
\end{ttheorem}

Note that in the bi-gyroassociative law \eqref{biassocs},
$P_1$ and $P_2$ are grouped together on the left side, while
$P_2$ and $P_3$ are grouped together on the right side.

When $m=1$ right gyrations are trivial, $\rgyr[P_1,P_2]=I_{m=1}=1$. Hence,
in the special case when $m=1$, the bi-gyroassociative law
\eqref{biassocs} descends to the gyroassociative law
of gyrogroup theory found, for instance, in
\cite{mybook01,mybook02,mybook03,mybook04,mybook05,mybook06,mybook07}.

The bi-gyroassociative law gives rise to the left and right cancellation
laws in the following theorem.
\begin{ttheorem}\label{rtmdc} 
{\bf (Left and Right Cancellation Laws in $(\Rnm,\op)$).}
The bi-gyrogroupoid $(\Rnm,\op)$ possesses the left and right
cancellation laws
\begin{equation} \label{rtmdc1a}
P_2 = \om P_1\rgyr[P_1,P_2] \op (P_1\op P_2)
\end{equation}
and
\begin{equation} \label{rtmdc1b}
P_1 = (P_1\op P_2) \om \lgyr[P_1,P_2]P_2
\end{equation}
for all $P_1,P_2\in\Rnm$.
\end{ttheorem}
\begin{proof}
The left cancellation law \eqref{rtmdc1a}
follows from the bi-gyroassociative law
\eqref{biassocs} with $P_1=\om P_2$, noting that $\lgyr[\om P_2,P_2]$
is trivial by \eqref{anim95p3}, p.~\pageref{anim95p3}.
The right cancellation law \eqref{rtmdc1b}
follows from the bi-gyroassociative law
\eqref{biassocs} with $P_3=\om P_2$, noting that  $\rgyr[P_2,\om P_2]$
is trivial.
\end{proof} 

The bi-gyroassociative law gives rise to the left and right
bi-gyroassociative laws in the following theorem.
\begin{ttheorem}\label{tarifk} 
{\bf (Left and Right Bi-gyroassociative Law in $(\Rnm,\op)$).}
The bi-gyroaddition $\op$ in $\Rnm$ possesses the
left bi-gyroassociative law
\begin{equation} \label{biassocsl}
P_1\op (P_2 \op P_3)
=
(P_1 \rgyr[P_3,P_2] \op P_2) \op \lgyr[P_1\rgyr[P_3,P_2],P_2]P_3
\end{equation}
and the right bi-gyroassociative law
\begin{equation} \label{biassocsr}
(P_1\op P_2) \op P_3
=
P_1 \rgyr[P_2,\lgyr[P_2,P_1]P_3] \op (P_2 \op \lgyr[P_2,P_1]P_3)
\end{equation}
for all $P_1,P_2\in\Rnm$.
\end{ttheorem}
\begin{proof}
The left bi-gyroassociative law \eqref{biassocsl} is obtained from
the bi-gyroassociative law \eqref{biassocs} by replacing
$P_1$ by $P_1\rgyr[P_3,P_2]$ and noting the
bi-gyration inversion law \eqref{anim107}.
Similarly, the
right bi-gyroassociative law \eqref{biassocsr} is obtained from
the bi-gyroassociative law \eqref{biassocs} by replacing
$P_3$ by $\lgyr[P_2,P_1]P_3$ and noting the
bi-gyration inversion law \eqref{anim107}.
\end{proof} 

\section{Bi-gyration Reduction Properties}
\label{seccbe}

A reduction property of a gyration $\lgyr[P_1,P_2]$ or $rgyr[P_1,P_2]$
is a property enabling the gyration to be expressed as a gyration
that involves $P_1\op P_2$.
Several reduction properties are derived in
Subsects.~\ref{seccbe01}\,--\,\ref{seccbe04} below.

\subsection{Bi-gyration Reduction Properties I}
\label{seccbe01}

When $P_3=\om P_2$, \eqref{anim115} specializes to
\begin{equation} \label{akor01}
\begin{split}
\lgyr[P_1\op P_2,\om\lgyr[P_1,P_2]P_2]\lgyr[P_1,P_2] &=I_n
\\[4pt]
\rgyr[P_1,P_2]\rgyr[P_1\op P_2,\om\lgyr[P_1,P_2]P_2] &=I_m
\end{split}
\end{equation}
or, equivalently by bi-gyration inversion, \eqref{anim107},
\begin{equation} \label{akor02}
\begin{split}
\lgyr[P_1,P_2] &= \lgyr[\om\lgyr[P_1,P_2]P_2,P_1\op P_2]
\\[4pt] 
\rgyr[P_1,P_2] &= \rgyr[\om\lgyr[P_1,P_2]P_2,P_1\op P_2]
\,.
\end{split}
\end{equation}
Similarly, when $P_2=\om P_1$, \eqref{anim115} specializes to
\begin{equation} \label{akor03}
\begin{split}
I_n &= \lgyr[P_1\rgyr[\om P_1,P_3],\om P_1\op P_3]\lgyr[\om P_1,P_2]
\\[4pt]
I_m &= \rgyr[\om P_1,P_3]\rgyr[P_1\rgyr[\om P_1,P_3],\om P_1\op P_3]
\end{split}
\end{equation}
or, equivalently by bi-gyration inversion, \eqref{anim107},
and renaming $P_3$ as $\om P_2$,
\begin{equation} \label{akor04}
\begin{split}
\lgyr[P_1,P_2] &= \lgyr[P_1\op P_2,\om P_1\rgyr[P_1,P_2]]
\\[4pt]
\rgyr[P_1,P_2] &= \rgyr[P_1\op P_2,\om P_1\rgyr[P_1,P_2]]
\,.
\end{split}
\end{equation}

Formalizing the results in \eqref{akor02} and \eqref{akor04} we obtain
the following theorem.
\begin{ttheorem}\label{akorp}
{\bf (Left and Right Gyration Reduction Properties).}
\begin{equation} \label{akor02s}
\begin{split}
\lgyr[P_1,P_2] &= \lgyr[\om\lgyr[P_1,P_2]P_2,P_1\op P_2]
\\[4pt]
\rgyr[P_1,P_2] &= \rgyr[\om\lgyr[P_1,P_2]P_2,P_1\op P_2]
\,.
\end{split}
\end{equation}
and
\begin{equation} \label{akor04s}
\begin{split}
\lgyr[P_1,P_2] &= \lgyr[P_1\op P_2,\om P_1\rgyr[P_1,P_2]]
\\[4pt]
\rgyr[P_1,P_2] &= \rgyr[P_1\op P_2,\om P_1\rgyr[P_1,P_2]]
\,.
\end{split}
\end{equation}
for all $P_1,P_2\in\Rnm$.
\end{ttheorem}

\subsection{Bi-gyration Reduction Properties II}
\label{seccbe02}

In general, the product of bi-boosts in a pseudo-orthogonal
group $SO(m,n)$ is a Lorentz transformation which is not a boost.
In some special cases, however, the product of bi-boosts is again
a bi-boost, as shown below.

Let $P_1,P_2\in\Rnm$, and let $J(P_1,P_2)$ be the
bi-boost symmetric product
\begin{equation} \label{imad02}
J(P_1,P_2) = B(P_1) B(P_2) B(P_1) =
\begin{pmatrix} P_1 \\[4pt] I_n \\[4pt] I_m \end{pmatrix}
\begin{pmatrix} P_2 \\[4pt] I_n \\[4pt] I_m \end{pmatrix}
\begin{pmatrix} P_1 \\[4pt] I_n \\[4pt] I_m \end{pmatrix}
\,,
\end{equation}
which is symmetric with respect to the central bi-boost factor
$(P_2,I_n,I_m)^t$.
Then, by the Lorentz product law \eqref{anim98},
\begin{equation} \label{imad03}
\begin{split}
J(P_1,P_2) &=
\begin{pmatrix} P_1\op P_2 \\[4pt] \lgyr[P_1,P_2]
               \\[4pt] \rgyr[P_1,P_2] \end{pmatrix}
\begin{pmatrix} P_1 \\[4pt] I_n \\[4pt] I_m \end{pmatrix}
\\[6pt] &=
\begin{pmatrix} (P_1\op P_2) \op \lgyr[P_1,P_2]P_1 \\[4pt]
\lgyr[P_1\op P_2, \lgyr[P_1,P_2]P_1] \lgyr[P_1,P_2] \\[4pt]
\rgyr[P_1,P_2]\rgyr[P_1\op P_2, \lgyr[P_1,P_2]P_1] \end{pmatrix}
=: \begin{pmatrix} P_3 \\[4pt] O_n \\[4pt] O_m \end{pmatrix}
\,.
\end{split}
\end{equation}

By means of \eqref{anim78}, p.~\pageref{anim78}, it is clear that
\begin{equation} \label{imad04}
J(P_1,P_2)^{-1} = J(-P_1,-P_2)
\,.
\end{equation}
Hence, by the
gyroautomorphic inverse property \eqref{anim94p1}, p.~\pageref{anim94p1},
and by the bi-gyration even property, \eqref{anim94p2}, p.~\pageref{anim94p2},
it is clear from \eqref{imad03} that
\begin{equation} \label{imad05}
J(P_1,P_2)^{-1} = J(-P_1,-P_2) =
\begin{pmatrix} -P_3 \\[4pt] O_n \\[4pt] O_m \end{pmatrix}
\,.
\end{equation}

But, it follows from the inverse Lorentz transformation
\eqref{anim83}, p.~\pageref{anim83}, that
\begin{equation} \label{imad06}
J(P_1,P_2)^{-1} =
\begin{pmatrix} -O_n^{-1} P_3 O_m^{-1} \\[4pt] O_n^{-1} \\[4pt] O_m^{-1}
\end{pmatrix}
\,.
\end{equation}

Comparing the right sides of \eqref{imad06} and \eqref{imad05},
we find that $O_n=O_n^{-1}$ and $O_m=O_m^{-1}$, implying
$O_n=I_n$ and $O_m=I_m$.
Hence, the bi-boost product $J(P_1,P_2)$ is, again, a bi-boost,
so that by \eqref{imad03},
\begin{equation} \label{imad07}
J(P_1,P_2) =
\begin{pmatrix} (P_1\op P_2) \op \lgyr[P_1,P_2]P_1 \\[4pt]
 I_n \\[4pt] I_m \end{pmatrix}
\,.
\end{equation}

Following \eqref{imad07} and \eqref{imad03} we have the
bi-gyration identities
\begin{equation} \label{imad08}
\begin{split}
\lgyr[P_1\op P_2, \lgyr[P_1,P_2]P_1] \lgyr[P_1,P_2] &= I_n
\\[4pt]
\rgyr[P_1,P_2]\rgyr[P_1\op P_2, \lgyr[P_1,P_2]P_1] &= I_m
\,,
\end{split}
\end{equation}
implying, by the bi-gyration inversion law \eqref{madays},
\begin{equation} \label{imad0r}
\begin{split}
\lgyr[P_1,P_2] &= \lgyr[\lgyr[P_1,P_2]P_1,P_1\op P_2] 
\\[4pt]
\rgyr[P_1,P_2] &= \rgyr[\lgyr[P_1,P_2]P_1,P_1\op P_2]
\,,
\end{split}
\end{equation}
for all $P_1,P_2\in\Rnm$.

The results in \eqref{imad07}\,--\,\eqref{imad08} can readily
be extended to the symmetric product of any number of bi-boosts that appear
symmetrically with respect to a central factor. Thus, for instance,
the symmetric bi-boost product $J$,
\begin{equation} \label{imad09}
J=B(P_k)B(P_{k-1}) \ldots B(P_2)B(P_1)B(P_0)B(P_1)B(P_2) \ldots B(P_{k-1})
B(P_k)
\,,
\end{equation}
is symmetric with respect to the central factor $B(P_0)$,
for any $k\in\Nb$, and all $P_i\in\Rnm$, $i=0,1,2,\ldots, k$.
In particular, the bi-boost product $J$ in \eqref{imad09} is, again,
a bi-boost.

We now manipulate the first bi-gyration identity in \eqref{imad08} into
an elegant form that will be elevated to the status of a
theorem in Theorem \ref{tmjkdf} below.
Let us consider the following chain of equations, which are numbered
for subsequent explanation.
\begin{equation} \label{mad10}
\begin{split}
I_n ~
&\overbrace{=\!\!=\!\!=}^{(1)} \hspace{0.2cm}
\lgyr[P_1\op P_2, \lgyr[P_1,P_2]P_1] \lgyr[P_1,P_2]
\\
&\overbrace{=\!\!=\!\!=}^{(2)} \hspace{0.2cm}
\lgyr[P_1,P_2]\lgyr[\lgyr[P_2,P_1](P_1\op P_2),P_1]
\\
&\overbrace{=\!\!=\!\!=}^{(3)} \hspace{0.2cm}
\lgyr[P_1,P_2]\lgyr[\lgyr[P_2,P_1](P_1\op P_2)\rgyr[P_2,P_1],P_1\rgyr[P_2,P_1]]
\\
&\overbrace{=\!\!=\!\!=}^{(4)} \hspace{0.2cm}
\lgyr[P_1,P_2]\lgyr[P_2\op P_1,P_1\rgyr[P_2,P_1]]
\,.
\end{split}
\end{equation}
Derivation of the numbered equalities in \eqref{mad10} follows:
\begin{enumerate}
\item \label{rdmbf01}
This equation is the first equation in the first bi-gyration identity
in \eqref{imad08}.
\item \label{rdmbf02}
Follows from the commuting relation \eqref{spc01}, p.~\pageref{spc01},
with $O_n=\lgyr[P_1,P_2]$,
noting that $\lgyr[P_2,P_1]=\lgyr^{-1}[P_1,P_2]$.
\item \label{rdmbf03}
Follows from the bi-gyration invariance relation
\eqref{bilva}, p.~\pageref{bilva}.
\item \label{rdmbf04}
Follows from the bi-gyrocommutative law
\eqref{anim109s}, p.~\pageref{anim109s}.
\end{enumerate}

By \eqref{mad10} and the bi-gyration inversion law \eqref{madays},
\begin{equation} \label{dmgdc1}
\lgyr[P_2,P_1] = \lgyr[P_2\op P_1,P_1\rgyr[P_2,P_1]]
\end{equation}

Renaming $(P_1,P_2)$ in \eqref{dmgdc1} as $(P_2,P_1)$, we obtain
the first identity in the following theorem.
\begin{ttheorem}\label{tmjkdf} 
{\bf (Left Gyration Reduction Properties).}
\begin{equation} \label{dmgdc1sa}
\lgyr[P_1,P_2] = \lgyr[P_1\op P_2,P_2\rgyr[P_1,P_2]]
\end{equation}
and
\begin{equation} \label{dmgdc1sb}
\lgyr[P_1,P_2] = \lgyr[P_1\rgyr[P_2,P_1],P_2\op P_1]
\end{equation}
for all $P_1,P_2\in\Rnm$.
\end{ttheorem}
\begin{proof}
The bi-gyration identity \eqref{dmgdc1sa} is identical with
\eqref{dmgdc1}.
The bi-gyration identity \eqref{dmgdc1sb} is obtained from \eqref{dmgdc1sa}
by applying the bi-gyration inversion law \eqref{madays} followed by
renaming $(P_1,P_2)$ as $(P_2,P_1)$.
\end{proof}

When $m=1$ right gyrations are trivial, $\rgyr[P_1,P_2]=I_{m=1}=1$. Hence,
in the special case when $m=1$, the bi-gyration reduction properties
\eqref{dmgdc1sa}\,--\,\eqref{dmgdc1sb} descend
to the gyration properties of gyrogroup theory found, for instance,
in \cite{mybook02}.

The bi-gyration identity \eqref{dmgdc1sb} involves both left and right
gyrations. We manipulate it into an identity that involves only
left gyrations in
the following chain of equations, which are numbered
for subsequent explanation.
\begin{equation} \label{akhak01z}
\begin{split}
\lgyr[P_1,P_2] ~
&\overbrace{=\!\!=\!\!=}^{(1)} \hspace{0.2cm}
\lgyr[P_1\rgyr[P_2,P_1],P_2\op P_1]
\\
&\overbrace{=\!\!=\!\!=}^{(2)} \hspace{0.2cm}
\lgyr[\lgyr[P_1,P_2]P_1\rgyr[P_2,P_1],\lgyr[P_1,P_2](P_2\op P_1)]
\\
&\hspace{-1.8cm} \overbrace{=\!\!=\!\!=}^{(3)} \hspace{0.2cm}
\lgyr[\lgyr[P_1,P_2]P_1\rgyr[P_2,P_1]\rgyr[P_1,P_2],
\lgyr[P_1,P_2](P_2\op P_1)\rgyr[P_1,P_2]]
\\
&\overbrace{=\!\!=\!\!=}^{(4)} \hspace{0.2cm}
\lgyr[\lgyr[P_1,P_2]P_1,P_1\op P_2]
\end{split}
\end{equation}
Derivation of the numbered equalities in \eqref{akhak01z} follows:
\begin{enumerate}
\item \label{vdhgr01z}
This equation is the bi-gyration identity \ref{dmgdc1sb}.
\item \label{vdhgr02z}
Follows from Result \eqref{urfka} of Corollary \ref{vfsv}, p.~\pageref{vfsv},
noting that, by Item \eqref{vdhgr01z}, the left gyrations
$\lgyr[P_1\rgyr[P_2,P_1],P_2\op P_1]$ and $\lgyr[P_1,P_2]$ commute since
they are equal.
\item \label{vdhgr03z}
Follows from \eqref{bilva}, p.~\pageref{bilvb}.
\item \label{vdhgr04z}
Follows from Item \eqref{vdhgr03z} by applying both the
bi-gyration inversion law \eqref{madays} and the
bi-gyrocommutative law \eqref{anim109s}, p.~\pageref{anim109s}.
\end{enumerate}

By means of the bi-gyration inversion law \eqref{madays},
the second bi-gyration identity in \eqref{imad08} gives rise to the
bi-gyration identity
\begin{equation} \label{edok01}
\rgyr[P_1,P_2] = \rgyr[\lgyr[P_1,P_2]P_1,P_1\op P_2]
\,,
\end{equation}
leading to the following theorem.
\begin{ttheorem}\label{tmjkdg}
{\bf (Right Gyration Reduction Properties).}
\begin{equation} \label{dmgdc1sc}
\rgyr[P_1,P_2] = \rgyr[\lgyr[P_1,P_2]P_1,P_1\op P_2]
\end{equation}
and
\begin{equation} \label{dmgdc1sd}
\rgyr[P_1,P_2] = \rgyr[P_2\op P_1,\lgyr[P_2,P_1]P_2]
\end{equation}
for all $P_1,P_2\in\Rnm$.
\end{ttheorem}
\begin{proof}
The bi-gyration identity \eqref{dmgdc1sc} is identical with
\eqref{edok01}.
The bi-gyration identity \eqref{dmgdc1sd} is obtained from \eqref{dmgdc1sc}
by applying the bi-gyration inversion law \eqref{madays} followed by
renaming $(P_1,P_2)$ as $(P_2,P_1)$.
\end{proof}

The bi-gyration identity \eqref{dmgdc1sd} involves both left and right
gyrations. We manipulate it into an identity that involves only
right gyrations in
the following chain of equations, which are numbered
for subsequent explanation.
\begin{equation} \label{akhak01}
\begin{split}
\rgyr[P_1,P_2] ~
&\overbrace{=\!\!=\!\!=}^{(1)} \hspace{0.2cm}
\rgyr[P_2\op P_1,\lgyr[P_2,P_1]P_2]
\\
&\overbrace{=\!\!=\!\!=}^{(2)} \hspace{0.2cm}
\rgyr[(P_2\op P_1)\rgyr[P_1,P_2],\lgyr[P_2,P_1]P_2\rgyr[P_1,P2]]
\\
&\hspace{-1.8cm} \overbrace{=\!\!=\!\!=}^{(3)} \hspace{0.2cm}
\rgyr[\lgyr[P_1,P_2](P_2\op P_1)\rgyr[P_1,P_2],
\lgyr[P_1,P_2] \lgyr[P_2,P_1]P_2\rgyr[P_1,P2]]
\\
&\overbrace{=\!\!=\!\!=}^{(4)} \hspace{0.2cm}
\rgyr[P_1\op P_2,P_2\rgyr[P_1,P_2]]
\end{split}
\end{equation}
Derivation of the numbered equalities in \eqref{akhak01} follows:
\begin{enumerate}
\item \label{mylbl01}
This equation is the bi-gyration identity \ref{dmgdc1sd}.
\item \label{mylbl02}
Follows from Result \eqref{urfkb} of Corollary \ref{vfsv}, p.~\pageref{vfsv},
noting that, by Item \ref{mylbl01}, the right gyrations
$\rgyr[P_2\op P_1,\lgyr[P_2,P_1]P_2]$ and $\rgyr[P_1,P_2]$ commute since
they are equal.
\item \label{mylbl03}
Follows from \eqref{bilvb}, p.~\pageref{bilvb}.
\item \label{mylbl04}
Follows from Item \eqref{mylbl03} by applying both the
bi-gyration inversion law \eqref{madays} and the
bi-gyrocommutative law \eqref{anim109s}, p.~\pageref{anim109s}.
\end{enumerate}

Formalizing the results in \eqref{akhak01z} and \eqref{akhak01} we
obtain the following theorem.
\begin{ttheorem}\label{tmjkgf}
{\bf (Bi-gyration Reduction Properties).}
\begin{equation} \label{y1sc1}
\lgyr[P_1,P_2] =
\lgyr[\lgyr[P_1,P_2]P_1,P_1\op P_2]
\end{equation}
and
\begin{equation} \label{y1sc2}
\rgyr[P_1,P_2] =
\rgyr[P_1\op P_2,P_2\rgyr[P_1,P_2]]
\end{equation}
for all $P_1,P_2\in\Rnm$.
\end{ttheorem}

\subsection{Bi-gyration Reduction Properties III}
\label{seccbe03}

As in Subsect.~\ref{seccbe02},
let $P_1,P_2\in\Rnm$, and let $J(P_1,P_2)$ be the
bi-boost symmetric product
\begin{equation} \label{imad02t}
J(P_1,P_2) =
\begin{pmatrix} P_1 \\[4pt] I_n \\[4pt] I_m \end{pmatrix}
\begin{pmatrix} P_2 \\[4pt] I_n \\[4pt] I_m \end{pmatrix}
\begin{pmatrix} P_1 \\[4pt] I_n \\[4pt] I_m \end{pmatrix}
\,,
\end{equation}
which is symmetric with respect to the central bi-boost factor
$(P_2,I_n,I_m)^t$.
Then
\begin{equation} \label{imad03t}
\begin{split}
J(P_1,P_2) &=
\begin{pmatrix} P_1 \\[4pt] I_n \\[4pt] I_m \end{pmatrix}
\begin{pmatrix} P_2\op P_1 \\[4pt] \lgyr[P_2,P_1]
               \\[4pt] \rgyr[P_2,P_1] \end{pmatrix}
\\[6pt] &=
 \begin{pmatrix} P_1 \rgyr[P_2,P_1] \op (P_2\op P_1) \\[4pt]
 \lgyr[P_1\rgyr[P_2,P_1],P_2\op P_1]\lgyr[P_2,P_1] \\[4pt]
 \rgyr[P_2,P_1]\rgyr[P_1\rgyr[P_2,P_1],P_2\op P_1] \end{pmatrix}
=: \begin{pmatrix} P_3 \\[4pt] O_n \\[4pt] O_m \end{pmatrix}
\,.
\end{split}
\end{equation}

By means of \eqref{anim78}, p.~\pageref{anim78}, it is clear that
\begin{equation} \label{imad04t}
J(P_1,P_2)^{-1} = J(-P_1,-P_2)
\,.
\end{equation}
Hence, by the
gyroautomorphic inverse property \eqref{anim94p1}, p.~\pageref{anim94p1},
and by the bi-gyration even property, \eqref{anim94p2}, p.~\pageref{anim94p2},
it is clear from \eqref{imad03t} that
\begin{equation} \label{imad05t}
J(P_1,P_2)^{-1} = J(-P_1,-P_2) =
\begin{pmatrix} -P_3 \\[4pt] O_n \\[4pt] O_m \end{pmatrix}
\,.
\end{equation}

But, it follows from the inverse Lorentz transformation
\eqref{anim83}, p.~\pageref{anim83}, that
\begin{equation} \label{imad06t}
J(P_1,P_2)^{-1} =
\begin{pmatrix} -O_n^{-1} P_3 O_m^{-1} \\[4pt] O_n^{-1} \\[4pt] O_m^{-1}
\end{pmatrix}
\,.
\end{equation}

Comparing the right sides of \eqref{imad06t} and \eqref{imad05t},
we find that $O_m=I_m$ and $O_n=I_n$.
Hence, the bi-boost product $J(P_1,P_2)$ is, again, a bi-boost,
so that by \eqref{imad03t},
\begin{equation} \label{imad07t}
J(P_1,P_2) =
\begin{pmatrix} P_1\rgyr[P_2,P_1]\op (P_2\op P_1) \\[4pt]
 I_n \\[4pt] I_m \end{pmatrix}
\,.
\end{equation}

Following \eqref{imad07t} and \eqref{imad03t} we have the
bi-gyration identities
\begin{equation} \label{imad08t}
\begin{split}
\lgyr[P_1\rgyr[P_2,P_1],P_2\op P_1]\lgyr[P_2,P_1] &= I_n
\\[4pt]
\rgyr[P_2,P_1]\rgyr[P_1\rgyr[P_2,P_1],P_2\op P_1] &= I_m
\,,
\end{split}
\end{equation}
implying
\begin{equation} \label{imad0rt}
\begin{split}
\lgyr[P_1,P_2] &= \lgyr[P_1\rgyr[P_2,P_1],P_2\op P_1] 
\\[4pt]
\rgyr[P_1,P_2] &= \rgyr[P_1\rgyr[P_2,P_1],P_2\op P_1]
\,,
\end{split}
\end{equation}
for all $P_1,P_2\in\Rnm$.

The first entries of \eqref{imad03} and \eqref{imad03t} imply
the interesting identity
\begin{equation} \label{kvrn}
(P_1\op P_2)\op\lgyr[P_1,P_2] P_1
=
P_1\rgyr[P_2,P_1] \op (P_2\op P_1)
\,.
\end{equation}

\subsection{Bi-gyration Reduction Properties IV}
\label{seccbe04}

Let $(P_1,I_n,I_m)^t$ and $(P_2,I_n,I_m)^t$ be two given bi-boosts
in the pseudo-Euclidean space $\Rmcn$, and let the bi-boost
$(X,O_n,O_m)^t$ be given by the equation
\begin{equation} \label{otma01}
\begin{pmatrix} X \\[4pt] O_n \\[4pt] O_m \end{pmatrix}
=
\begin{pmatrix} \om P_1 \\[4pt] I_n \\[4pt] I_m \end{pmatrix}^{-1}
\begin{pmatrix} P_2 \\[4pt] I_n \\[4pt] I_m \end{pmatrix}
\,.
\end{equation}

Then the following two consequences of \eqref{otma01} are equivalent,
\begin{equation} \label{otma02}
\begin{pmatrix} X \\[4pt] O_n \\[4pt] O_m \end{pmatrix}
=
\begin{pmatrix} P_1 \\[4pt] I_n \\[4pt] I_m \end{pmatrix}
\begin{pmatrix} P_2 \\[4pt] I_n \\[4pt] I_m \end{pmatrix}
=
\begin{pmatrix} P_1 \op P_2 \\[4pt] \lgyr[P_1,P_2] \\[4pt] \rgyr[P_1,P_2]
\end{pmatrix}
\end{equation}
and
\begin{equation} \label{otma03}
\begin{pmatrix} P_2 \\[4pt] I_n \\[4pt] I_m \end{pmatrix}
=
\begin{pmatrix} \om P_1 \\[4pt] I_n \\[4pt] I_m \end{pmatrix}
\begin{pmatrix} X \\[4pt] O_n \\[4pt] O_m \end{pmatrix}
=
\begin{pmatrix} \om P_1O_m\op X \\[4pt] \lgyr[\om P_1O_m,X]O_n
\\[4pt] O_m \rgyr[\om P_1O_m,X] \end{pmatrix}
\,.
\end{equation}

The matrix equation \eqref{otma03} in $\Rmcn$ implies
\begin{equation} \label{otma04}
\begin{split}
O_n &= \lgyr[X,\om P_1 O_m]
\\[4pt]
O_m &= \rgyr[X,\om P_1 O_m]
\,,
\end{split}
\end{equation}
so that, by the first entry of the matrix equation \eqref{otma02},
\begin{equation} \label{otma05}
\begin{split}
O_n &= \lgyr[P_1\op P_2,\om P_1 O_m]
\\[4pt]
O_m &= \rgyr[P_1\op P_2,\om P_1 O_m]
\,.
\end{split}
\end{equation}

Inserting $O_n$ and $O_m$ from the second and the third entries of
the matrix equation \eqref{otma02} into \eqref{otma05}, we obtain the
reduction properties
\begin{equation} \label{otma06}
\begin{split}
\lgyr[P_1,P_2] &= \lgyr[P_1\op P_2,\om P_1\rgyr[P_1,P_2]]
\\[4pt]
\rgyr[P_1,P_2] &= \rgyr[P_1\op P_2,\om P_1\rgyr[P_1,P_2]]
\,,
\end{split}
\end{equation}
thus recovering \eqref{akor04s}.

As a first example, the first reduction property in \eqref{otma06}
gives rise to the reduction property
\begin{equation} \label{regfvs}
\lgyr[P_1,P_2] = \lgyr[(P_1\op P_2)\rgyr[P_2,P_1],\om P_1]
\end{equation}
in the following chain of equations, which are numbered
for subsequent explanation.
\begin{equation} \label{asjer01s}
\begin{split}
\lgyr[P_1,P_2]
&\overbrace{=\!\!=\!\!=}^{(1)} \hspace{0.2cm}
\lgyr[P_1\op P_2,\om P_1\rgyr[P_1,P_2]]
\\
&\overbrace{=\!\!=\!\!=}^{(2)} \hspace{0.2cm}
\lgyr[(P_1\op P_2)\rgyr[P_2,P_1],\om P_1\rgyr[P_1,P_2]\rgyr[P_2,P_1]]
\\
&\overbrace{=\!\!=\!\!=}^{(3)} \hspace{0.2cm}
\lgyr[(P_1\op P_2)\rgyr[P_2,P_1],\om P_1]
\,.
\end{split}
\end{equation}
Derivation of the numbered equalities in \eqref{asjer01s} follows:
\begin{enumerate}
\item \label{gmask01}
This is the first identity in \eqref{otma06}.
\item \label{gmask02}
Item \eqref{gmask02} is derived from Item \eqref{gmask01} by applying
Identity \eqref{bilva} of Theorem \ref{spcowd}, p.~\pageref{spcowd},
with $O_m=\rgyr[P_2,P_1]$.
\item \label{gmask03}
Item \ref{gmask03} follows immediately from Item \ref{gmask02}
by the bi-gyration inversion law \eqref{anim107}, p.~\pageref{anim107}.
\end{enumerate}

As a second example, the second reduction property in \eqref{otma06}
gives rise to the reduction property
\begin{equation} \label{regfv}
\rgyr[P_1,P_2] = \rgyr[(P_1\op P_2)\rgyr[P_2,P_1],\om P_1]
\end{equation}
in the following chain of equations, which are numbered
for subsequent explanation.
\begin{equation} \label{asjer01}
\begin{split}
\rgyr[P_1,P_2]
&\overbrace{=\!\!=\!\!=}^{(1)} \hspace{0.2cm}
\rgyr[P_1\op P_2,\om P_1\rgyr[P_1,P_2]]
\\
&\overbrace{=\!\!=\!\!=}^{(2)} \hspace{0.2cm}
\rgyr[(P_1\op P_2)\rgyr[P_2,P_1],\om P_1\rgyr[P_1,P_2]\rgyr[P_2,P_1]]
\\
&\overbrace{=\!\!=\!\!=}^{(3)} \hspace{0.2cm}
\rgyr[(P_1\op P_2)\rgyr[P_2,P_1],\om P_1]
\,.
\end{split}
\end{equation}
Derivation of the numbered equalities in \eqref{asjer01} follows:
\begin{enumerate}
\item \label{fmask01}
This is the second identity in \eqref{otma06}.
\item \label{fmask02}
Being the inverse of $\rgyr[P_1,P_2]$, the right gyrations
$\rgyr[P_2,P_1]$ and $\rgyr[P_1,P_2]$ commute.
Hence, by Item \ref{fmask01}, the right gyrations
$\rgyr[P_2,P_1]$ and
$\rgyr[P_1\op P_2,\om P_1\rgyr[P_1,P_2]]$ commute.
The latter commutativity, in turn, implies Item \ref{fmask02}
by Corollary \ref{vfsv}, p.~\pageref{vfsv},
with $O_m=\rgyr[P_2,P_1]$.
\item \label{fmask03}
Item \ref{fmask03} follows immediately from Item \ref{fmask02}
by the bi-gyration inversion law \eqref{anim107}, p.~\pageref{anim107}.
\end{enumerate}

Formalizing the results in \eqref{regfvs} and \eqref{regfv}
we obtain the following interesting theorem.
\begin{ttheorem}\label{fkthmd} 
For all $P_1,P_2\in\Rnm$,
\begin{equation} \label{fkthmg}
\begin{split}
\lgyr[P_1,P_2] &= \lgyr[P_1\opp P_2,\omp P_1]
\\
\rgyr[P_1,P_2] &= \rgyr[P_1\opp P_2,\omp P_1]
\,,
\end{split}
\end{equation}
where $\opp$ is a binary operation in $\Rnm$ given by
\begin{equation} \label{fkthmg1}
P_1\opp P_2 = (P_1\op P_2)\rgyr[P_2,P_1]
\,.
\end{equation}
\end{ttheorem}

It follows from \eqref{fkthmg1} that
\begin{equation} \label{madkit00}
\omp P=\om P=-P
\,.
\end{equation}
for all $P\in\Rnm$.

\section{Bi-gyrogroups} \label{sedint}

Theorem \ref{fkthmd} indicates that it will
prove useful to replace the binary operation $\op$ in $\Rnm$
by the {\it bi-gyrogroup operation} $\opp$ in $\Rnm$ in Def.~\ref{hdken}
below.

\begin{ddefinition}\label{hdken}
{\bf (Bi-gyrogroup Operation, Bi-gyrogroups).}
Let $(\Rnm,\op)$ be a bi-gyrogroupoid (Def.~\ref{dfkhvb}, p.~\pageref{dfkhvb}).
The bi-gyrogroup binary operation $\opp$ in $\Rnm$ is given by
\begin{equation} \label{hdken01}
P_1\opp P_2 = (P_1\op P_2)\rgyr[P_2,P_1]
\end{equation}
for all $P_1,P_2\in\Rnm$.
The resulting groupoid $(\Rnm,\opp)$ is called a bi-gyrogroup.
\end{ddefinition}

Following \eqref{hdken01} we have, by right gyration inversion,
\eqref{madaysb}, p.~\pageref{madaysb},
\begin{equation} \label{hdken01s}
P_1\op P_2 = (P_1\opp P_2)\rgyr[P_1,P_2]
\end{equation}
for all $P_1,P_2\in\Rnm$.

We will find in the sequel that the bi-gyrogroups $(\Rnm,\opp)$,
rather than the bi-gyrogroupoids $(\Rnm,\op)$, form the desired
elegant algebraic structure that the parametrization of the
Lorentz group $SO(m,n)$ encodes. The point is that we must
study bi-gyrogroupoids in order to pave the way to the
study of bi-gyrogroups.

The bi-gyrogroup operation $\opp$ is determined in \eqref{hdken01}
in terms of the bi-gyrogroupoid operation $\op$ and a right gyration.
It can be determined equivalently by $\op$ and a left gyration as well.
Indeed, it follows from \eqref{hdken01} and the
bi-gyrocommutative law \eqref{anim109s}, p.~\pageref{anim109s},
in $(\Rnm,\op)$ that
\begin{equation} \label{hdken03t}
P_1\opp P_2 = \lgyr[P_1,P_2](P_2\op P_1)
\end{equation}
and hence
\begin{equation} \label{hdken03u}
P_1\op P_2 = \lgyr[P_1,P_2](P_2\opp P_1)
\end{equation}
for all $P_1,P_2\in\Rnm$.

Following Def.~\ref{hdken} of the bi-gyrogroup binary operation $\opp$
in $\Rnm$, it proves useful to express the bi-gyrations of $\Rnm$
in terms of $\opp$ rather than $\op$, in the following theorem.

\begin{ttheorem}\label{smldt}
{\bf (Bi-gyrogroup Bi-gyrations).}
The left and right bi-gyration in the
bi-gyrogroup $(\Rnm,\opp)$ are given by the equations
\begin{equation} \label{anim94u}
\begin{split}
\lgyr[P_1,P_2] &=
\sqrt{I_n+(P_1\opp P_2)(P_1\opp P_2)^t}^{~-1}
\left\{ P_1P_2^t + \sqrt{I_n+P_1P_1^t}\sqrt{I_n+P_2P_2^t} \right\}
\\[6pt]
\rgyr[P_1,P_2] &=
\left\{ P_1^tP_2 + \sqrt{I_m+P_1^tP_1}\sqrt{I_m+P_2^tP_2} \right\}
\sqrt{I_m+(P_2\opp P_1)^t(P_2\opp P_1)}^{~-1}
\end{split}
\end{equation}
for all $P_1,P_2\in\Rnm$.
\end{ttheorem}
\begin{proof}
Noting that $\rgyrab\in SO(m)$, the first equation in
\eqref{anim94u} follows from \eqref{hdken01s}
and the second equation in \eqref{anim94s}, p.~\pageref{anim94s}.
Similarly, noting that $\lgyrab\in SO(n)$, the second equation in 
\eqref{anim94u} follows from \eqref{hdken03u}
and the third equation in \eqref{anim94s}.
\end{proof}

Note that the first equation in \eqref{anim94u} and the second
equation in \eqref{anim94s}, p.~\pageref{anim94s}, are identically the
same equations with a single exception: the binary operation $\op$
in \eqref{anim94s} is replaced by the binary operation $\opp$
in \eqref{anim94u}.
Note also that the order of gyrosummation in the second equation
in \eqref{anim94u} is $P_2\opp P_1$ rather than $P_1\opp P_2$.

Clearly, the identity element of the groupoid $(\Rnm,\opp)$ is $0_{n,m}$,
and the inverse $\omp P$ of $P\in(\Rnm,\opp)$ is
$\omp P=\om P=-P$, as stated in \eqref{madkit00},
noting that $\rgyr[\om P,P]=I_m$ is trivial according to
Corollary \ref{vftr}, p.~\pageref{vftr}.

Following a study of bi-gyrogroups in the sequel we will present an
axiomatic approach to bi-gyrogroups, which forms a natural extension
of the axiomatic approach to groups and to gyrogroups.

\begin{ttheorem}\label{dokre}
{\bf (Bi-gyrogroup Left and Right Automorphisms).}
\begin{equation} \label{adin105s}
\begin{split}
O_n(P_1 \opp P_2) &= O_nP_1 \opp O_nP_2
\\
(P_1 \opp P_2)O_m &= P_1O_m \opp P_2O_m
\\
O_n(P_1 \opp P_2)O_m &= O_nP_1O_m \opp O_nP_2O_m
\end{split}
\end{equation}
for all $P_1,P_2\in\Rnm$, $O_n\in SO(n)$ and $O_m\in SO(m)$.
\end{ttheorem}
\begin{proof}
The first identity in \eqref{adin105s} is proved
in the following chain of equations, which are numbered
for subsequent explanation.
\begin{equation} \label{anthv1}
\begin{split}
O_n(P_1\opp P_2)
&\overbrace{=\!\!=\!\!=}^{(1)} \hspace{0.2cm}
O_n(P_1\op P_2)\rgyr[P_2,P_1]
\\
&\overbrace{=\!\!=\!\!=}^{(2)} \hspace{0.2cm}
(O_nP_1\op O_nP_2)\rgyr[P_2,P_1]
\\
&\overbrace{=\!\!=\!\!=}^{(3)} \hspace{0.2cm}
(O_nP_1\op O_nP_2)\rgyr[O_nP_2,O_nP_1]
\\
&\overbrace{=\!\!=\!\!=}^{(4)} \hspace{0.2cm}
O_nP_1\opp O_nP_2
\,.
\end{split}
\end{equation}
Derivation of the numbered equalities in \eqref{anthv1} follows:
\begin{enumerate}
\item \label{nkvdf01}
Follows from Def.~\ref{hdken}.
\item \label{nkvdf02}
Follows from the first identity in \eqref{adin105}, p.~\pageref{adin105}.
\item \label{nkvdf03}
Follows from \eqref{bilvb}, p.~\pageref{bilvb}.
\item \label{nkvdf04}
Follows from Def.~\ref{hdken}.
\end{enumerate}

The second identity in \eqref{adin105s} is proved
in the following chain of equations, which are numbered
for subsequent explanation.
\begin{equation} \label{anthv2}
\begin{split}
(P_1\opp P_2)O_m
&\overbrace{=\!\!=\!\!=}^{(1)} \hspace{0.2cm}
(P_1\op P_2)\rgyr[P_2,P_1]O_m
\\
&\overbrace{=\!\!=\!\!=}^{(2)} \hspace{0.2cm}
(P_1\op P_2)O_m\rgyr[P_2O_m,P_1O_m]
\\
&\overbrace{=\!\!=\!\!=}^{(3)} \hspace{0.2cm}
(P_1O_m\op P_2O_m)\rgyr[P_2O_m,P_1O_m]
\\
&\overbrace{=\!\!=\!\!=}^{(4)} \hspace{0.2cm}
P_1O_m\opp P_2O_m
\,.
\end{split}
\end{equation}
Derivation of the numbered equalities in \eqref{anthv2} follows:
\begin{enumerate}
\item \label{nkvdg01}
Follows from Def.~\ref{hdken}.
\item \label{nkvdg02}
Follows from in \eqref{spc02}, p.~\pageref{spc02}.
\item \label{nkvdg03}
Follows from the second identity in \eqref{adin105}, p.~\pageref{adin105}.
\item \label{nkvdg04}
Follows from Def.~\ref{hdken}.
\end{enumerate}

Finally, the third identity in \eqref{adin105s} follows immediately
from the first two identities in \eqref{adin105s}.
\end{proof}

The maps
$O_n: P\mapsto O_nP$,
$O_m: P\mapsto PO_m$, and
$(O_n,O_m): P\mapsto O_nPO_m$ of $\Rnm$ onto itself
are bijective. Hence, by Theorem \ref{dokre},
\begin{enumerate}
\item
the map $O_n: P\mapsto O_nP$ is a left automorphism of the
bi-gyrogroup $(\Rnm,\opp)$;
\item
the map $O_m: P\mapsto PO_m$ is a right automorphism of the
bi-gyrogroup $(\Rnm,\opp)$; and
\item
the map $(O_n,O_m): P\mapsto O_nPO_m$ is a bi-automorphism of the
bi-gyrogroup $(\Rnm,\opp)$ (A bi-automorphism being an automorphism
consisting of a left and a right automorphism).
\end{enumerate}

\begin{ttheorem}\label{maskita}
{\bf (Left Cancellation law in $(\Rnm,\opp)$).}
The bi-gyrogroup $(\Rnm,\opp)$ possesses the left cancellation law
\begin{equation} \label{markit01}
\omp P_1\opp(P_1\opp P_2) = P_2
\,,
\end{equation}
for all $P_1,P_2\in\Rnm$.
\end{ttheorem}
\begin{proof}
The proof is provided by the
following chain of equations, which are numbered
for subsequent explanation.
\begin{equation} \label{madkit02}
\begin{split}
\omp P_1\opp(P_1\opp P_2)
&\overbrace{=\!\!=\!\!=}^{(1)} \hspace{0.2cm}
\om P_1\opp(P_1\op P_2)\rgyr[P_2,P_1]
\\
&\overbrace{=\!\!=\!\!=}^{(2)} \hspace{0.2cm}
(\om P_1\op(P_1\op P_2)\rgyr[P_2,P_1])\rgyr[(P_1\op P_2)\rgyr[P_2,P_1],\om P_1]
\\
&\overbrace{=\!\!=\!\!=}^{(3)} \hspace{0.2cm}
(\om P_1\op(P_1\op P_2) \rgyr[P_2,P_1])\rgyr[P_1,P_2]
\\
&\overbrace{=\!\!=\!\!=}^{(4)} \hspace{0.2cm}
\om P_1\rgyr[P_1,P_2] \op (P_1\op P_2)
\\
&\overbrace{=\!\!=\!\!=}^{(5)} \hspace{0.2cm}
P_2
\,.
\end{split}
\end{equation}
Derivation of the numbered equalities in \eqref{madkit02} follows:
\begin{enumerate}
\item \label{dnhst01}
Follows from \eqref{madkit00} and from
Def.~\ref{hdken} of $\opp$ applied to $P_1\opp P_2$.
\item \label{dnhst02}
Follows from Def.~\ref{hdken} of $\opp$.
\item \label{dnhst03}
Follows from \eqref{regfv}, p.~\pageref{regfv}.
\item \label{dnhst04}
Follows from the second identity in \eqref{adin105} of
Theorem \ref{dokrd}, p.~\pageref{dokrd}, applied with $O_m=\rgyr[P_2,P_1]$,
and from the bi-gyration inversion law \eqref{anim107}, p.~\pageref{anim107}.
\item \label{dnhst05}
Follows from the left cancellation law \eqref{rtmdc1b}, p.~\pageref{rtmdc1b},
in $(\Rnm,\op)$.
\end{enumerate}
\end{proof} 

\begin{llemma}\label{mdf1}
Let $O_n\in SO(n)$ and $O_m\in SO(m)$, $n,m\in\Nb$. Then,
\begin{equation} \label{flgd1}
O_nPO_m = P
\end{equation}
for all $P\in\Rnm$ if and only if $O_n=I_n$ and $O_m=I_m$.
\end{llemma}
\begin{proof}
If $O_n=I_n$ and $O_m=I_m$, then obviously \eqref{flgd1}
is true for all $P\in\Rnm$.

Conversely, assuming $O_nPO_m = P$, or equivalently,
\begin{equation} \label{flgd01}
O_n^tP = PO_m
\,,
\end{equation}
$O_n\in SO(n)$, $O_m\in SO(m)$,
for all $P\in\Rnm$, we will show that $O_n=I_n$ and $O_m=I_m$.

Let
\begin{equation} \label{flgd02}
O_n^t =
\begin{pmatrix} a_{11} & ~~\ldots ~~ & a_{1n} \\[4pt]
                \vdots &     ~       &        \\[4pt]
                a_{n1} & ~~\ldots ~~ & a_{nn} \end{pmatrix}
\in SO(n)
\end{equation}
and
\begin{equation} \label{flgd03}
O_m =
\begin{pmatrix} b_{11} & ~~\ldots ~~ & b_{1m} \\[4pt]
                \vdots &     ~       &        \\[4pt]
                b_{m1} & ~~\ldots ~~ & b_{mm} \end{pmatrix}
\in SO(m)
\,.
\end{equation}

Furthermore, let $P_{ij}\in\Rnm$ be the matrix
\begin{equation} \label{flgd04}
P_{ij} =
\begin{pmatrix} 0 & \ldots & 0 & \dots & 0 \\[4pt]
                \vdots &     ~       &        \\[4pt]
                0 & \ldots & 1 & \dots & 0 \\[4pt]
                \vdots &     ~       &        \\[4pt]
                0 & \ldots & 0 & \dots & 0 \end{pmatrix}
\in\Rnm
\end{equation}
with one at the $ij$-entry and zeros elsewhere,
$i=1,\ldots,n$, $j=1,\ldots,m$.

Then the matrix product $O_n^tP_{ij}$,
\begin{equation} \label{flgd05}
O_n^tP_{ij} =
\begin{pmatrix} 0 & \ldots & a_{1i} & \ldots & 0 \\[4pt]
                \vdots &     ~       &        \\[4pt]
                0 & \ldots & a_{ii} & \ldots & 0 \\[4pt]
                \vdots &     ~       &        \\[4pt]
                0 & \ldots & a_{ni} & \ldots & 0 \end{pmatrix}
\in\Rnm
\,,
\end{equation}
is a matrix with $j$-th column $(a_{1i},\ldots,a_{ii},\ldots,a_{ni})^t$
and zeros elsewhere. Shown explicitly in \eqref{flgd05} are the first
column, the $j$-th column and the $m$-th column of the matrix $O_n^tP_{ij}$,
along with its first row, $i$-th row and $n$-th row.

Similarly, the matrix product $P_{ij}O_m$,
\begin{equation} \label{flgd06}
P_{ij}O_m =
\begin{pmatrix} 0 & \ldots &    0   & \ldots & 0 \\[4pt]
                \vdots &     ~       &        \\[4pt]
            b_{1j}& \ldots & b_{jj} & \ldots & b_{jm} \\[4pt]
                \vdots &     ~       &        \\[4pt]
                0 & \ldots &    0   & \ldots & 0 \end{pmatrix}
\in\Rnm
\,,
\end{equation}
is a matrix with $i$-th row $(b_{1j},\ldots,b_{jj},\ldots,b_{jm})$
and zeros elsewhere. Shown explicitly in \eqref{flgd06} are the first
column, the $j$-th column and the $m$-th column of the matrix $P_{ij}O_m$,
along with its first row, $i$-th row and $n$-th row.

It follows from \eqref{flgd01} that \eqref{flgd05} and \eqref{flgd06}
are equal.
Hence, by comparing entries of the matrices in
\eqref{flgd05}\,--\,\eqref{flgd06} we have
\begin{equation} \label{flgd07}
a_{ii}=b_{jj}
\end{equation}
and
\begin{equation} \label{flgd08}
\begin{split}
a_{ii_1} &= 0
\\[4pt]
b_{jj_1} &= 0
\end{split}
\end{equation}
for all $i,i_1=1,\ldots,n$,
and all $j,j_1=1,\ldots,m$, $i_1\ne i$ and $j_1\ne j$.

By \eqref{flgd07}\,--\,\eqref{flgd08} and
\eqref{flgd02}\,--\,\eqref{flgd03} we have
\begin{equation} \label{flgd09}
\begin{split}
O_n^t &= \lambda I_n
\\[4pt]
O_m   &= \lambda I_m
\,.
\end{split}
\end{equation}
Moreover, $\lambda=1$ since, by assumption,
$O_n\in SO(n)$ and $O_m\in SO(m)$. Hence, $O_n=I_n$ and $O_m=I_m$,
as desired.
\end{proof}

The following Lemma \ref{mdf2} is an immediate consequence of
Lemma \ref{mdf1}.
\begin{llemma}\label{mdf2}
Let $O_{n,k}\in SO(n)$ and $O_{m,k}\in SO(m)$, $n,m\in\Nb$, $k=1,2$. Then,
\begin{equation} \label{flge1}
O_{n,1}PO_{m,1} = O_{n,2}PO_{m,2}
\end{equation}
for all $P\in\Rnm$ if and only if $O_{n,1}=O_{n,2}$ and $O_{m,1}=O_{m,2}$.
\end{llemma}

\section{Bi-gyration Decomposition and Polar Decomposition}
\label{trgd}

In this section we present manipulations that lead to the
bi-gyroassociative and bi-gyrocommutative laws of the binary operation $\opp$
in Theorems \ref{thmasso45} and \ref{timbk} below.

The product of two bi-boosts, $B(P_1)$ and $B(P_2)$, $P_1,P_2\in\Rnm$,
is a Lorentz transformation $\Lambda=B(P_1)B(P_2)\in SO(m,n)$
that need not be a bi-boost. As such, it
possesses the bi-gyration decomposition \eqref{anim66a}, p.~\pageref{anim66a},
as well as the polar decomposition \eqref{anim66c}, p.~\pageref{anim66c},
along with the bi-gyration in \eqref{anim92}, p.~\pageref{anim92}.

The bi-gyration decomposition of the bi-boost product gives rise to
the binary operation $\op$ in $\Rnm$ as follows.
By \eqref{anim91}, p.~\pageref{anim91}, the bi-boost product
$B(P_1)B(P_2)$ possesses the unique bi-gyration decomposition
\eqref{anim93},
\begin{equation} \label{elfd01}
B(P_1)B(P_2) = \rho(\rgyr[P_1,P_2])B(P_{12})\lambda(\lgyr[P_1,P_2])
\end{equation}
where, by Def.~\ref{dfkhvb}, p.~\pageref{dfkhvb},
\begin{equation} \label{elfd01k}
P_{12} =: P_1\op P_2
\,.
\end{equation}

Similarly,
the polar decomposition of the bi-boost product gives rise to
the binary operation $\opp$ in $\Rnm$ as follows.
By \eqref{anim66c}, p.~\pageref{anim66c}, and
\eqref{anim92}, the bi-boost product
$B(P_1)B(P_2)$ possesses the unique polar decomposition
\begin{equation} \label{elfd03}
B(P_1)B(P_2) = B(P_{12}^{\prime\prime})\rho(\rgyr[P_1,P_2])
\lambda(\lgyr[P_1,P_2])
\end{equation}
where, by definition,
\begin{equation} \label{elfd03k}
P_{12}^{\prime\prime} =: P_1\op^{\prime\prime} P_2
\,.
\end{equation}

In order to see the relationship between the binary operations
$\op$ and $\opp$ in $\Rnm$ we employ the second identity
in \eqref{anim80}, p.~\pageref{anim80},
with $O_m=\rgyr[P_1,P_2]$ to manipulate the polar decomposition
\eqref{elfd03} into the equivalent bi-gyration decomposition,
\begin{equation} \label{elfd04}
\begin{split}
B(P_1)B(P_2) &=
B(P_{12}^{\prime\prime})\rho(\rgyr[P_1,P_2])\lambda(\lgyr[P_1,P_2])
\\ &=
\rho(\rgyr[P_1,P_2]) B(P_{12}^{\prime\prime}\rgyr[P_1,P_2])
\lambda(\lgyr[P_1,P_2])
\,.
\end{split}
\end{equation}

Comparing \eqref{elfd04} with \eqref{elfd01}, noting that the
bi-gyration decomposition is unique, we find that
$P_{12}^{\prime\prime}\rgyr[P_1,P_2] = P_{12}$,
or equivalently, by means of \eqref{elfd01k} and \eqref{elfd03k},
\begin{equation} \label{elfd05}
\begin{split}
P_1\op^{\prime\prime} P_2 &= (P_1\op P_2)\rgyr[P_2,P1]
\\[4pt]
P_1\op P_2 &= (P_1\op^{\prime\prime} P_2)\rgyr[P_1,P2]
\end{split}
\end{equation}
in agreement with the definition of $\opp$ in Def.~\ref{hdken}.
Hence,
\begin{equation} \label{eqf250}
\op^{\prime\prime} = \opp
\,.
\end{equation}
It follows from \eqref{eqf250} that the bi-gyrogroup operation
$\opp=\op^{\prime\prime}$
in Def.~\ref{hdken} stems from the polar decomposition \eqref{elfd03},
just as the bi-gyrogroupoid operation $\op$ stems from the
bi-gyration decomposition \eqref{elfd01}.

It is convenient here to temporarily use the short notation
\begin{equation} \label{dunu00}
\begin{split}
\lab &:= \lgyr[P_1,P_2]
\\[4pt]
\rab &:= \rgyr[P_1,P_2]
\end{split}
\end{equation}
in intermediate results, turning back to the full notation in
final results, noting that $L_{P_1,P_2}^{-1}=L_{P_2,P_1}$
and $R_{P_1,P_2}^{-1}=R_{P_2,P_1}$.

Identities \eqref{elfd03} and \eqref{elfd05} imply
\begin{equation} \label{dunu01}
\rho(\rab)\lambda(\lab) = B(-(P_1\op P_2)\rba) B(P_1)B(P_2)
\,.
\end{equation}

Identities \eqref{elfd01} and \eqref{elfd05} imply,
by right gyration inversion,
the following chain of equations, which are numbered
for subsequent explanation.
\begin{equation} \label{dunu02}
\begin{split}
B(P_1\op P_2)\lambda(\lab)
&\overbrace{=\!\!=\!\!=}^{(1)} \hspace{0.2cm}
\rho(\rba) B(P_1)B(P_2)
\\
&\overbrace{=\!\!=\!\!=}^{(2)} \hspace{0.2cm}
B(P_1\rab)\rho(\rba) B(P_2)
\\
&\overbrace{=\!\!=\!\!=}^{(3)} \hspace{0.2cm}
B(P_1\rab) B(P_2\rab) \rho(\rba)
\,.
\end{split}
\end{equation}
Derivation of the numbered equalities in \eqref{dunu02} follows:
\begin{enumerate}
\item \label{vdy01}
This identity is obtained from \eqref{elfd01} and \eqref{elfd01k} by
using the right gyration inversion law in \eqref{anim107} according
to which $\rho(\rgyr[P_1,P_2])^{-1}=\rho(\rba)$.
\item \label{vdy02}
Follows from Item \eqref{vdy01} by an application to $B(P_1)$ of the
second identity in
\eqref{anim80}, p.~\pageref{anim80}, with $O_m=\rba$, noting the
right gyration inversion law, $R_{P_1,P_2} R_{P_2,P_1} = I_m$.
\item \label{vdy03}
Like Item \eqref{vdy02}, Item \eqref{vdy03}
follows from an application to $B(P_2)$ of the second identity in
\eqref{anim80}, p.~\pageref{anim80}, with $O_m=\rba$, noting the
right gyration inversion law, $R_{P_1,P_2} R_{P_2,P_1} = I_m$.
\end{enumerate}

By means of \eqref{dunu02} and right gyration inversion we have
\begin{equation} \label{dunu03}
B(P_1\op P_2) = B(P_1\rab)B(P_2\rab)\rho(\rba)\lambda(\lba)
\end{equation}
so that, by bi-boost inversion,
\begin{equation} \label{dunu04}
\rho(\rba)\lambda(\lba) = B(\om P_2\rab)B(\om P_1\rab) B(P_1\op P_2)
\,.
\end{equation}

Inverting both sides of \eqref{dunu04} and noting that the
matrices $\lambda(\lab)$ and $\rho(\rab)$ commute, we obtain the
identity
\begin{equation} \label{dunu05}
\rho(\rab)\lambda(\lab) = B(\om (P_1\op P_2)) B(P_1\rab)B(P_2\rab)
\,.
\end{equation}

Comparing \eqref{dunu01} and \eqref{dunu05}, we obtain the identity
\begin{equation} \label{dunu06}
\begin{split}
B(\om(P_1\op P_2)\rba) B(P_1)B(P_2) &=
B(\om(P_1\op P_2)) B(P_1\rab)B(P_2\rab)
\\[4pt] &=
\rho(\rab) \lambda(\lab)
\,,
\end{split}
\end{equation}
which, in full notation, takes the form
\begin{equation} \label{dunu07}
\begin{split}
&B(\om(P_1\op P_2)\rgyr[P_2,P_1]) B(P_1)B(P_2)
\\[4pt] &=
B(\om(P_1\op P_2))B(P_1\rgyr[P_1,P_2]) B(P_2\rgyr[P_1,P_2])
\\[4pt] &=
\rho(\rgyr[P_1,P_2]) \lambda(\lgyr[P_1,P_2])
\,.
\end{split}
\end{equation}

By Def.~\ref{hdken}, the extreme sides of \eqref{dunu07} yield the
identity
\begin{equation} \label{dunu08}
\rho(\rgyr[P_1,P_2]) \lambda(\lgyr[P_1,P_2]) =
B(\om(P_1\opp P_2)) B(P_1)B(P_2)
\,,
\end{equation}
so that for all $P_1,P_2,X\in\Rnm$,
\begin{equation} \label{dunu09}
\rho(\rgyr[P_1,P_2]) \lambda(\lgyr[P_1,P_2])B(X) =
B(\om(P_1\opp P_2)) B(P_1)B(P_2)B(X)
\,.
\end{equation}

Let $J_1$ ($J_2$) denote the left (right) side of \eqref{dunu09}.
Using the column notation in \eqref{anim68}, p.~\pageref{anim68},
we manipulate the left side, $J_1$, of \eqref{dunu09} as follows,
where we apply the Lorentz transformation product law
\eqref{anim98}, p.~\pageref{anim98}, and note Corollary \ref{vftr}
on trivial bi-gyrations.
\begin{equation} \label{dunu10}
\begin{split}
J_1 &= \rho(\rgyr[P_1,P_2]) \lambda(\lgyr[P_1,P_2])B(X)
\\[8pt] &=
\begin{pmatrix} 0_{n,m} \\[4pt] \lgyrab \\[4pt] I_m \end{pmatrix}
\begin{pmatrix} 0_{n,m} \\[4pt] I_n \\[4pt] \rgyrab \end{pmatrix}
\begin{pmatrix} X       \\[4pt] I_n \\[4pt] I_m \end{pmatrix}
=
\begin{pmatrix} 0_{n,m} \\[4pt] \lgyrab \\[4pt] \rgyrab \end{pmatrix}
\begin{pmatrix} X       \\[4pt] I_n \\[4pt] I_m \end{pmatrix}
\\[8pt] &=
\begin{pmatrix} \lgyrab X \\[4pt] \lgyr[0_{n,m},\lgyrab X] \lgyrab
\\[4pt] \rgyrab \rgyr[0_{n,m},\lgyrab X] \end{pmatrix}
=
\begin{pmatrix} \lgyrab X \\[4pt] \lgyrab \\[4pt] \rgyrab \end{pmatrix}
=:
\begin{pmatrix} A_1 \\[4pt] B_1 \\[4pt] C_1 \end{pmatrix}
\,.
\end{split}
\end{equation}

Similarly, applying the Lorentz transformation product law \eqref{anim98}
we manipulate the right side, $J_2$, of \eqref{dunu09} as follows.
\begin{equation} \label{dunu11}
\begin{split}
J_2 &= B(\om(P_1\opp P_2)) B(P_1)B(P_2)B(X)
\\[8pt] &=
\begin{pmatrix} \om(P_1\opp P_2) \\[4pt] I_n \\[4pt] I_m \end{pmatrix}
\begin{pmatrix} P_1 \\[4pt] I_n \\[4pt] I_m \end{pmatrix}
\begin{pmatrix} P_2 \\[4pt] I_n \\[4pt] I_m \end{pmatrix}
\begin{pmatrix} X   \\[4pt] I_n \\[4pt] I_m \end{pmatrix}
=
\begin{pmatrix} \om(P_1\opp P_2) \\[4pt] I_n \\[4pt] I_m \end{pmatrix}
\begin{pmatrix} P_1 \\[4pt] I_n \\[4pt] I_m \end{pmatrix}
\begin{pmatrix} P_2\op X \\[4pt] \lgyr[P_2,X] \\[4pt] \rgyr[P_2,X]
\end{pmatrix}
\\[8pt] &=
\begin{pmatrix} \om(P_1\opp P_2) \\[4pt] I_n \\[4pt] I_m \end{pmatrix}
\begin{pmatrix} P_1\rgyr[P_2,X]\op (P_2\op X) \\[4pt]
\lgyr[P_1\rgyr[P_2,X],P_2\op X] \lgyr[P_2,X] \\[4pt]
\rgyr[P_2,X] \rgyr[P_1\rgyr[P_2,X],P_2\op X] \end{pmatrix}
\,.
\end{split}
\end{equation}

In the following equations \eqref{dunu12} we adjust each entry of the
right column of the extreme right side of \eqref{dunu11} to our needs.

By the second equation in \eqref{adin105}, p.~\pageref{adin105},
with $O_m=\rgyr[P_2,X]$, and the
right gyration inversion law \eqref{madaysb}, and by
\eqref{hdken01}\,--\,\eqref{hdken01s}, we have
\begin{subequations} \label{dunu12}
\begin{equation} \label{dunu12a}
\begin{split}
P_1\rgyr[P_2,X]\op(P_2\op X) &=
\{ P_1\op(P_2\op X)\rgyr[X,P_2]\}\rgyr[P_2,X]
\\[4pt] &=
\{ P_1\op(P_2\opp X)\} \rgyr[P_2,X]
\\[4pt] &=
\{ P_1\opp(P_2\opp X) \} \rgyr[P_1,P_2\opp X] \rgyr[P_2,X]
\,.
\end{split}
\end{equation}

By \eqref{bilva} with $O_m=\rgyr[X,P_2]$, and the
right gyration inversion law \eqref{madaysb}, and by
\eqref{hdken01}, we have
\begin{equation} \label{dunu12b}
\begin{split}
\lgyr[P_1\rgyr[P_2,X],P_2\op X] &=
\lgyr[P_1,(P_2\op X)\rgyr[X,P_2]]
\\[4pt] &=
\lgyr[P_1,P_2\opp X]
\,.
\end{split}
\end{equation}

By \eqref{spc02} with $O_m=\rgyr[P_2,X]$, and the
right gyration inversion law \eqref{madaysb}, and by
\eqref{hdken01}, we have
\begin{equation} \label{dunu12c}
\begin{split}
\rgyr[P_2,X] \rgyr[P_1\rgyr[P_2,X],P_2\op X] &=
\rgyr[P_1,(P_2\op X)\rgyr[X,P_2]] \rgyr[P_2,X]
\\[4pt] &=
\rgyr[P_1,P_2\opp X] \rgyr[P_2,X]
\,.
\end{split}
\end{equation}
\end{subequations}

By means of the equations in \eqref{dunu12}, the extreme right side
of \eqref{dunu11} can be written as
\begin{equation} \label{dunu13}
J_2 =
\begin{pmatrix} \om(P_1\opp P_2) \\[4pt] I_n \\[4pt] I_m \end{pmatrix}
\begin{pmatrix} 
\{ P_1\opp(P_2\opp X) \} \rgyr[P_1,P_2\opp X] \rgyr[P_2,X] \\[4pt]
\lgyr[P_1,P_2\opp X] \lgyr[P_2,X] \\[4pt]
\rgyr[P_1,P_2\opp X] \rgyr[P_2,X] \end{pmatrix}
=: \begin{pmatrix} A_2 \\[4pt] B_2 \\[4pt] C_2 \end{pmatrix}
\,.
\end{equation}

We now face the task of calculating $A_2$, $B_2$ and $C_2$ by means
of the Lorentz product law \eqref{anim98}.
Applying the Lorentz product law to \eqref{dunu13}, we calculate
the second entry, $B_2$, of $J_2$ and simplify it in the
following chain of equations, which are numbered
for subsequent explanation,
and where we use the notation
\begin{equation} \label{fiens}
O_m = \rgyr[P_1,P_2\opp X]\rgyr[P_2,X]
\,.
\end{equation}
\begin{equation} \label{dunu14}
\begin{split}
B_2
&\overbrace{=\!\!=\!\!=}^{(1)} \hspace{0.2cm}
\lgyr[\om(P_1\opp P_2) O_m, \{ P_1\opp(P_2\opp X)\} O_m]
\lgyr[P_1,P_2\opp X] \lgyr[P_2,X]
\\
&\overbrace{=\!\!=\!\!=}^{(2)} \hspace{0.2cm}
\lgyr[\om(P_1\opp P_2) , P_1\opp(P_2\opp X)]
\lgyr[P_1,P_2\opp X] \lgyr[P_2,X]
\,.
\end{split}
\end{equation}
Derivation of the numbered equalities in \eqref{dunu14} follows:
\begin{enumerate}
\item \label{dlth01}
This equation is obtained by calculating the Lorentz transformation
product in \eqref{dunu13} by means of \eqref{anim98}, selecting the
resulting second entry, and using the notation in \eqref{fiens}.
\item \label{dlth02}
Follows from Item \eqref{dlth01} by omitting the matrix $O_m$ from
the two entries of $\lgyr$
according to \eqref{bilva}, p.~\pageref{bilva}.
\end{enumerate}

By \eqref{dunu09}, $J_1=J_2$ and hence, by \eqref{dunu10} and
\eqref{dunu13}, $B_2=B_1$, that is, by \eqref{dunu14} and
\eqref{dunu10},
\begin{equation} \label{dunu15}
\lgyr[\om(P_1\opp P_2) , P_1\opp(P_2\opp X)]
\lgyr[P_1,P_2\opp X] \lgyr[P_2,X]
= \lgyr[P_1,P_2]
\end{equation}
for all $P_1,P_2,X\in\Rnm$.

Similarly, we calculate
the third entry, $C_2$, of $J_2$ and simplify it in the
following chain of equations, which are numbered
for subsequent explanation.
\begin{equation} \label{dunu16}
\begin{split}
C_2
&\overbrace{=\!\!=\!\!=}^{(1)} \hspace{0.2cm}
\rgyr[P_1,P_2\opp X]\rgyr[P_2,X]
\\ & \hspace{0.0cm}
\rgyr[\om(P_1\opp P_2)\rgyr[P_1,P_2\opp X] \rgyr[P_2,X] ,
\{ P_1\opp(P_1\opp X)\} \rgyr[P_1,P_2\opp X] \rgyr[P_2,X]]
\\
&\overbrace{=\!\!=\!\!=}^{(2)} \hspace{0.2cm}
\rgyr[P_1,P_2\opp X]\rgyr[\om(P_1\opp P_2)\rgyr[P_1,P_2\opp X],
\{ P_1\opp (P_2\opp X)\} \rgyr[P_1,P_2\opp X]]
\\ & \hspace{1.0cm}
\rgyr[P_2,X]
\\
&\overbrace{=\!\!=\!\!=}^{(3)} \hspace{0.2cm}
\rgyr[\om(P_1\opp P_2), P_1\opp(P_2\opp X)]
\rgyr[P_1,P_2\opp X]\rgyr[P_2,X]
\,.
\end{split}
\end{equation}
Derivation of the numbered equalities in \eqref{dunu16} follows:
\begin{enumerate}
\item \label{ukdnb01}
This equation is obtained by calculating the Lorentz transformation
product in \eqref{dunu13} by means of \eqref{anim98}, and selecting
the resulting third entry.
\item \label{ukdnb02}
Follows from Item \eqref{ukdnb01} by applying Identity
\eqref{spc02}, p.~\pageref{spc02}, with $O_m=\rgyr[P_2,X]$.
\item \label{ukdnb03}
Follows from Item \eqref{ukdnb02} by applying Identity
\eqref{spc02}, p.~\pageref{spc02}, with $O_m=\rgyr[P_1,P_2 \opp X]$.
\end{enumerate}

By \eqref{dunu09}, $J_1=J_2$ and hence, by \eqref{dunu10} and
\eqref{dunu13}, $C_2=C_1$, that is, by \eqref{dunu16}
and \eqref{dunu10},
\begin{equation} \label{dunu17}
\rgyr[\om(P_1\opp P_2), P_1\opp(P_2\opp X)]
\rgyr[P_1,P_2\opp X]\rgyr[P_2,X]
= \rgyr[P_1,P_2]
\end{equation}
for all $P_1,P_2,X\in\Rnm$.

We are now in a position to calculate
the first entry, $A_2$, of $J_2$ and simplify it in the
following chain of equations, which are numbered
for subsequent explanation.
\begin{equation} \label{dunu18}
\begin{split}
A_2
&\overbrace{=\!\!=\!\!=}^{(1)} \hspace{0.2cm}
\om(P_1\opp P_2) \rgyr[P_1,P_2\opp X] \rgyr[P_2,X]
\\ & \hspace{1.0cm}
\op \{ P_1\opp (P_2\opp X)\} \rgyr[P_1,P_2\opp X] \rgyr[P_2,X]
\\
&\overbrace{=\!\!=\!\!=}^{(2)} \hspace{0.2cm}
\{\om(P_1\opp P_2)\op\{ P_1\opp(P_2\opp X)\}\}
\rgyr[P_1,P_2\opp X] \rgyr[P_2,X]
\\
&\overbrace{=\!\!=\!\!=}^{(3)} \hspace{0.2cm}
\{\om(P_1\opp P_2)\opp\{ P_1\opp(P_2\opp X)\}\}
\rgyr[\om(P_1\opp P_2),P_1\opp (P_2\opp X)]
\\ & \hspace{1.0cm}
\rgyr[P_1,P_2\opp X] \rgyr[P_2,X]
\\
&\overbrace{=\!\!=\!\!=}^{(4)} \hspace{0.2cm}
\{\om(P_1\opp P_2)\opp\{ P_1\opp(P_2\opp X)\}\}
\rgyr[P_1,P_2]
\,.
\end{split}
\end{equation}
Derivation of the numbered equalities in \eqref{dunu18} follows:
\begin{enumerate}
\item \label{tksnr01}
This equation is obtained by calculating the Lorentz transformation
product in \eqref{dunu13} by means of \eqref{anim98}, and selecting
the resulting first entry.
\item \label{tksnr02}
Item \eqref{tksnr02} is obtained by using the second Identity in
\eqref{adin105} with $$O_m=\rgyr[P_1,P_2\opp X] \rgyr[P_2,X].$$
\item \label{tksnr03}
The binary operation $\op$ that appears in Item \eqref{tksnr02}
is expressed here in terms of the binary operation $\opp$
by means of \eqref{hdken01s}.
\item \label{tksnr04}
Item \eqref{tksnr04} follows from Item \eqref{tksnr03} by
Identity \eqref{dunu17}.
\end{enumerate}

By \eqref{dunu09}, $J_1=J_2$ and hence, by \eqref{dunu10} and
\eqref{dunu13}, $A_2=A_1$, that is, by \eqref{dunu18} and
\eqref{dunu10},
\begin{equation} \label{dunu19}
\{\om(P_1\opp P_2)\opp\{ P_1\opp(P_2\opp X)\}\}
\rgyr[P_1,P_2]
= \lgyr[P_1,P_2]X
\,.
\end{equation}
Hence, by right gyration inversion,
\begin{equation} \label{dunu20}
\omp(P_1\opp P_2)\opp\{ P_1\opp(P_2\opp X)\}
= \lgyr[P_1,P_2]X \rgyr[P_2,P_1]
\end{equation}
for all $P_1,P_2,X\in\Rnm$.

Left gyroadding $(P_1\opp P_2)\opp$ to both sides of
\eqref{dunu20} and applying the left cancellation law
\eqref{markit01}, we obtain the {\it left bi-gyroassociative law},
\begin{equation} \label{bkre}
\begin{split}
&(P_1\opp P_2)\opp\lgyr[P_1,P_2] X \rgyr[P_2,P_1]
\\[6pt] &=
(P_1\opp P_2)\opp\{\omp(P_1\opp P_2)\opp\{ P_1\opp(P_2\opp X)\}\}
\\[6pt] &=
P_1\opp(P_2\opp X)
\,.
\end{split}
\end{equation}

\begin{ttheorem}\label{thmasso45} 
{\bf (Bi-gyrogroup Left and Right Bi-gyroassociative Law).}
The binary operation $\opp$ in $\Rnm$ possesses
the left bi-gyroassociative law
\begin{equation} \label{dunu21}
P_1\opp(P_2\opp X) = (P_1\opp P_2) \opp
\lgyr[P_1,P_2] X \rgyr[P_2,P_1]
\end{equation}
and the right bi-gyroassociative law
\begin{equation} \label{dunu22}
(P_1\opp P_2) \opp X = P_1\opp
(P_2\opp\lgyr[P_2,P_1] X \rgyr[P_1,P_2])
\end{equation}
for all $P_1,P_2,X\in\Rnm$.
\end{ttheorem}
\begin{proof}
The left bi-gyroassociative law \eqref{dunu21} is proved in
\eqref{bkre}.

The right bi-gyroassociative law \eqref{dunu22} results from an application
of the left bi-gyroassociative law to the right side of \eqref{dunu22},
by means of bi-gyration inversion,
\begin{equation} \label{bkre1}
\begin{split}
&P_1\opp(P_2\opp\lgyr[P_2,P_1] X \rgyr[P_1,P_2])
\\[8pt] &=
(P_1\opp P_2)\opp\lgyr[P_1,P_2]\lgyr[P_2,P_1] X
\rgyr[P_1,P_2]\rgyr[P_2,P_1]
\\[8pt] &=
(P_1\opp P_2)\opp X
\,.
\end{split}
\end{equation}
\end{proof} 

\section{Bi-gyrocommutative Law}
\label{trge}

The bi-gyrocommutative law in $(\Rnm,\opp)$ is obtained in Sect.~\ref{trgd}
by comparing the bi-gyration decomposition and the polar decomposition
of the bi-boost product $\Lambda=B(P_1)B(P_2)$.
In this section we derive the bi-gyrocommutative law in $(\Rnm,\opp)$
from its counterpart \eqref{anim109s}, p.~\pageref{anim109s}, in
$(\Rnm,\op)$.

\begin{ttheorem}\label{timbk} 
{\bf (Bi-gyrocommutative Law in $(\Rnm,\opp)$).}
The binary operation $\opp$ in $\Rnm$ possesses the
bi-gyrocommutative law
\begin{equation} \label{bagt1}
P_1 \opp P_2 = \lgyr[P_1,P_2] (P_2\opp P_1) \rgyr[P_2,P_1]
\end{equation}
for all $P_1,P_2\in\Rnm$.
\end{ttheorem}
\begin{proof}
By means of \eqref{hdken01s}, p.~\pageref{hdken01s}, and
right gyration inversion \eqref{madaysb}, p.~\pageref{madaysb},
the bi-gyrocommutative law \eqref{anim109s}, p.~\pageref{anim109s},
in $(\Rnm,\op)$ can be expressed in terms of $\opp$ rather than $\op$,
obtaining
\begin{equation} \label{bagt2}
\begin{split}
(P_1\opp P_2)\rgyrab &= \lgyrab (P_2\opp P_1)\rgyrba\rgyrab
\\[4pt]
&= \lgyrab(P_2\opp P_1)
\,.
\end{split}
\end{equation}

Identity \eqref{bagt1} of the Theorem follows immediately
from \eqref{bagt2} by right gyration inversion.
\end{proof} 

\section{Gyrogroup Gyrations}
\label{trgf}

The bi-gyroassociative laws \eqref{dunu21}\,--\,\eqref{dunu22}
and the bi-gyrocommutative law \eqref{bagt1} suggest the following
definition of gyrations in terms of left and right gyrations.
\begin{ddefinition}\label{itsk1}
{\bf (Gyrogroup Gyrations).}
The gyrator $$\gyr:\Rnm\times\Rnm\rightarrow\Aut(\Rnm,\opp)$$
generates automorphisms called gyrations, $\gyrabp\in\Aut(\Rnm,\opp)$,
given by the equation
\begin{equation} \label{kvir01}
\gyrabp X=\lgyrab X\rgyrba
\end{equation}
for all $P_1,P_2,X\in\Rnm$, where left gyrations, $\lgyrab$,
and right gyrations, $\rgyrba$,
are given in \eqref{anim94s}, p.~\pageref{anim94s}.
The gyration $\gyrabp$ is said to be the gyration generated by
$P_1,P_2\in\Rnm$.
Being automorphisms of $(\Rnm,\opp)$, gyrations are also called
gyroautomorphisms.
\end{ddefinition}

Def.~\ref{itsk1} will turn out rewarding, leading to the discovery that
any bi-gyrogroup $(\Rnm,\opp)$ is a gyrocommutative gyrogroup.

\begin{ttheorem}\label{itsk2} 
{\bf (Gyrogroup Gyroassociative and gyrocommutative Laws).}
The binary operation $\opp$ in $\Rnm$ obeys the left and the right
gyroassociative law
\begin{equation} \label{kvir02}
P_1\opp (P_2\opp X) = (P_1\opp P_2) \opp\gyrabp X
\end{equation}
and
\begin{equation} \label{kvir03}
(P_1\opp P_2) \opp X = P_1\opp (P_2\opp\gyrbap X)
\end{equation}
and the gyrocommutative law
\begin{equation} \label{kvir03d5}
P_1\opp P_2 = \gyrabp (P_2\opp P_1)
\,.
\end{equation}
\end{ttheorem}
\begin{proof}
Identities \eqref{kvir02}\,--\,\eqref{kvir03} follow immediately from
Def.~\ref{itsk1} and the left and right bi-gyroassociative law
\eqref{dunu21}\,--\,\eqref{dunu22}. Similarly, \eqref{kvir03d5}
follow immediately from Def.~\ref{itsk1} and the
bi-gyrocommutative law \eqref{bagt1}.
\end{proof} 

\begin{llemma}\label{itsk3}
The relation \eqref{kvir01} between gyrations $\gyrabp$ and
corresponding bi-gyrations $(\lgyrab,\rgyrba)$, $P_1,P_2\in(\Rnm,\opp)$,
is bijective.
\end{llemma}
\begin{proof}
Let $P_k\in\Rnm$, $k=1,2,3,4$. Assuming
\begin{equation} \label{kvir04}
(\lgyrab,\rgyrba)=(\lgyr[P_3,P_4],\rgyr[P_4,P_3])
\,,
\end{equation}
it clearly follows from \eqref{kvir01} that
\begin{equation} \label{kvir05}
\gyrabp = \gyr[P_3,P_4]
\,.
\end{equation}

Conversely, assuming \eqref{kvir05}, then
\begin{equation} \label{kvir06}
\gyrabp X = \gyr[P_3,P_4]X
\end{equation}
for all $X\in\Rnm$, so that by \eqref{kvir01}
\begin{equation} \label{kvir07}
\lgyrab X\rgyrba = \lgyr[P_3,P_4]X\rgyr[P_4,P_3]
\end{equation}
for all $X\in\Rnm$.

Noting that $\lgyr[P,Q]\in SO(n)$ and $\rgyr[P,Q]\in SO(m)$
for any $P,Q\in\Rnm$, \eqref{kvir04} follows from \eqref{kvir07}
and Lemma \ref{mdf2}, p.~\pageref{mdf2}, and the proof is complete.
\end{proof}

It is anticipated in Def.~\ref{itsk1} that gyrations are
automorphisms. The following theorem asserts that this is indeed
the case.

\begin{ttheorem}\label{itsk4}
{\bf (Gyroautomorphism).}
Gyrations $\gyrabp$ of a bi-gyrogroup $(\Rnm,\opp)$ are
automorphisms of the bi-gyrogroup.
\end{ttheorem}
\begin{proof}
It follows from the bi-gyration inversion law in
Theorem \ref{tmdkert}, p.~\pageref{tmdkert}, and from \eqref{kvir01}
that $\gyrabp$ is invertible,
\begin{equation} \label{kvir08}
\gyr^{-1}[P_1,P_2] = \gyr[P_2,P_1]
\end{equation}
for all $P_1,P_2\in\Rnm$.

Furthermore, noting that $\lgyrab\in SO(n)$ and $\rgyrab\in SO(m)$
it follows from \eqref{kvir01} and the third identity in
\eqref{adin105s}, p.~\pageref{adin105s}, that
\begin{equation} \label{kvir09}
\gyr[P_1,P_2](P\opp Q) = \gyr[P_1,P_2]P\opp\gyr[P_1,P_2]Q
\end{equation}
for all $P_1,P_2,P,Q\in\Rnm$.
Hence, by \eqref{kvir08} and \eqref{kvir09},
gyrations of $(\Rnm,\opp)$ are automorphisms of $(\Rnm,\opp)$,
and the proof is complete.
\end{proof}

\begin{ttheorem}\label{itsk5}
{\bf (Left Gyration Reduction Properties).}
Left gyrations of a bi-gyrogroup $(\Rnm,\opp)$ possess
the left gyration left reduction property
\begin{equation} \label{kvir10}
\lgyr[P_1,P_2] = \lgyr[P_1\opp P_2,P_2]
\end{equation}
and 
the left gyration right reduction property
\begin{equation} \label{kvir11}
\lgyr[P_1,P_2] = \lgyr[P_1,P_2\opp P_1]
\,.
\end{equation}
\end{ttheorem}
\begin{proof}
By \eqref{dmgdc1sa}, p.~\pageref{dmgdc1sa},
\eqref{bilva}, p.~\pageref{bilva}, with $O_m=\rgyr[P_2,P_1]$,
gyration inversion, and \eqref{hdken01}, p.~\pageref{hdken01},
we have the following chain of equations,
\begin{equation} \label{akuv03}
\begin{split}
\lgyr[P_1,P_2] &= \lgyr[P_1\op P_2,P_2\rgyr[P_1,P_2]]
\\[6pt] &=
\lgyr[(P_1\op P_2)\rgyr[P_2,P_1],P_2\rgyr[P_1,P_2]\rgyr[P_2,P_1]]
\\[6pt] &=
\lgyr[(P_1\op P_2)\rgyr[P_2,P_1],P_2]
\\[6pt] &=
\lgyr[P_1\opp P_2,P_2]
\,,
\end{split}
\end{equation}
thus proving \eqref{kvir10}.

By \eqref{dmgdc1sb}, p.~\pageref{dmgdc1sb},
\eqref{bilvb}, p.~\pageref{bilvb}, with $O_m=\rgyr[P_1,P_2]$,
gyration inversion, and \eqref{hdken01}, p.~\pageref{hdken01},
we have the following chain of equations,
\begin{equation} \label{akuv04}
\begin{split}
\lgyr[P_1,P_2] &= \lgyr[P_1\rgyr[P_2,P_1],P_2\op P_1]
\\[6pt] &=
\lgyr[P_1\rgyr[P_2,P_1]\rgyr[P_1,P_2],(P_2\op P_1)\rgyr[P_1,P_2]]
\\[6pt] &=
\lgyr[P_1,(P_2\op P_1)\rgyr[P_1,P_2]]
\\[6pt] &=
\lgyr[P_1,P_2\opp P_1]
\,,
\end{split}
\end{equation}
thus proving \eqref{kvir11}.
\end{proof}

\begin{ttheorem}\label{itsk5r}
{\bf (Right Gyration Reduction Properties).}
Right gyrations of a bi-gyrogroup $(\Rnm,\opp)$ possess
the right gyration left reduction property
\begin{equation} \label{piku01}
\rgyr[P_1,P_2] = \rgyr[P_1\opp P_2,P_2]
\end{equation}
and
the right gyration right reduction property
\begin{equation} \label{piku02}
\rgyr[P_1,P_2] = \rgyr[P_1,P_2\opp P_1]
\,.
\end{equation}
\end{ttheorem}
\begin{proof}
By \eqref{dmgdc1sd}, p.~\pageref{dmgdc1sd},
\eqref{bilvb}, p.~\pageref{bilvb}, with $O_n=\lgyr[P_1,P_2]$,
gyration inversion, and \eqref{hdken03t}, p.~\pageref{hdken03t},
we have the following chain of equations,
\begin{equation} \label{akuv01}
\begin{split}
\rgyr[P_1,P_2] &= \rgyr[P_2\op P_1,\lgyr[P_2,P_1]P_2]
\\[6pt] &=
\rgyr[\lgyr[P_1,P_2](P_2\op P_1),\lgyr[P_1,P_2]\lgyr[P_2,P_1]P_2]
\\[6pt] &=
\rgyr[\lgyr[P_1,P_2](P_2\op P_1),P_2]
\\[6pt] &=
\rgyr[P_1\opp P_2,P_2]
\,,
\end{split}
\end{equation}
thus proving \eqref{piku01}.

By \eqref{dmgdc1sc}, p.~\pageref{dmgdc1sc},
\eqref{bilvb}, p.~\pageref{bilvb}, with $O_n=\lgyr[P_2,P_1]$,
gyration inversion, and \eqref{hdken03t}, p.~\pageref{hdken03t},
we have the following chain of equations,
\begin{equation} \label{akuv02}
\begin{split}
\rgyr[P_1,P_2] &= \rgyr[\lgyr[P_1,P_2]P_1,P_1\op P_2]
\\[6pt] &=
\rgyr[\lgyr[P_2,P_1]\lgyr[P_1,P_2]P_1,\lgyr[P_2,P_1](P_1\op P_2)]
\\[6pt] &=
\rgyr[P_1,\lgyr[P_2,P_1](P_1\op P_2)]
\\[6pt] &=
\rgyr[P_1,P_2\opp P_1]
\,,
\end{split}
\end{equation}
thus proving \eqref{piku02}.
\end{proof}

\begin{ttheorem}\label{itskf}
{\bf (Gyration Reduction Properties).}
The gyrations of any bi-gyrogroup $(\Rnm,\opp)$, $m,n\in\Nb$, possess the
left and right reduction property
\begin{equation} \label{ptlm1}
\gyr[P_1,P_2] = \gyr[P_1\opp P_2,P_2]
\end{equation}
and
\begin{equation} \label{ptlm2}
\gyr[P_1,P_2] = \gyr[P_1,P_2\opp P_1]
\,.
\end{equation}
\end{ttheorem}
\begin{proof}
Identities \eqref{ptlm1} and \eqref{ptlm2} follow from Def.~\ref{itsk1}
of $\gyr$ in terms of $\lgyr$ and $\rgyr$, and from Theorems
\ref{itsk5} and \ref{itsk5r}.
\end{proof}

\section{Gyrogroups} \label{seding}

We are now in a position to present the definition of
the abstract gyrocommutative gyrogroup, and prove that any
bi-gyrogroup $(\Rnm,\opp)$ is a gyrocommutative gyrogroup.

Forming a natural generalization of groups, gyrogroups emerged
in the 1988 study of the  parametrization of the Lorentz group
of Einstein's special relativity theory \cite{parametrization,mybook01}.
Einstein velocity addition, thus, provides a concrete example of a 
gyrocommutative gyrogroup operation in the ball of all relativistically
admissible velocities.
\begin{ddefinition}\label{defroupx}
{\bf (Gyrogroups).}
{\it
A groupoid $(G , \op )$
is a gyrogroup if its binary operation satisfies the following
axioms (G1)\,--\,(G5).
In $G$ there is at least one element, $0$, called a left identity, satisfying

\noindent
(G1) \hspace{1.2cm} $0 \op a=a$

\noindent
for all $a \in G$. There is an element $0 \in G$ satisfying axiom $(G1)$ such
that for each $a\in G$ there is an element $\om a\in G$, called a
left inverse of $a$, satisfying

\noindent
(G2) \hspace{1.2cm} $\om a \op a=0\,.$

\noindent
Moreover, for any $a,b,c\in G$ there exists a unique element $\gyr[a,b]c \in G$
such that the binary operation obeys the left gyroassociative law

\noindent
(G3) \hspace{1.2cm} $a\op(b\op c)=(a\op b)\op\gyrab c\,.$

\noindent
The map $\gyr[a,b]:G\to G$ given by $c\mapsto \gyr[a,b]c$
is an automorphism of the groupoid $(G,\op)$, that is,

\noindent
(G4) \hspace{1.2cm} $\gyrab\in\Aut (G,\op) \,,$

\noindent
and the automorphism $\gyr[a,b]$ of $G$ is called
the gyroautomorphism, or the gyration, of $G$ generated by $a,b \in G$.
The operator $\gyr : G\times G\rightarrow\Aut (G,\op)$ is called the
gyrator of $G$.
Finally, the gyroautomorphism $\gyr[a,b]$ generated by any $a,b \in G$
possesses the left reduction property

\noindent
(G5) \hspace{1.2cm} $\gyrab=\gyr [a\op b,b] \,,$
\newline
called the reduction axiom.
}
\end{ddefinition}

The gyrogroup axioms ($G1$)\,--\,($G5$)
in Definition \ref{defroupx} are classified into three classes:
\begin{enumerate}
\item
The first pair of axioms, $(G1)$ and $(G2)$, is a reminiscent of the
group axioms.
\item
The last pair of axioms, $(G4)$ and $(G5)$, presents the gyrator
axioms.
\item
The middle axiom, $(G3)$, is a hybrid axiom linking the two pairs of
axioms in (1) and (2).
\end{enumerate}

As in group theory, we use the notation
$a \om b = a \op (\om b)$
in gyrogroup theory as well.

In full analogy with groups, gyrogroups are classified into gyrocommutative and
non-gyrocommutative gyrogroups.

\begin{ddefinition}\label{defgyrocomm}
{\bf (Gyrocommutative Gyrogroups).}
{\it
A gyrogroup $(G, \oplus )$ is gyrocommutative if
its binary operation obeys the gyrocommutative law

\noindent
(G6) \hspace{1.2cm} $a\oplus b=\gyrab(b\oplus a)$

\noindent
for all $a,b\in G$.
}
\end{ddefinition}

\begin{ttheorem}\label{itsk6}
{\bf (Gyrocommutative Gyrogroup).}
Any bi-gyrogroup $(\Rnm,\opp)$, $n,m\in\Nb$, is a
gyrocommutative gyrogroup.
\end{ttheorem}
\begin{proof}
We will validate each of the six gyrocommutative gyrogroup axioms
$(G1)$--$(G6)$ in Defs.~\ref{defroupx} and \ref{defgyrocomm}.
\begin{enumerate}
\item \label{nvd01}
The bi-gyrogroup $(\Rnm,\opp)$ possesses the left identity $0_{n,m}$,
thus validating Axiom $(G1)$.
\item \label{nvd02}
Every element $P\in\Rnm$ possesses the left inverse
$\omp P:=-P\in\Rnm$, thus validating Axiom $(G2)$.
\item \label{nvd03}
The binary operation $\opp$ obeys the left gyroassociative law
\eqref{kvir02} by Theorem \ref{itsk2}, thus validating Axiom $(G3)$.
\item \label{nvd04}
The map $\gyr[P_1,P_2]$ is an automorphism of
$(\Rnm,\opp)$ by Theorem \ref{itsk4}, that is,
$\gyr[P_1,P_2]\in\Aut(\Rnm,\opp)$, thus validating Axiom $(G4)$.
\item \label{nvd05}
The binary operation $\opp$ in $\Rnm$ possesses the left reduction
property \eqref{ptlm1} by Theorem \ref{itskf},
thus validating Axiom $(G5)$.
\item \label{nvd06}
The binary operation $\opp$ in $\Rnm$ possesses the
gyrocommutative law \eqref{kvir03d5} by Theorem \ref{itsk2},
thus validating Axiom $(G6)$.
\end{enumerate}
\end{proof}

\section{The Abstract Bi-gyrogroup} \label{safdi1}

Following the key features of the bi-gyrogroups $(\Rnm,\opp)$, the
abstract (bi-gyrocommutative) bi-gyrogroup is defined to be an
abstract (gyrocommutative) gyrogroup the gyrations of which are
bi-gyrations. In order to define bi-gyrations in the abstract context,
we introduce the concept of bi-automorphisms of a groupoid.

An automorphism of a groupoid $(S,+)$ is a bijective map $f$ of $S$ onto
itself that respects the groupoid binary operation, that is,
$f(s_1+s_2)=f(s_1)+f(s_2)$ for all $s_1,s_2\in S$.
An automorphism group, $\Auto(S,+)$, of $(S,+)$ is a group of
automorphisms of $(S,+)$ with group operation given by
automorphism composition.

Let $\Aut_L(S,+)$ and $\Aut_R(S,+)$ be two automorphism groups of
$(S,+)$, called a left and a right automorphism group of $(S,+)$,
such that
\begin{equation} \label{kerm01}
\Aut_L(S,+) \cap \Aut_R(S,+) = I
\,,
\end{equation}
$I$ being the identity automorphism of $(S,+)$.

Finally, let
\begin{equation} \label{kerm02}
\Auto(S,+) = \Aut_L(S,+) \times \Aut_R(S,+)
\end{equation}
be the direct product of $\Aut_L(S,+)$ and $\Aut_R(S,+)$.
\begin{enumerate}
\item\label{avye01}
The application of $f_L\in\Aut_L(S,+)$ to $s\in S$ is denoted by
$f_L(s)$ or $f_Ls$.
\item\label{avye02}
The application of $f_R\in\Aut_R(S,+)$ to $s\in S$ is denoted by
$(s)f_R$ or $sf_R$.
\item\label{avye03}
Accordingly, the application of $(f_L,f_R)\in\Auto(S,+)$ to $s\in S$
is denoted by
\begin{equation} \label{kerm03}
(f_L,f_R)s = f_Lsf_R
\end{equation}
where we assume that the composed map in \eqref{kerm03} is associative,
that is
\begin{equation} \label{kerm04}
(f_Ls)f_R = f_L(sf_R)
\,.
\end{equation}
Furthermore, we assume that the composed map in \eqref{kerm03} is unique,
that is,
\begin{equation} \label{kerm05}
f_{L,1}sf_{R,1} = f_{L,2}sf_{R,2}
~~ \Longrightarrow ~~
f_{L,1}=f_{L,2} \hspace{0.8cm} {\rm and} f_{R,1}=f_{R,2}
\end{equation}
for any $f_{L,k}\in\Aut_L(S,+)$, $f_{R,k}\in\Aut_R(S,+)$, $k=1,2$,
and $s\in S$.
\end{enumerate}

The automorphism group $\Auto(S,+)=\Aut_L(S,+)\times\Aut_R(S,+)$
is said to be a
{\it bi-automorphism group} of the groupoid $(S,+)$.

Let now the groupoid $(S,+)$ be a gyrogroup.
A gyroautomorphism group $\Auto(S,+)$ of $(S,+)$ is any
automorphism group of $(S,+)$ that contains the gyrations of $(S,+)$.
If $\Auto(S,+)$ is a bi-automorphism group of $(S,+)$ then its
direct product structure \eqref{kerm02} induces a direct product
structure for its subset of gyrations
\begin{equation} \label{kerm06}
\gyr[s_1,s_2] = (\lgyr[s_1,s_2],\rgyr[s_1,s_2])
\end{equation}
for all $s_1,s_2\in (S,+)$, where
\begin{equation} \label{kerm07}
\begin{split}
\gyr[s_1,s_2] &\in \Auto(S,+)
\\
\lgyr[s_1,s_2] &\in \Aut_L(S,+)
\\
\rgyr[s_1,s_2] &\in \Aut_G(S,+)
\,.
\end{split}
\end{equation}

The gyrations $\gyr[s_1,s_2]$ in \eqref{kerm06}
of the gyrogroup $(S,+)$ are said to be {\it bi-gyrations}. The
application of a bi-gyration $\gyr[s_1,s_2]$ to $s$ is denoted by
\begin{equation} \label{kerm08}
\gyr[s_1,s_2]s = (\lgyr[s_1,s_2],\rgyr[s_1,s_2])s
=\lgyr[s_1,s_2]s\rgyr[s_1,s_2]
\,.
\end{equation}

\begin{ddefinition}\label{defgyrocomy}
{\bf (Bi-gyrogroups).}
{\it
A (gyrocommutative) gyrogroup whose gyrations are bi-gyrations is
said to be a (bi-gyrocommutative) bi-gyrogroup.
}
\end{ddefinition}

A detailed study of the abstract bi-gyrogroup is presented in
\cite{sukung15}.

Remarkably, our study of special (or, unimodular) pseudo-orthogonal
groups $SO(m,n)$ can be extended straightforwardly to an analogous study
of special (or, unimodular) pseudo-unitary groups $SU(m,n)$.
Accordingly, bi-gyrocommutative bi-gyrogroup theory for $(\Rnm,\opp)$,
as developed in this article,
can be extended straightforwardly to $(\Cb^{n\times m},\opp)$ where
\begin{enumerate}
\item
real $n\times m$ matrices $P\in\Rnm$ are replaced by complex $n\times m$
matrices $P\in\Cb^{n\times m}$;
\item
the transpose $P^t$ of $P\in\Rnm$ is replaced by the
conjugate transpose $P^*=(\bar{P})^t$ of $P\in\Cb^{n\times m}$; and
\item
the special orthogonal matrices $O_k\in SO(k)$, $k=m,n$, are replaced by
special unitary matrices $U_k\in SU(k)$.
\end{enumerate}

{\bf Acknowledgments}
The author owes a huge debt of gratitude to Nikita Barabanov for his
generous collaboration.

\end{document}